\documentclass[reprint,superscriptaddress,nofootinbib,amsmath,amssymb,aps,prd]{revtex4-2}
\usepackage[utf8]{inputenc}
\usepackage{mathrsfs}
\usepackage{amsmath}
\usepackage{amsthm}
\usepackage[bookmarks=ture,
 breaklinks=true,pdfborder={0 0 1},colorlinks=true,linkcolor=blue,citecolor=blue,urlcolor=blue]
 {hyperref}
\makeatletter
\theoremstyle{plain}
\newtheorem{thm}{\protect\theoremname}
\theoremstyle{plain}
\newtheorem{lem}[thm]{\protect\lemmaname}

\usepackage{dcolumn}
\usepackage{xcolor}
\usepackage{subfigure}
\usepackage{bm}
\usepackage{makecell}
\usepackage{soul}
\usepackage{color}\usepackage{xcolor}
\usepackage{tikz}
\usepackage{braket}
\usetikzlibrary{arrows.meta}
\usetikzlibrary{decorations,arrows}
\usetikzlibrary{decorations.markings}
\DeclareSymbolFont{matha}{OML}{txmi}{m}{it}
\DeclareMathSymbol{v}{\mathord}{matha}{118}

\makeatother

\providecommand{\lemmaname}{Lemma}
\providecommand{\theoremname}{Theorem}

\begin{document}
\title{Dynamic and Thermodynamic Stability of Superconducting-superfluid
Stars}
\author{Delong Kong}
\email{kongdelong22@mails.ucas.ac.cn}

\affiliation{\textit{School of Physical Sciences, University of Chinese Academy
of Sciences, Beijing 100049, China}}
\author{Yu Tian}
\email{ytian@ucas.ac.cn}

\affiliation{\textit{School of Physical Sciences, University of Chinese Academy
of Sciences, Beijing 100049, China}}
\author{Hongbao Zhang}
\email{hongbaozhang@bnu.edu.cn}

\affiliation{School of Physics and Astronomy, Beijing Normal University, Beijing
100875, China}
\affiliation{Key Laboratory of Multiscale Spin Physics, Ministry of Education,
Beijing Normal University, Beijing 100875, China}

\date{\today}

\begin{abstract}
We present a comprehensive analysis of the dynamic and thermodynamic stability of neutron stars composed of superconducting-superfluid mixtures within the Iyer-Wald formalism. We derive the first law of thermodynamics and the necessary and sufficient condition under which dynamic equilibrium implies thermodynamic equilibrium. By constructing the phase space and canonical energy, we show that the dynamic stability for perturbations, restricted in symplectic complement of trivial perturbations with the ADM 3-momentum unchanged, is equivalent to the non-negativity of the canonical energy. Furthermore, dynamic stability against restricted axisymmetric perturbations guarantees the dynamic stability against all axisymmetric perturbations. We also prove that the positivity of canonical energy on all axisymmetric perturbations within the Lagrangian displacement framework with fixed angular momentum is necessary for thermodynamic stability. In particular, the equivalence of dynamic and thermodynamic stability is established for radial perturbations of static, spherically symmetric isentropic configurations.
\end{abstract}

\maketitle

\section{Introduction}\label{sec1}

The dynamic stability of self-gravitating compact objects is a cornerstone of theoretical astrophysics, determining whether they end up as a stable stars, collapsing to black holes, or exploding in a supernova. For the relativistic neutral or charged stars described by the single perfect fluid model, the criterion for dynamic stability has been established, and it has been found to be closely related to the thermodynamic stability \cite{Green:2013ica,Shi:2022jya}. Specifically, the criterion for dynamic stability of perfect fluid star in dynamic equilibrium is given by non-negativity of the canonical energy associated with the timelike Killing field, and the necessary condition for the thermodynamic stability of stars in thermodynamic equilibrium with respect to the axisymmetric perturbations is the positivity of the canonical energy. Furthermore, the dynamic and thermodynamic stability are equivalent if the background star is static, spherically symmetric isentropic and the perturbations are spherically symmetric.

While the single perfect fluid model provides an excellent approximation for many stellar objects such as red giants and white dwarfs, it is fundamentally inadequate for neutron stars due to their exceptionally complex internal microphysics. The interior of a neutron star is a multi-constituent system exhibiting distinct states across its layers. Broadly, its structure comprises a solid crystalline crust, a liquid outer core, and a potentially exotic inner core, with the equation of state and composition being active areas of research constrained by multi-messenger observations \cite{Langlois:2000gs,Ozel:2016oaf,LIGOScientific:2018cki,Annala:2019puf,Miller:2021qha}. Crucially for macroscopic dynamics, a defining feature of the relevant layers (the inner crust and the outer core) is the prevalence of superfluidity and superconductivity. In these regions, neutrons are believed to form a superfluid, while protons are in superconducting state, all permeated by a degenerate electron gas \cite{haensel2006neutron,Chamel:2008ca}. In the relevant layers, the superfluid neutrons make up the most of the mass density, the superconducting protons make up a small but significant part of the mass density, and the degenerate non-superconducting electrons make up a negligibly small fraction of the mass density, but nevertheless have an important role when the electromagnetic effects are concerned. This multi-component nature is not a minor correction but the central paradigm for explaining key observational phenomena. Therefore, a physically realistic hydrodynamic description of neutron star interiors must move beyond the single-fluid approximation. A multi-constituent fluid model, accounting for the distinct dynamics between the superfluid neutrons, the superconducting protons, and the electron plasma, is essential.

The two-fluid model of non-relativistic superfluid is developed by Laudau, and then generalized to the relativistic case and multi-constituent fluid model by Khalatnikov and Carter \cite{KHALATNIKOV198270,Carter:1987qr,Carter:1992gmy,CARTER1992243,Carter:1993aq}. For a basic representation of the neutron star's superconducting superfluid region, the two-fluid model is often adequate, where one constituent is the superfluid neutron while the other represents everything else, i.e., the approximately rigid background consists of protons and electrons that are tended by the short range electromagnetic interactions. Particularly, the two-constituent fluid model including allowance for ``transfusion'', meaning the slow transfer of baryonic matter (due to the process like beta decay and so on) between the neutron superfluid and the other ``normal'' fluid, has been constructed by Langlois et al. \cite{Langlois:1997bz}. However, because of its neglect of the electromagnetically interacting constituents, which will play an essential role in phenomena involving magnetic effects, the superconducting-superfluid mixtures \cite{Carter:1998rn} would be a more suitable framework.

The stability properties of non-relativistic rotating superfluid neutron stars have been studied long ago by using the canonical
energy formulation \cite{Andersson:2004nv}. And more recently, there have been some studies on the radial stability of neutron stars in relativistic case \cite{Caballero:2024qtv,Canullan-Pascual:2024jsu,Caballero:2025omv,Kumar:2025oyx}. For instance, reference \cite{Caballero:2024qtv} establishes a set of stability criteria for two perfect-fluid relativistic star, which is taken as a model of dark matter admixed neutron stars, by studying the radial mode perturbation equations, and provides also an alternative stability criterion (i.e., the positivity of canonical energy) in the same way as done for single perfect fluid stars. In \cite{Kumar:2025oyx}, the eigenvalue problem for a coupled system of equations with small-amplitude radial perturbations is solved and the critical line corresponding to stability boundaries is derived. Nevertheless, most of above researchers concern the model of two non-interacting perfect fluids, which are the simplification of the two-fluid model mentioned above where there are non-gravitational interactions between the two fluids in general, and the background star is taken to be spherically symmetric.

In this paper, we consider the relativistic neutron stars described by the non-transfusive superconducting-superfluid mixtures. By working with Iyer-Wald formalism \cite{Lee:1990nz,Wald:1993nt,Iyer:1994ys} rather than analyzing the perturbation equations, we will give the necessary and sufficient conditions for thermodynamic equilibrium, construct the phase space of our system, and then establish the criteria for both dynamic stability and thermodynamic stability as the positivity of canonical energy, which are equivalent if the background star is static, spherically symmetric isentropic and the perturbations are spherically symmetric, i.e., the radial perturbations. In our stability analysis, which includes the treatment of trivial displacements, we employ the most general form of Eulerian perturbations without any gauge fixing. This approach is necessitated by the presence of multiple Lagrangian displacements in a multi-constituent system and is more general than the single perfect fluid case, where a simplification to the Lagrangian perturbations, such as the gauge used in \cite{Green:2013ica,Shi:2022jya}, is possible.

Similar to the arguments given in the introduction of \cite{Green:2013ica}, when some perturbations acting on the star, if it is thermodynamically stable, the dissipative processes will make it back to the equilibrium, otherwise the dissipative processes will make it away from equilibrium. Our criterion for thermodynamic stability is established by examining the behavior of the entropy $S$, and it should be equivalent to the criterion by examining the dissipative dynamics as the dissipative processes tend to drive entropy to increase. So although our approach is incapable of yielding any information concerning the growth rate of any thermodynamic instability, it should lead to results equivalent to those that obtained by considering dissipation in principle. Furthermore, our results can be directly generalized to the case of superconducting-superfluid mixtures with one ``normal'' component and arbitrary number of ``super'' components, and even the case including the allowance for transfusion.

The rest of this paper is structured as follows. In Sec. \ref{sec2}, we review relativistic model of the superconducting-superfluid mixtures. In Sec. \ref{sec3}, we give a quick introduction to the Iyer-Wald formalism, and list some results that will be used in later. In Sec. \ref{sec4}, we derive the first law of thermodynamics, and after defining the thermodynamic equilibrium, the necessary and sufficient condition for the star in thermodynamic equilibrium is derived. Although the fluid model we used does not include the transfusion, but we will also give a simple argument about the transfusive case and show one of our definitions is suitable for such case at the end of this section. In Sec. \ref{sec5}, we devote ourselves to constructing the phase space and calculating the symplectic complement of trivial perturbations. As a by-product, we also derive the eigenvalue equation of propagation velocity $v$ of the sound wave, which will be useful in finding the degeneracy of pre-symplectic form. With such preparations, the criterion for the dynamic stability is established by introducing the canonical energy and taking advantage of the physically stationary perturbations in Sec. \ref{sec6}, and the necessary condition for the thermodynamic stability is established in Sec. \ref{sec7}. Finally, we give the conclusion and discussion in Sec. \ref{sec8}.

The notations and conventions of \cite{Wald:1984rg} will be followed by us, except the indices are not required to be balanced in our equations if no confusion arises. The capital Latin letters $\left(\text{X},\text{Y},...\right)$ denote the ``chemical'' indices which take the value $\text{n}$ for the neutrons, $\text{p}$ for protons, and $\text{e}$ for electrons, while $\Upsilon$ is used to be the chemical indices for ``super'' constituents, i.e., the neutrons and the protons. The early Latin letters $\left(a,b,c,...\right)$ and late Latin letters $\left(p,q,r,...\right)$ denote abstract spacetime indices, while the middle Latin letters $\left(i,j,k,...\right)$ denote concrete spatial indices on a spacelike Cauchy surface unless specified otherwise. Bold typeface will indicate the differential form indices on spacetime have been omitted, for instance, $\boldsymbol{N}$ denotes the tensor field $N_{abc}=N_{\left[abc\right]}$. We generally set $16\pi G = 1$ and $c=1$.

\section{Relativistic model of superconducting-superfluid mixtures}

\label{sec2}

In this section, we will review the Carter's relativistic model of superconducting-superfluid mixtures \cite{Carter:1998rn}. The three independent constituents under consideration are the superfluid neutrons with conserved particle current $n_{\text{n}}^{a}$, the superconducting protons with conserved particle current $n_{\text{p}}^{a}$, and the degenerate non-superconducting background of electrons with conserved particle current $n_{\text{e}}^{a}$. The electrons and protons have crucially important roles as electromagnetic effects are concerned, and in terms of the electron charge coupling constant $e$ the corresponding total electric current vector will be given by 
\begin{equation}
j^{a}=e\left(n_{\text{p}}^{a}-n_{\text{e}}^{a}\right).\label{electric current}
\end{equation}
Besides three principal constituents that have just been listed, there is a fourth constituent, namely the conserved entropy current $s^{a}=sn_{\text{e}}^{a}$.
Within the standard multi-constituent fluid model for superfluids, the entropy is associated exclusively with the normal (non-superfluid) component. Accordingly, as we consider that all neutrons and protons are in ``super" states in our model, the entropy current is only carried by the ``normal'' electron fluid with $s$ denoting the entropy per electron. For convenience, we shall use the capital Latin letters $\text{X}=\text{n},\text{p},\text{e}$ below as the ``chemical'' indices of three relevant constituents, and using this convention, the equation Eq. \eqref{electric current} can be rewritten in the concise form
\begin{equation}
j^{a}=\sum_{\text{X}}e^{\text{X}}n_{\text{X}}^{a},
\end{equation}
where the charges per neutron, proton, and electron are given respectively by $e^{\text{n}}=0$, $e^{\text{p}}=e$, and $e^{\text{e}}=-e$.

\subsection{Master function}

The central quantity of Carter's theory of multi-constituent fluid is the so-called \emph{master function} \cite{Carter:1987qr}, which is taken to be the total thermodynamic energy density $-\Lambda_{\text{M}}$, depending only on the metric $g_{ab}$ and the ``hydrodynamic'' part of the system, i.e., $n_{\text{n}}^{a}$, $n_{\text{p}}^{a}$, $n_{\text{e}}^{a}$ and $s$. More exactly, $\Lambda_{\text{M}}$ is a function of all scalar combinations obtained by their mutual contractions 
\begin{equation}
\Lambda_{\text{M}}=\Lambda_{\text{M}}\left(n_{\text{X}}^{2},x_{\text{XY}}^{2},s\right),
\end{equation}
where $n_{\text{X}}^{2}=-n_{\text{X}}^{a}n_{\text{X}}^{b}g_{ab}$, and $x_{\text{XY}}^{2}=-n_{\text{X}}^{a}n_{\text{Y}}^{b}g_{ab}$. The master function encodes all information about the local thermodynamic state of the fluid, and can also serve as a Lagrangian density in the absence of electromagnetic effects. The variation of $\Lambda_{M}$ gives that
\begin{equation}
\delta\Lambda_{\text{M}}=\sum_{\text{X}}\mu_{a}^{\text{X}}\delta n_{\text{X}}^{a}-Tn_{\text{e}}\delta s+\sum_{\text{X}}\frac{1}{2}n_{\text{X}}^{a}\mu^{\text{X}b}\delta g_{ab},\label{Lambda variation}
\end{equation}
where the effective momentum covectors respectively associated with the corresponding current $n_{\text{X}}^{a}$ are given by
\begin{equation}
\mu_{a}^{\text{X}}=-2\frac{\partial\Lambda_{\text{M}}}{\partial n_{\text{X}}^{2}}n_{\text{X}a}-\sum_{\text{Y}\neq\text{X}}\frac{\partial\Lambda_{\text{M}}}{\partial x_{\text{XY}}^{2}}n_{\text{Y}a},\label{effective momentum}
\end{equation}
and the temperature is given by
\begin{equation}
T=-\frac{1}{n_{\text{e}}}\frac{\partial\Lambda_{\text{M}}}{\partial s},
\end{equation}
Denote
\begin{align}
\mathscr{A} & =-\frac{\partial\Lambda_{\text{M}}}{\partial x_{\text{np}}^{2}},\quad\mathscr{B}=-\frac{\partial\Lambda_{\text{M}}}{\partial x_{\text{ne}}^{2}},\quad\mathscr{C}=-\frac{\partial\Lambda_{\text{M}}}{\partial x_{\text{pe}}^{2}},\nonumber \\
\mathscr{D} & =-2\frac{\partial\Lambda_{\text{M}}}{\partial n_{\text{n}}^{2}},\quad\mathscr{E}=-2\frac{\partial\Lambda_{\text{M}}}{\partial n_{\text{p}}^{2}},\quad\mathscr{F}=-2\frac{\partial\Lambda_{\text{M}}}{\partial n_{\text{e}}^{2}},
\end{align}
the effective momentum covectors Eq. \eqref{effective momentum} can be rewritten in the form 
\begin{equation}
\mu_{a}^{\text{X}}=\sum_{\text{Y}}\mathbb{I}_{\text{XY}}n_{\text{Y}a},
\end{equation}
where 
\begin{equation}
\mathbb{I}=\left(\begin{array}{ccc}
\mathscr{D} & \mathscr{A} & \mathscr{B}\\
\mathscr{A} & \mathscr{E} & \mathscr{C}\\
\mathscr{B} & \mathscr{C} & \mathscr{F}
\end{array}\right),\label{inertia matrix}
\end{equation}
is the \emph{inertia matrix} and we will assume that it is positive definite like it in the two fluid model \cite{CARTER1992243}.

\subsection{Dynamical fields, variations, and Lagrangian displacements}

\label{subsec2.2}

To develop a Lagrangian description of the superconducting-superfluid mixtures, not only are we required to have the spacetime manifold $\mathcal{M}$, on which the metric $g_{ab}$ and the electromagnetic potential $A_{a}$ are defined, but also for each constituent we should introduce a fiducial manifold $\mathcal{M}_{\text{X}}^{\prime}$, called fluid spacetime, which is diffeomorphic to $\mathcal{M}$. Then with a fixed scalar field $s^{\prime}$ on $\mathcal{M}_{\text{e}}^{\prime}$, and fixed closed 3-form $\boldsymbol{N}^{\text{X}\prime}$ on $\mathcal{M}_{\text{X}}^{\prime}$, one can define the physical fluid fields on $\mathcal{M}$ by pushing forward with diffeomorphism $\chi_{\text{X}}$ as
\begin{equation}
\iota_{n_{\text{X}}}\boldsymbol{\epsilon}\equiv\boldsymbol{N}^{\text{X}}=\chi_{\text{X}*}\boldsymbol{N}^{\text{X}\prime},\quad s=\chi_{\text{e}*}s^{\prime},
\end{equation}
where $\boldsymbol{\epsilon}$ is the associated spacetime volume element. So we can take $\phi=\left(g_{ab},A_{a},\chi_{\text{n}},\chi_{\text{p}},\chi_{\text{e}}\right)$ as the dynamical fields, for convenience, we shall write $\phi=\left(g_{ab},A_{a},\chi_{\text{X}}\right)$.

The variations about an arbitrary field configuration $\phi$ can be formulated by introducing a one-parameter family of dynamical fields $\phi\left(\lambda\right)=\left(g_{ab}\left(\lambda\right),A_{a}\left(\lambda\right),\chi_{\text{X}}(\lambda)\right)$ with $\phi\left(0\right)=\phi$. Since for each constituent, $\rho_{\text{X}}\left(\lambda\right)=\chi_{\text{X}}(\lambda)\circ\chi_{\text{X}}^{-1}$ give rise to one-parameter family of diffeomorphism on $\mathcal{M}$ generated to first order by a vector field $\xi_{\text{X}}^{a}$ known as \emph{Lagrangian displacement}, hence the first order perturbation is completely specified by $\delta\phi=\left(\delta g_{ab},\delta A_{a},\xi_{\text{X}}^{a}\right)$. The first order variations of $\boldsymbol{N}^{\text{X}}$ and $s$ are given by
\begin{equation}
\delta\boldsymbol{N}^{\text{X}}=-\mathscr{L}_{\xi_{\text{X}}}\boldsymbol{N}^{\text{X}},\quad\delta s=-\mathscr{L}_{\xi_{\text{e}}}s.\label{N=000026s variation}
\end{equation}
A first order perturbation is said to be \emph{trivial} if $\delta g_{ab}=0$, $\delta A_{a}=0$, $\mathscr{L}_{\xi_{\text{X}}}\boldsymbol{N}^{\text{X}}=0$,
and $\mathscr{L}_{\xi_{\text{e}}}s=0$, i.e., if all of the physical variables are unchanged by the perturbation. The associated displacement $\xi_{\text{X}}^{a}$ is called \emph{trivial displacement}.

As a consequence of the variations Eq. \eqref{N=000026s variation}, the perturbations of the particle current $n_{\text{X}}^{a}$, the relevant unit flow $u_{\text{X}}^{a}=\frac{1}{n_{\text{X}}}n_{\text{X}}^{a}$, and the particle density $n_{\text{X}}$ in its own rest frame are given by
\begin{align}
\delta n_{\text{X}}^{a} & =-\mathscr{L}_{\xi_{\text{X}}}n_{\text{X}}^{a}-n_{\text{X}}^{a}\nabla_{b}\xi_{\text{X}}^{b}-\frac{1}{2}n_{\text{X}}^{a}g^{bc}\delta g_{bc},\label{n variation}\\
\delta u_{\text{X}}^{a} & =\frac{1}{2}u_{\text{X}}^{a}u_{\text{X}}^{b}u_{\text{X}}^{c}\delta g_{bc}-q_{\text{X}b}^{a}\mathscr{L}_{\xi_{\text{X}}}u_{\text{X}}^{b},\label{u variation}\\
\delta n_{\text{X}} & =-\mathscr{L}_{\xi_{\text{X}}}n_{\text{X}}-n_{\text{X}}q_{\text{X}b}^{a}\nabla_{a}\xi_{\text{X}}^{b}-\frac{1}{2}nq_{\text{X}}^{ab}\delta g_{ab},
\end{align}
where $q_{\text{X}}^{ab}=g^{ab}+u_{\text{X}}^{a}u_{\text{X}}^{b}$.

\subsection{Lagrangian of superconducting-superfluid mixtures}

The standard minimal prescription for inclusion of electromagnetic interactions is to use a combined Lagrangian scalar density in which the ``matter'' contribution $\Lambda_{\text{M}}$ is augmented by an electromagnetic field contribution $\Lambda_{F}=-\frac{1}{4}F_{ab}F^{ab}$ and a gauge dependent coupling term of the usual form to give a total Lagrangian $\boldsymbol{\mathcal{L}}$ expressible as
\begin{equation}
\boldsymbol{\mathcal{L}}=\boldsymbol{\epsilon}\left(\Lambda_{\text{M}}-\frac{1}{4}F_{ab}F^{ab}+j^{a}A_{a}\right),
\end{equation}
where $F_{ab}=2\nabla_{[a}A_{b]}$. With Eqs. \eqref{Lambda variation}, \eqref{N=000026s variation}, and \eqref{n variation}, one finds that the variation of Lagrangian is given by
\begin{align}
\delta\boldsymbol{\mathcal{L}}= & -\boldsymbol{\epsilon}\sum_{\text{X}}\left(n_{\text{X}}^{b}w_{ba}^{\text{X}}-\pi_{a}^{\text{X}}\nabla_{b}n_{\text{X}}^{b}\right)\xi_{\text{X}}^{a}+\boldsymbol{\epsilon}Tn_{\text{e}}\nabla_{a}s\xi_{\text{e}}^{a}\nonumber \\
 & +\boldsymbol{\epsilon}\left(j^{a}-\nabla_{b}F^{ab}\right)\delta A_{a}+\frac{1}{2}\boldsymbol{\epsilon}T^{ab}\delta g_{ab}\nonumber \\
 & +\boldsymbol{\epsilon}\nabla_{a}\left(-F^{ab}\delta A_{b}+\sum_{\text{X}}\pi_{b}^{\text{X}}n_{\text{X}}^{[a}\xi_{\text{X}}^{b]}\right),
\end{align}
where the gauge dependent total momentum covectors are given by
\begin{equation}
\pi_{a}^{\text{X}}=\mu_{a}^{\text{X}}+e^{\text{X}}A_{a},
\end{equation}
the vorticity tensors are given by
\begin{equation}
w_{ab}^{\text{X}}=2\nabla_{[a}\pi_{b]}^{\text{X}}=2\nabla_{[a}\mu_{b]}^{\text{X}}+e^{\text{X}}F_{ab},
\end{equation}
and the energy-momentum tensor is given by\footnote{It is easy to check that $\sum_{\text{X}}n_{\text{X}}^{a}\mu^{\text{X}b}=\sum_{\text{X}}n_{\text{X}}^{b}\mu^{\text{X}a}$ by using Eq. \eqref{effective momentum}.}
\begin{align}
T^{ab} & =T_{F}^{ab}+T_{\text{M}}^{ab},\\
T_{F}^{ab} & =F^{ca}F_{c}^{\ b}-\frac{1}{4}F_{cd}F^{cd}g^{ab},\\
T_{\text{M}}^{ab} & =\sum_{\text{X}}n_{\text{X}}^{a}\mu^{\text{X}b}+\Psi_{\text{M}}g^{ab},\label{TM}
\end{align}
with the generalization of pressure \cite{Carter:1987qr}
\begin{equation}
\Psi_{\text{M}}=\Lambda_{\text{M}}-\sum_{\text{X}}\mu_{a}^{\text{X}}n_{\text{X}}^{a}.\label{Pressure}
\end{equation}
Combining Eqs. \eqref{Lambda variation} and \eqref{Pressure}, we find that the variation of $\Psi_{\text{M}}$ is given by
\begin{equation}
\delta\Psi_{\text{M}}=-\sum_{\text{X}}n_{\text{X}}^{a}\delta\mu_{a}^{\text{X}}-Tn_{\text{e}}\delta s+\sum_{\text{X}}\frac{1}{2}n_{\text{X}}^{a}\mu^{\text{X}b}\delta g_{ab}.\label{variation of pressure}
\end{equation}

Taking account the separate conservation of particle currents of each constituent and the conservation of entropy current, 
\begin{equation}
\nabla_{a}n_{\text{X}}^{a}=0,\quad u_{\text{e}}^{a}\nabla_{a}s=0,\label{conservation of n=000026s}
\end{equation}
we find the equations of motion of neutrons, protons, and electrons are respectively given by
\begin{align}
f_{\text{n}}^{a} & =-n_{\text{n}}^{b}w_{ba}^{\text{n}}=0,\label{eomn}\\
f_{\text{p}}^{a} & =-n_{\text{p}}^{b}w_{ba}^{\text{p}}=0,\label{eomp}\\
f_{\text{e}}^{a} & =-n_{\text{e}}^{b}w_{ba}^{\text{e}}+Tn_{\text{e}}\nabla_{\mu}s=0,\label{eome}
\end{align}
and the equation of motion of electromagnetic field is given by
\begin{equation}
E_{F}^{a}=j^{a}-\nabla_{b}F^{ab}=0.
\end{equation}

\section{Lagrangian framework for diffeomorphism covariant theories}

\label{sec3}

In this section, we will review the Iyer-Wald formalism for diffeomorphism covariant theories. We will apply these results to the vacuum Einstein-Maxwell Lagrangian in Sec. \ref{sec4}, and to the Einstein-superconducting-superfluid Lagrangian in Sec. \ref{sec5}.

Consider the variation of the diffeomorphism covariant Lagrangian
\begin{equation}
\delta\boldsymbol{\mathcal{L}}=\boldsymbol{E}\cdot\delta\phi+d\boldsymbol{\theta}\left(\phi;\delta\phi\right),\label{variation of Lagrangian}
\end{equation}
where $\boldsymbol{E}$ is the equations of motion, and $\boldsymbol{\theta}(\phi,\delta\phi)$ is the symplectic potential. Now with $\delta\phi$ formally viewed
as a vector in the tangent space at $\phi$ of the space of field configuration $\mathcal{F}$, denoted as $\delta\phi^{A}$, we obtain a linear map at each $\phi\in\mathcal{F}$ from vectors, $\delta\phi^{A}$, into numbers by integration of the 3-form $\boldsymbol{\theta}\left(\phi;\delta\phi\right)$ over a Cauchy surface $\Sigma$. We can interpret this linear map as defining a 1-form field $\Theta_{A}$ on $\mathcal{F}$ by
\begin{equation}
\Theta_{A}\delta\phi^{A}=\int_{\Sigma}\boldsymbol{\theta}\left(\phi;\delta\phi\right),
\end{equation}
and the pre-symplectic form (rather than symplectic form since it has degeneracy as argued below) is defined by 
\begin{equation}
W_{AB}=\left(\mathcal{D}\Theta\right)_{AB},
\end{equation}
where $\mathcal{D}$ represents the exterior derivative on forms on $\mathcal{F}$. For 1-form $\Theta_{A}$, one has\footnote{A standard result of differential geometry states that for 1-form $\boldsymbol{\mathcal{Y}}$ and vectors $\mathcal{X}_1,\mathcal{X}_2$,
\begin{equation}
d\boldsymbol{\mathcal{Y}}\left(\mathcal{X}_1,\mathcal{X}_2\right)=\mathcal{X}_1\left(\boldsymbol{\mathcal{Y}}\left(\mathcal{X}_2\right)\right)-\mathcal{X}_2\left(\boldsymbol{\mathcal{Y}}\left(\mathcal{X}_1\right)\right)-\boldsymbol{\mathcal{Y}}\left(\left[\mathcal{X}_1,\mathcal{X}_2\right]\right).
\end{equation}}
\begin{align}
\left(\mathcal{D}\Theta\right)_{AB}\delta_{1}\phi^{A}\delta_{2}\phi^{B}= & \mathbb{L}_{\delta_{1}\phi}\left(\Theta_{B}\delta_{2}\phi^{B}\right)-\mathcal{\mathbb{L}}_{\delta_{2}\phi}\left(\Theta_{A}\delta_{1}\phi^{A}\right)\nonumber \\
 & -\Theta_{A}\left[\delta_{1}\phi,\delta_{2}\phi\right]^{A},\label{presymplectic}
\end{align}
where $\mathbb{L}$ denotes the Lie derivative on $\mathcal{F}$. Since with the covariant derivative $D_{A}$ on $\mathcal{F}$, we can formally write that the variation induced by the field variations
$\delta_{1}\phi$ as
\begin{equation}
\delta_{1}\left(\Theta_{A}\delta_{2}\phi^{A}\right)=\delta_{1}\phi^{A}D_{A}\left(\Theta_{A}\delta_{2}\phi^{A}\right)=\mathbb{L}_{\delta_{1}\phi}\left(\Theta_{A}\delta_{2}\phi^{A}\right),
\end{equation}
and note that
\begin{equation}
\delta_{1}\left(\Theta_{A}\delta_{2}\phi^{A}\right)=\int_{\Sigma}\delta_{1}\boldsymbol{\theta}\left(\phi;\delta_{2}\phi\right),
\end{equation}
then Eq. \eqref{presymplectic} amounts to saying
\begin{equation}
W_{AB}\delta_{1}\phi^{A}\delta_{2}\phi^{B}=\int_{\Sigma}\boldsymbol{\omega}\left(\phi;\delta_{1}\phi,\delta_{2}\phi\right),\label{pre-symplectic form}
\end{equation}
where the pre-symplectic current 3-form $\boldsymbol{\omega}$ on spacetime is defined by
\begin{align}
 & \boldsymbol{\omega}\left(\phi;\delta_{1}\phi,\delta_{2}\phi\right)\nonumber \\
= & \delta_{1}\boldsymbol{\theta}\left(\phi;\delta_{2}\phi\right)-\delta_{2}\boldsymbol{\theta}\left(\phi;\delta_{1}\phi\right)-\boldsymbol{\theta}\left(\phi;\delta_{1}\delta_{2}\phi-\delta_{2}\delta_{1}\phi\right),\label{pre-symplectic current}
\end{align}
and $\delta_{1}$ and $\delta_{2}$ denote the variation of quantities induced by the field variations $\delta_{1}\phi$ and $\delta_{2}\phi$ respectively. It immediately that 
\begin{align}
 &d\boldsymbol{\omega}\left(\phi;\delta_{1}\phi,\delta_{2}\phi\right)\nonumber\\
=&\delta_{1}d\boldsymbol{\theta}\left(\phi;\delta_{2}\phi\right)-\delta_{2}d\boldsymbol{\theta}\left(\phi;\delta_{1}\phi\right)-d\boldsymbol{\theta}\left(\phi;\delta_{1}\delta_{2}\phi-\delta_{2}\delta_{1}\phi\right)\nonumber\\
=&\delta_1\left(\delta_2\boldsymbol{\mathcal{L}}-\boldsymbol{E}\cdot\delta_2\phi\right)-\delta_2\left(\delta_1\boldsymbol{\mathcal{L}}-\boldsymbol{E}\cdot\delta_1\phi\right)\nonumber\\
 &-\left[\left(\delta_1\delta_2-\delta_2\delta_1\right)\boldsymbol{\mathcal{L}}-\boldsymbol{E}\cdot\left(\delta_1\delta_2-\delta_2\delta_1\right)\phi\right]\nonumber\\
=&\delta_{2}\boldsymbol{E}\cdot\delta_{1}\phi-\delta_{1}\boldsymbol{E}\cdot\delta_{2}\phi,
\end{align}
so $\boldsymbol{\omega}$ is closed whenever $\delta_{1}\phi$ and $\delta_{2}\phi$ satisfy the linearized equations of motion $\delta_{1}\boldsymbol{E}=\delta_{2}\boldsymbol{E}=0$. Consequently, if the linearized equations of motion hold, then $W_{AB}\delta\phi^{A}\delta\phi^{B}$ is conserved in the sense that it takes the same value if the integral defining this quantity is performed over the surface $\Sigma^{\prime}$ rather than $\Sigma$, where $\Sigma^{\prime}$ and $\Sigma$ bound a compact region. For asymptotically flat spacetime, $W_{AB}\delta\phi^{A}\delta\phi^{B}$ takes the same value on any two asymptotically flat Cauchy surfaces $\Sigma$ and $\Sigma^{\prime}$ provided that $\delta_{1}\phi$ and $\delta_{2}\phi$ satisfy the linearized equations of motion and have suitable fall-off at infinity.

For a diffeomorphism covariant Lagrangian, the Noether current 3-form on spacetime associated with an arbitrary vector field $\mathcal{X}^{a}$ is defined by
\begin{equation}
\boldsymbol{\mathcal{J}}_{\mathcal{X}}=\boldsymbol{\theta}\left(\phi;\mathscr{L}_{\mathcal{X}}\phi\right)-\iota_{\mathcal{X}}\boldsymbol{\mathcal{L}}.\label{Noether current}
\end{equation}
A simple calculation \cite{Iyer:1994ys} shows that the first variation of $\boldsymbol{\mathcal{J}}_{\mathcal{X}}$ (with $\mathcal{X}^{a}$ fixed, i.e., unvaried) satisfies 
\begin{equation}
\delta\boldsymbol{\mathcal{J}}_{\mathcal{X}}=-\iota_{\mathcal{X}}\left(\boldsymbol{E}\cdot\delta\phi\right)+\boldsymbol{\omega}\left(\phi;\delta\phi,\mathscr{L}_{\mathcal{X}}\phi\right)+d\left[\iota_{\mathcal{X}}\boldsymbol{\theta}\left(\phi;\delta\phi\right)\right],
\end{equation}
where it has not been assumed that $\phi$ satisfies the field equations nor that $\delta\phi$ satisfies the linearized field equations. Furthermore, it can be shown that $\boldsymbol{\mathcal{J}}_{\mathcal{X}}$ can be written in the form \cite{Iyer:1995kg}
\begin{equation}
\boldsymbol{\mathcal{J}}_{\mathcal{X}}=\boldsymbol{C}_{\mathcal{X}}+d\boldsymbol{Q}_{\mathcal{X}},
\end{equation}
where $\boldsymbol{Q}_{\mathcal{X}}$ is the Noether charge and $\boldsymbol{C}_{\mathcal{X}}\equiv \mathcal{X}^{a}\boldsymbol{C}_{a}$ with $\boldsymbol{C}_{a}=0$ being the constraint equations of the theory \cite{Seifert:2006kv}. Having written in this form, we obtain the fundamental identity
\begin{equation}
\boldsymbol{\omega}\left(\phi;\delta\phi,\mathscr{L}_{\mathcal{X}}\phi\right)=\iota_{\mathcal{X}}\boldsymbol{E}\cdot\delta\phi+\delta\boldsymbol{C}_{\mathcal{X}}+d\left[\delta\boldsymbol{Q}_{\mathcal{X}}-\iota_{\mathcal{X}}\boldsymbol{\theta}\left(\phi;\delta\phi\right)\right],\label{fundamental identity}
\end{equation}
It should be emphasized that this fundamental identity holds for arbitrary $\mathcal{X}^{a}$, $\phi$, and $\delta\phi$.

One immediate consequence of Eq. \eqref{fundamental identity} is the gauge invariance of the symplectic form. If $\phi$ satisfies the equations of motion, $\boldsymbol{E}=0$, $\delta\phi$ satisfies the linearized constraints, $\delta\boldsymbol{C}_{a}=0$, and $\mathcal{X}^{a}$ is of compact support (or vanishes sufficiently rapidly at infinity and/or any boundaries), integration of Eq. \eqref{fundamental identity} over a Cauchy surface $\Sigma$ yields
\begin{equation}
W_{AB}\delta\phi^{A}\mathscr{L}_{\mathcal{X}}\phi^{B}=0.
\end{equation}
Consequently, the value of $W_{AB}\delta\phi^{A}\delta\phi^{B}$ is unchanged if either $\delta_{1}\phi$ or $\delta_{2}\phi$ is altered by a gauge transformation $\delta\phi\rightarrow\delta\phi+\mathscr{L}_{\mathcal{X}}\phi$ with $\mathcal{X}^{a}$ of compact support.

Another very important application concerns the case where $\mathcal{X}^{a}$ approaches a nontrivial asymptotic symmetry rather than being of compact support, in which case we can derive a formula for the Hamiltonian, $H_{\mathcal{X}}$, conjugate to the notion of ``translations'' defined by $\mathcal{X}^{a}$, and, thereby, a definition of ADM-type conserved quantities. Consider asymptotically flat spacetime with one asymptotically flat ``end''. Integrating Eq. \eqref{fundamental identity} over a Cauchy surface $\Sigma$, we have
\begin{align}
 & W_{AB}\delta\phi^{A}\mathcal{L}_{\mathcal{X}}\phi^{B}\nonumber \\
= & \int_{\Sigma}\left(\iota_{\mathcal{X}}\boldsymbol{E}\cdot\delta\phi+\delta\boldsymbol{C}_{\mathcal{X}}\right)+\int_{S_{\infty}}\left[\delta\boldsymbol{Q}_{\mathcal{X}}-\iota_{\mathcal{X}}\boldsymbol{\theta}\left(\phi;\delta\phi\right)\right],\label{symplectic form}
\end{align}
where the second integral is taken over a 2-sphere $S$ that limits to infinity (additional boundary terms would appear if $\Sigma$ terminated at a bifurcate Killing horizon or if here were additional asymptotically flat ends). Suppose in this limit, we have
\begin{equation}
\lim_{S\rightarrow S_{\infty}}\int_{S}\iota_{\mathcal{X}}\boldsymbol{\theta}\left(\phi;\delta\phi\right)=\lim_{S\rightarrow S_{\infty}}\delta\int_{S}\iota_{\mathcal{X}}\boldsymbol{B},\label{limit}
\end{equation}
for some 3-form $\boldsymbol{B}$ constructed from $\phi$ and the background asymptotic structure near infinity. Then, if $\phi$ satisfies the equations of motion, $\boldsymbol{E}=0-$but $\delta\phi$ is not required to satisfy the linearized equations of motion$-$we have
\begin{equation}
W_{AB}\delta\phi^{A}\mathscr{L}_{\mathcal{X}}\phi^{B}=\delta H_{\mathcal{X}},\label{delta Hamiltonian}
\end{equation}
where
\begin{equation}
H_{\mathcal{X}}=\int_{\Sigma}\boldsymbol{C}_{\mathcal{X}}+\int_{S_{\infty}}\left(\boldsymbol{Q}_{\mathcal{X}}-\iota_{\mathcal{X}}\boldsymbol{B}\right).\label{Hamiltonian}
\end{equation}

Writing 
\begin{equation}
\delta H_{\mathcal{X}}=\delta\phi^{A}D_{A}H_{\mathcal{X}}=\left(\mathcal{D}H_{\mathcal{X}}\right)_{A}\delta\phi^{A},
\end{equation}
we may rewrite Eq. \eqref{delta Hamiltonian} as
\begin{equation}
W_{AB}\mathscr{L}_{\mathcal{X}}\phi^{B}=\left(\mathcal{D}H_{\mathcal{X}}\right)_{A}.
\end{equation}
We now pass from the field configuration space, $\mathcal{F}$, to phase space, $\mathcal{P}$, by factoring by the degeneracy orbits of $W_{AB}$. On $\mathcal{P}$, $W_{AB}$ is well defined and, by construction, is nondegenerate. Let $W^{AB}$ denote the inverse of $W_{AB}$, so that $W^{AB}W_{BC}=\delta_{C}^{A}$ which denotes the identity map on $\mathcal{P}$. Then we have
\begin{equation}
\left(\mathscr{L}_{\mathcal{X}}\phi\right)^{A}=W^{AB}\left(\mathcal{D}H_{\mathcal{X}}\right)_{B},
\end{equation}
which is the usual form of Hamilton's equations of motion on a symplectic manifold. Thus if both the asymptotic conditions on $\phi$ and the asymptotic behavior of $\mathcal{X}^{a}$ are such that a 3-form $\boldsymbol{B}$ satisfying Eq. \eqref{limit} exists, then Eq. \eqref{Hamiltonian} yields a Hamiltonian conjugate to the notion of ``translations'' defined by $\mathcal{X}^{a}$. Note that when evaluated on solutions, $\boldsymbol{C}_{\mathcal{X}}=0$, $H_{\mathcal{X}}$ is purely a ``surface term''
\begin{equation}
H_{\mathcal{X}}\vert_{\boldsymbol{E}=0}=\int_{S_{\infty}}\left(\boldsymbol{Q}_{\mathcal{X}}-\iota_{\mathcal{X}}\boldsymbol{B}\right).\label{surface charge}
\end{equation}

In the case where $\mathcal{X}^{a}$ is asymptotic to a time translation $t^{a}$ at infinity, and $\boldsymbol{B}$ satisfying Eq. \eqref{limit} can be found, then Eq. \eqref{surface charge} defines the ADM mass
\begin{equation}
M=\int_{S_{\infty}}\left(\boldsymbol{Q}_{t}-\iota_{t}\boldsymbol{B}\right).\label{ADM mass}
\end{equation}
In the case where $\mathcal{X}^{a}$ is asymptotic to a rotation $\varphi^{a}$ tangent to $\Sigma$ at infinity and $S$ is chosen so that $\mathcal{X}^{a}$ is tangent to $S$, the pull back of $\iota_{\mathcal{X}}\boldsymbol{\theta}$ to $S$ vanishes, then Eq. \eqref{surface charge} with $\mathcal{X}^{a}=\varphi^{a}$ and $\boldsymbol{B}=0$ defines minus the ADM angular momentum
\begin{equation}
J=-\int_{S_{\infty}}\boldsymbol{Q}_{\varphi}.\label{ADM angular momentum}
\end{equation}

Finally, let us return to Eq. \eqref{symplectic form} in the case where $\phi$ has a time translation symmetry, i.e. $\mathcal{L}_{t}\phi=0$ for a vector field $t^{a}$ that approaches a time translation at infinity. We further assume that the equations of motion, $\boldsymbol{E}=0$, hold in a neighborhood of infinity, but we do not assume that they hold in the interior of the spacetime. We similarly assume that $\delta\phi$ satisfies the linearized constraints near infinity, but do not assume that these hold in the interior of the spacetime, nor do we make any symmetry assumptions on $\delta\phi$. Then the left side of Eq. \eqref{symplectic form} vanishes if $\mathcal{X}^{a}=t^{a}$, and the surface integral on the right side simply yields $\delta M$. Thus we obtain
\begin{equation}
\delta M=-\int_{\Sigma}\left(\iota_{t}\boldsymbol{E}\cdot\delta\phi+\delta\boldsymbol{C}_{t}\right).\label{delta M}
\end{equation}

\section{First law of thermodynamics and thermodynamic equilibrium}

\label{sec4}

By applying the results of the previous section to the vacuum Einstein-Maxwell Lagrangian, we will derive the first law of thermodynamics for our superconducting-superfluid star. We will also show that on a $t-\varphi$ reflection invariant Cauchy surface $\Sigma$ of the background spacetime, a solution to the linearized Einstein-Maxwell constraint equations always can be found for any given axisymmetric specifications of variation of the thermodynamic quantities. Accordingly, we finally give two kinds of definition of thermodynamic equilibrium, and find the necessary and sufficient conditions for a superconducting-superfluid star in dynamic equilibrium to be in both kinds of thermodynamic equilibrium.

\subsection{First law of thermodynamics}

Consider the vacuum Einstein-Maxwell Lagrangian
\begin{equation}
\boldsymbol{\mathcal{L}}=\boldsymbol{\epsilon}\left(R-\frac{1}{4}F_{ab}F^{ab}\right).
\end{equation}
For this Lagrangian, the equations of motion Eq. \eqref{variation of Lagrangian} are given by
\begin{align}
\boldsymbol{E}_{G}^{ab} & =\boldsymbol{\epsilon}\left(-G^{ab}+\frac{1}{2}T_{F}^{ab}\right),\\
\boldsymbol{E}_{F}^{a} & =-\boldsymbol{\epsilon}\nabla_{b}F^{ab},
\end{align}
the constraint 3-form $\boldsymbol{C}_{\mathcal{X}}$ is
\begin{equation}
\left(\boldsymbol{C}_{\mathcal{X}}\right)_{abc}=-\mathcal{X}^{d}\boldsymbol{\epsilon}_{eabc}\left[2\left(E_{G}\right)_{\ d}^{e}+E_{F}^{e}A_{d}\right],
\end{equation}
and the Noether charge 2-form $\boldsymbol{Q}_{\mathcal{X}}$ is
\begin{equation}
\left(\boldsymbol{Q}_{\mathcal{X}}\right)_{ab}=-\boldsymbol{\epsilon}_{abcd}\left(\nabla^{c}\mathcal{X}^{d}+\frac{1}{2}F^{cd}\mathcal{X}^{e}A_{e}\right).\label{Noether charge}
\end{equation}

In what follows, we consider a stationary, axisymmetric spacetime with the electromagnetic potential satisfying $\mathscr{L}_{t}A_{a}=\mathscr{L}_{\varphi}A_{a}=0$ for the timelike and axial Killing fields $t^{a}$ and $\varphi^{a}$, where the metric $g_{ab}$ and the electromagnetic potential $A_{a}$ are vacuum solution of Einstein-Maxwell's equation near infinity, and satisfy
\begin{align}
G^{ab}-\frac{1}{2}T_{F}^{ab} & =\frac{1}{2}T_{\text{M}}^{ab},\\
\nabla_{b}F^{ab} & =j^{a},
\end{align}
for some $T_{\text{M}}^{ab}$ and $j^{a}$ of the superconducting-superfluid mixtures form Eqs. \eqref{TM} and \eqref{electric current} having compact spatial support. In addition, we assume that $\delta g_{ab}$ and $\delta A_{a}$ satisfies the linearized Einstein-superconducting-superfluid equations
\begin{align}
\delta\left(G^{ab}-\frac{1}{2}T_{F}^{ab}\right) & =\frac{1}{2}\delta T_{\text{M}}^{ab},\label{linearized EG}\\
\delta\left(\nabla_{b}F^{ab}\right) & =\delta j^{a},\label{linearized EF}
\end{align}
where $\delta T_{\text{M}}^{ab}$ and $\delta j^{a}$ take the forms of perturbed Eqs. \eqref{TM} and \eqref{electric current} of compact spatial support. However, we impose no symmetry conditions on $\delta g_{ab}$ and $\delta A_{a}$.

For convenience, we choose $\Sigma$ to be axisymmetric Cauchy surface in the sense that $\varphi^{a}$ is tangent to $\Sigma$. With Eq. \eqref{Noether charge} and the axial Killing field $\varphi^{a}$, the ADM angular momentum can be expressed as follows
\begin{widetext}
\begin{align}
J & =-\int_{S_{\infty}}\boldsymbol{Q}_{\varphi}=-\int_{\Sigma}d\boldsymbol{Q}_{\varphi}\nonumber \\
 & =2\int_{\Sigma}\nabla_{e}\left(\nabla^{[d}\varphi^{e]}+\frac{1}{2}F^{de}A_{f}\varphi^{f}\right)\boldsymbol{\epsilon}_{dabc}\nonumber \\
 & =\int_{\Sigma}\left[2R_{\ e}^{d}\varphi^{e}+\nabla_{e}F^{de}A_{f}\varphi^{f}+F^{de}\varphi^{f}\nabla_{e}A_{f}+F^{de}\left(\mathscr{L}_{\varphi}A_{e}-\varphi^{f}\nabla_{f}A_{e}\right)\right]\boldsymbol{\epsilon}_{dabc}\nonumber \\
 & =\int_{\Sigma}\left[2R_{\ e}^{d}\varphi^{e}+\nabla_{e}F^{de}A_{f}\varphi^{f}+F^{de}F_{ef}\varphi^{f}\right]\boldsymbol{\epsilon}_{dabc}\nonumber \\
 & =\int_{\Sigma}\left[2\left(R_{\ f}^{d}-\frac{1}{2}F^{de}F_{fe}\right)\varphi^{f}+\nabla_{e}F^{de}A_{f}\varphi^{f}\right]\boldsymbol{\epsilon}_{dabc}\nonumber \\
 & =\int_{\Sigma}\left[\left(T_{\text{M}}\right)_{\ e}^{d}+j^{d}A_{e}\right]\varphi^{e}\boldsymbol{\epsilon}_{dabc}\nonumber \\
 & =\int_{\Sigma}\sum_{\text{X}}\left(\pi_{a}^{\text{X}}\varphi^{a}\boldsymbol{N}^{\text{X}}\right),\label{total angular momentum}
\end{align}
\end{widetext}
where we have used the Stokes' theorem in the second step, the Killing field identity $\nabla_{a}\nabla_{b}\varphi_{c}=R_{dabc}\varphi^{d}$ in the fourth step, $\mathscr{L}_{\varphi}A_{a}=0$ in the fifth step, and the fact that the pull back of $\varphi^{d}\boldsymbol{\epsilon}_{dabc}$ to $\Sigma$ vanishes in the seventh and eighth steps. So we see that the total angular momentum is the sum of angular momentum of each constituents, i.e.,
\begin{equation}
J=\sum_{\text{X}}J^{\text{X}},\quad J^{\text{X}}=\int_{\Sigma}J_{\text{X}}^{a}\boldsymbol{\epsilon}_{abcd},
\end{equation}
with angular momentum current
\begin{equation}
J_{\text{X}}^{a}=\varphi^{b}\pi_{b}^{\text{X}}n_{\text{X}}^{a}.\label{angular momentum current}
\end{equation}

Our superconducting-superfluid star is in \emph{dynamic equilibrium} if it is the solution to the equations of motion Eqs. \eqref{eomn}-\eqref{eome}, it satisfies
\begin{align}
\mathscr{L}_{t}n_{\text{X}}^{a} & =\mathscr{L}_{\varphi}n_{\text{X}}^{a}=0,\\
\mathscr{L}_{t}s & =\mathscr{L}_{\varphi}s=0,
\end{align}
and the unit flow associated to each constituent are given by the following circular flow condition
\begin{equation}
u_{\text{X}}^{a}=\frac{1}{\vert v_{\text{X}}\vert}\left(t^{a}+\Omega_{\text{X}}\varphi^{a}\right),\label{circular flow condition}
\end{equation}
with $\Omega_{\text{X}}$ the angular velocity and
\begin{equation}
\vert v_{\text{X}}\vert^{2}=-g_{ab}\left(t^{a}+\Omega_{\text{X}}\varphi^{a}\right)\left(t^{b}+\Omega_{\text{X}}\varphi^{b}\right).
\end{equation}
When the star is in dynamic equilibrium, the angular momentum currents Eq. \eqref{angular momentum current} of each constituent are conserved separately, indeed
\begin{align}
\nabla_{a}J_{\text{X}}^{a} & =\nabla_{a}\left(\varphi^{b}\pi_{b}^{\text{X}}n_{\text{X}}^{a}\right)\nonumber \\
 & =n_{\text{X}}^{a}\left(\mathscr{L}_{\varphi}\pi_{a}^{\text{X}}-\varphi^{b}\nabla_{b}\pi_{a}^{\text{X}}+\varphi^{b}\nabla_{a}\pi_{b}^{\text{X}}\right)\nonumber \\
 & =n_{\text{X}}^{a}\mathscr{L}_{\varphi}\pi_{a}^{\text{X}}+\varphi^{b}n_{\text{X}}^{a}w_{ab}^{\text{X}}\nonumber \\
 & =n_{\text{X}}^{a}\mathscr{L}_{\varphi}\pi_{a}^{\text{X}}+\delta_{\text{X}}^{\text{e}}Tn_{\text{e}}\mathscr{L}_{\varphi}s\nonumber \\
 & =0,\label{conservation of j}
\end{align}
where we have used the conservation of particle current in the second step, the fluid equations of motion in the fourth step, and the dynamic equilibrium conditions in the fifth steps. $\delta_{\text{X}}^{\text{e}}$ is the Kronecker symbol to determine whether the index $\text{X}$ is corresponding to electron.

With above preparation, Eq. \eqref{delta M} yields\footnote{The vector fields $t^{a}$ and $\varphi^{a}$ are fixed, i.e., $\delta t^{a}=\delta\varphi^{a}=0$.}
\begin{widetext}
\begin{align}
\delta M= & \int_{\Sigma}t^{a}\left\{ \left(G^{bc}-\frac{1}{2}T_{F}^{bc}\right)\delta g_{bc}\boldsymbol{\epsilon}_{adef}+\nabla_{b}F^{cb}\delta A_{c}\boldsymbol{\epsilon}_{adef}+\delta\left[2\left(E_{G}\right)_{\ a}^{b}\boldsymbol{\epsilon}_{bdef}+A_{a}E_{F}^{b}\boldsymbol{\epsilon}_{bdef}\right]\right\} \nonumber \\
= & \int_{\Sigma}t^{a}\left\{ \frac{1}{2}T_{\text{M}}^{bc}\delta g_{bc}\boldsymbol{\epsilon}_{adef}+j^{b}\delta A_{b}\boldsymbol{\epsilon}_{adef}-\delta\left[\left(T_{\text{M}}\right)_{\ a}^{b}\boldsymbol{\epsilon}_{bdef}+A_{a}j^{b}\boldsymbol{\epsilon}_{bdef}\right]\right\} \nonumber \\
= & \int_{\Sigma}\sum_{\text{X}}t^{a}\left[\frac{1}{2}n_{\text{X}}^{b}\mu^{\text{X}c}\delta g_{bc}\boldsymbol{\epsilon}_{adef}+e^{\text{X}}n_{\text{X}}^{b}\delta A_{b}\boldsymbol{\epsilon}_{adef}-\delta\left(\mu_{a}^{\text{X}}n_{\text{X}}^{b}\boldsymbol{\epsilon}_{bdef}+e^{\text{X}}A_{a}n_{\text{X}}^{b}\boldsymbol{\epsilon}_{bdef}\right)\right]\nonumber \\
 & +\int_{\Sigma}t^{a}\left[\frac{1}{2}\Psi g^{bc}\delta g_{bc}\boldsymbol{\epsilon}_{adef}-\delta\left(\Psi\boldsymbol{\epsilon}_{adef}\right)\right]\nonumber \\
= & \int_{\Sigma}\sum_{\text{X}}t^{a}\left[\frac{1}{2}n_{\text{X}}^{b}\mu^{\text{X}c}\delta g_{bc}\boldsymbol{\epsilon}_{adef}+e^{\text{X}}n_{\text{X}}^{b}\delta A_{b}\boldsymbol{\epsilon}_{adef}-\delta\left(\pi_{a}^{\text{X}}n_{\text{X}}^{b}\boldsymbol{\epsilon}_{bdef}\right)\right]-\int_{\Sigma}t^{a}\delta\Psi\boldsymbol{\epsilon}_{adef}\nonumber \\
= & \int_{\Sigma}\sum_{\text{X}}t^{a}\left[n_{\text{X}}^{b}\delta\pi_{b}^{\text{X}}\boldsymbol{\epsilon}_{adef}-\delta\left(\pi_{a}^{\text{X}}n_{\text{X}}^{b}\boldsymbol{\epsilon}_{bdef}\right)\right]+\int_{\Sigma}t^{a}Tn_{\text{e}}\delta s\boldsymbol{\epsilon}_{adef}\nonumber \\
= & \int_{\Sigma}\sum_{\text{X}}\left\{ \vert v_{\text{X}}\vert u_{\text{X}}^{a}\left[n_{\text{X}}^{b}\delta\pi_{b}^{\text{X}}\boldsymbol{\epsilon}_{adef}-\delta\left(\pi_{a}^{\text{X}}n_{\text{X}}^{b}\boldsymbol{\epsilon}_{bdef}\right)\right]+\Omega_{\text{X}}\varphi^{a}\delta\left(\pi_{a}^{\text{X}}n_{\text{X}}^{b}\boldsymbol{\epsilon}_{bdef}\right)\right\} +\int_{\Sigma}\vert v_{\text{e}}\vert n_{\text{e}}^{a}T\delta s\boldsymbol{\epsilon}_{adef}\nonumber \\
= & \int_{\Sigma}\sum_{\text{X}}\left[\vert v_{\text{X}}\vert\left(-u_{\text{X}}^{a}\pi_{a}^{\text{X}}\right)\delta\left(n_{\text{X}}^{b}\boldsymbol{\epsilon}_{bdef}\right)+\Omega_{\text{X}}\delta\left(\varphi^{a}\pi_{a}^{\text{X}}n_{\text{X}}^{b}\boldsymbol{\epsilon}_{bdef}\right)\right]-\int_{\Sigma}\vert v_{\text{e}}\vert Ts\delta\left(n_{\text{e}}^{a}\boldsymbol{\epsilon}_{adef}\right)+\int_{\Sigma}\vert v_{\text{e}}\vert T\delta\left(sn_{\text{e}}^{a}\boldsymbol{\epsilon}_{adef}\right),
\end{align}
\end{widetext}
where we have used the Eq. \eqref{variation of pressure} in the fifth step, and the circular flow condition Eq. \eqref{circular flow condition} and the fact that pullback of $\varphi^{a}\boldsymbol{\epsilon}_{abcd}$ to $\Sigma$ vanishes in the sixth step. Define the redshifted chemical potentials of each constituent and the redshifted temperature by
\begin{align}
\tilde{\mu}_{\text{n}} & =-u_{\text{n}}^{a}\pi_{a}^{\text{n}}\vert v_{\text{n}}\vert,\nonumber \\
\tilde{\mu}_{\text{p}} & =-u_{\text{p}}^{a}\pi_{a}^{\text{p}}\vert v_{\text{p}}\vert,\nonumber \\
\tilde{\mu}_{\text{e}} & =-(u_{\text{e}}^{a}\pi_{a}^{\text{e}}+Ts)\vert v_{\text{e}}\vert,\nonumber \\
\tilde{T} & =T\vert v_{\text{e}}\vert,
\end{align}
then we end up with the desired form of the first law of thermodynamics holding for arbitrary perturbations off of a superconducting-superfluid star in dynamic equilibrium
\begin{equation}
\delta M=\int_{\Sigma}\left(\sum_{\text{X}}\tilde{\mu}_{\text{X}}\delta\boldsymbol{N}^{\text{X}}+\tilde{T}\delta\boldsymbol{S}+\sum_{\text{X}}\Omega_{\text{X}}\delta\boldsymbol{J}^{\text{X}}\right),\label{first law}
\end{equation}
where $\boldsymbol{N}^{\text{X}}$, $\boldsymbol{S}$, and $\boldsymbol{J}^{\text{X}}$ are the Hodge dual 3-form respectively to the particle current $n_{\text{X}}^{a}$, the entropy current $sn_{\text{e}}^{a}$, and the angular momentum current $J_{\text{X}}^{a}$. By the conservation law Eqs. \eqref{conservation of n=000026s} and \eqref{conservation of j}, they are all closed forms. The number of particles $N^{\text{X}}$ of each constituent, the total entropy $S$, and the angular momentum $J^{\text{X}}$ of each constituent are given by
\begin{equation}
N^{\text{X}}=\int_{\Sigma}\boldsymbol{N}^{\text{X}},\quad S=\int_{\Sigma}\boldsymbol{S},\quad J^{\text{X}}=\int_{\Sigma}\boldsymbol{J}^{\text{X}}.\label{NSJ}
\end{equation}

\subsection{Existence of desired solutions to the linearized constraints}

Before we going to talk about the thermodynamic equilibrium, we want first to check that whether the linearized constraint equations Eqs. \eqref{linearized EG} and \eqref{linearized EF} will prevent us from choosing $\delta\boldsymbol{N}^{\text{X}}$, $\delta\boldsymbol{S}$, and $\delta\boldsymbol{J}^{\text{X}}$ freely. Let $\Sigma$ be a $t-\varphi$ reflection invariant Cauchy surface for a star in dynamic equilibrium, and we would like to fix our coordinate system in which the metric takes
\begin{equation}
ds^{2}=-\alpha^{2}d\tau^{2}+h_{ij}\left(dx^{i}+\beta^{i}d\tau\right)\left(dx^{j}+\beta^{j}d\tau\right),\label{coordinate}
\end{equation}
with $\Sigma$ given by the surface of $\tau=0$, and the unit normal covector of $\Sigma$ is $\nu_{a}=-\alpha\left(d\tau\right)_{a}$. Let $\boldsymbol{e}$ be a fixed, non-dynamical volume element on $\Sigma$, so the volume element associated with the induced metric on $\Sigma$ is $\sqrt{h}\boldsymbol{e}$. Consider perturbations off of this background, the linearized Hamiltonian constraint on $\Sigma$ is
\begin{equation}
0=-2\delta\left(\sqrt{h}\nu^{a}\nu^{b}G_{ab}\right)+\delta\left(\sqrt{h}\nu^{a}\nu^{b}T_{ab}\right),\label{Einstein constraint}
\end{equation}
the linearized momentum constraint is
\begin{equation}
0=-2\delta\left(\sqrt{h}h_{a}^{b}\nu^{c}G_{bc}\right)+\delta\left(\sqrt{h}h_{a}^{b}\nu^{c}T_{bc}\right),
\end{equation}
and the linearized constraint from the electromagnetic potential
\begin{equation}
0=-\delta\left(\sqrt{h}\nu_{a}\nabla_{b}F^{ab}\right)+\delta\left(\sqrt{h}\nu_{a}j^{a}\right).\label{Maxwell constraint}
\end{equation}

Let
\begin{align}
\mathcal{N}_{\text{X}} & \equiv-\sqrt{h}n_{\text{X}}^{a}\nu_{a},\nonumber \\
\mathcal{S} & \equiv s\mathcal{N}_{\text{e}},\nonumber \\
\mathcal{J}_{\text{X}} & \equiv\varphi^{a}\pi_{a}^{\text{X}}\mathcal{N}_{\text{X}},
\end{align}
then the pullback of $\delta\boldsymbol{N}^{\text{X}}$, $\delta\boldsymbol{S}$, and $\delta\boldsymbol{J}^{\text{X}}$ on $\Sigma$ are given by
\begin{align}
\delta\boldsymbol{N}^{\text{X}} & =\left(\delta\mathcal{N}_{\text{X}}\right)\boldsymbol{e},\nonumber \\
\delta\boldsymbol{S} & =\left(\delta\mathcal{S}\right)\boldsymbol{e},\nonumber \\
\delta\boldsymbol{J}^{\text{X}} & =\left(\delta\mathcal{J}_{\text{X}}\right)\boldsymbol{e},
\end{align}
And to facilitate our calculation, below we like to work with the gauge in which $\alpha=1$, $\beta^{i}=0$ and $\delta\alpha=\delta\beta^{i}=0$
on $\Sigma$. 

By the Gauss-Codazzi equation
\begin{equation}
2\nu_{a}\nu_{b}G^{ab}=R^{\left(3\right)}-K_{ab}K^{ab}+K^{2},
\end{equation}
the circular flow condition $u_{\text{X}}^{a}=-\nu^{a}u_{\text{X}}^{b}\nu_{b}+\frac{u_{\text{X}}^{b}\varphi_{b}}{\varphi^{c}\varphi_{c}}\varphi^{a}$, as well as the background $K=0$ due to the fact that $K_{ab}$ is odd under $t-\varphi$ reflection, the linearized Hamiltonian constraint takes the form
\begin{align}
 & \sqrt{h}\bigg\{-R^{\left(3\right)ij}\delta h_{ij}+D^{i}D^{j}\delta h_{ij}-D^{i}D_{i}\delta h_{j}^{j}\nonumber \\
 & +h^{-1}\pi_{ij}\pi^{ij}\delta h_{k}^{k}-2h^{-1}\pi_{ij}\delta\pi^{ij}-2h^{-1}\pi_{i}^{j}\pi^{ik}\delta h_{jk}\nonumber \\
 & -E_{G}^{ab}\nu_{a}\nu_{b}\delta h_{j}^{j}+\frac{1}{2}T_{\text{M}}^{ij}\delta h_{ij}\nonumber \\
 & -h^{-1}\Pi_{i}\delta\Pi^{i}+\frac{1}{2}h^{-1}\Pi_{i}\Pi^{i}\delta h_{j}^{j}-\frac{1}{2}h^{-1}\Pi^{i}\Pi^{j}\delta h_{ij}\nonumber \\
 & +\left(D_{k}A^{i}-D^{i}A_{k}\right)D^{[k}A^{j]}\delta h_{ij}-2D^{[i}A^{j]}D_{[i}\delta A_{j]}\bigg\}\nonumber \\
 & -\sum_{\text{X}}\frac{1}{u_{\text{X}}^{d}\nu_{d}}\frac{u_{\text{X}}^{b}\varphi_{b}}{\varphi^{c}\varphi_{c}}\mathcal{N}_{\text{X}}e^{\text{X}}\varphi^{a}\delta A_{a}\nonumber \\
= & \sum_{\text{X}}\left(-\frac{1}{u_{\text{X}}^{c}\nu_{c}}\right)\left(-u_{\text{X}}^{a}\mu_{a}^{\text{X}}-e^{\text{X}}\frac{u_{\text{X}}^{b}\varphi_{b}}{\varphi^{d}\varphi_{d}}\varphi^{a}A_{a}\right)\delta\mathcal{N}_{\text{X}}\nonumber \\
 & -\frac{1}{u_{\text{e}}^{c}\nu_{c}}\left(-Ts\delta\mathcal{N}_{\text{e}}+T\delta\mathcal{S}\right)-\sum_{\text{X}}\frac{1}{u_{\text{X}}^{c}\nu_{c}}\frac{u_{\text{X}}^{a}\varphi_{a}}{\varphi^{b}\varphi_{b}}\delta\mathcal{J}_{\text{X}}.\label{Hamiltonian constraint}
\end{align}
Here $\delta h_{i}^{j}=h^{jk}\delta h_{ik}$, $D$ and $R^{\left(3\right)ij}$ are the derivative operator and the Ricci tensor associated with $h_{ij}$ on $\Sigma$, while $\pi^{ij}=\sqrt{h}\left(K^{ij}-Kh^{ij}\right)$ and $\Pi^{i}=\sqrt{h}\nu_{a}F^{ai}$ with $K_{ij}$ the extrinsic curvature. Similarly, with the Codazzi-Maindardi equation
\begin{equation}
h_{ab}\nu_{c}G^{bc}=D_{b}K_{\ a}^{b}-D_{a}K,
\end{equation}
the $\varphi^{i}$-component of the linearized momentum constraint is
\begin{align}
 & \sqrt{h}\varphi^{i}\bigg\{2D_{j}\left(h^{-\frac{1}{2}}\delta\pi_{i}^{j}\right)+2D_{j}(h^{-\frac{1}{2}}\pi^{jk})\delta h_{ik}\nonumber \\
 & +2h^{-\frac{1}{2}}\pi^{jk}D_{j}\delta h_{ik}-\pi^{jk}D_{i}\delta h_{jk}\nonumber \\
 & -2h^{-\frac{1}{2}}\Pi^{j}D_{[i}\delta A_{j]}-2h^{-\frac{1}{2}}D_{[i}A_{j]}\delta\Pi^{j}\bigg\}\nonumber \\
 & -\sum_{\text{X}}e^{\text{X}}\mathcal{N}_{\text{X}}\varphi^{i}\delta A_{i}\nonumber \\
= & \sum_{\text{X}}\left(e^{\text{X}}\varphi^{a}A_{a}\delta\mathcal{N}_{\text{X}}-\delta\mathcal{J}_{\text{X}}\right),\label{momentum constraint para phi}
\end{align}
and the components of the linearized momentum constraints perpendicular to $\varphi^{a}$ are
\begin{align}
 & \sqrt{h}\left(h^{il}-\frac{\varphi^{i}\varphi^{l}}{\vert\varphi\vert^{2}}\right)\bigg\{2D_{j}\left(h^{-\frac{1}{2}}\delta\pi_{i}^{j}\right)\nonumber \\
 & +2h^{-\frac{1}{2}}\pi^{jk}D_{j}\delta h_{ik}-\pi^{jk}D_{i}\delta h_{jk}\nonumber \\
 & -2h^{-\frac{1}{2}}\Pi^{j}D_{[i}\delta A_{j]}-2h^{-\frac{1}{2}}D_{[i}A_{j]}\delta\Pi^{j}\nonumber \\
 & -\left(D_{k}A^{j}-D^{j}A_{k}\right)\Pi^{k}\delta h_{ij}\bigg\}\nonumber \\
= & -\sum_{\text{X}}\mathcal{N}_{\text{X}}\delta\left(\mu_{i}^{\text{X}}h^{il}\right)_{\perp},\label{momentum constraint perp phi}
\end{align}
where $\delta\pi_{i}^{j}=h_{ik}\delta\pi^{jk}$ and the symbol $\perp$ means that $\delta\left(\mu_{i}^{\text{X}}h^{il}\right)_{\perp}=\left(h_{\ j}^{l}-\frac{\varphi^{l}\varphi_{j}}{\vert\varphi\vert^{2}}\right)\delta\left(\mu_{i}^{\text{X}}h^{ij}\right)$. Finally, the linearized constraint from the electromagnetic potential is given by
\begin{equation}
D_{i}\left(h^{-\frac{1}{2}}\delta\Pi^{i}\right)=-\sum_{\text{X}}e^{\text{X}}\delta\mathcal{N}_{\text{X}}.\label{gauge constraint}
\end{equation}

The following lemma show that on a $t-\varphi$ reflection invariant Cauchy surface $\Sigma$ of the background spacetime, a solution to the linearized Einstein-Maxwell constraint equations Eqs. \eqref{Einstein constraint}-\eqref{Maxwell constraint} always can be found for any given axisymmetric specifications of $\delta\boldsymbol{N}^{\text{X}}$, $\delta\boldsymbol{S}$, and $\delta\boldsymbol{J}^{\text{X}}$.
\begin{lem}
\label{lem:solution to constraint}Let $\delta\mathcal{N}_{\text{X}}$, $\delta\mathcal{S}$, and $\delta\mathcal{J}_{\text{X}}$ be specified arbitrary smooth, axisymmetric functions with compact support, such that $\delta\mathcal{J}_{\text{X}}/\varphi^{a}\varphi_{a}$ also is smooth. Then we can choose the remaining initial data $\left(\delta h_{ij},\delta\pi^{ij},\delta A_{i},\delta\Pi^{i},\delta\mu_{\perp}^{\text{X}i}\right)$ so as to solve the linearized constraints Eqs. \eqref{Hamiltonian constraint}, \eqref{momentum constraint para phi}, \eqref{momentum constraint perp phi}, and \eqref{gauge constraint}.
\end{lem}

\begin{proof}
We choose $\delta h_{ij}$, $\delta\pi^{ij}$, $\delta A_{i}$, and $\delta\Pi^{i}$ be of the form
\begin{align}
\delta h_{ij} & =\psi h_{ij},\quad\delta\pi^{ij}=\sqrt{h}D^{(i}H^{j)}-\psi\pi^{ij},\nonumber \\
\delta A_{i} & =0,\quad\delta\Pi^{i}=\sqrt{h}D^{i}\Phi,\label{solution}
\end{align}
then the linearized constraint Eq. \eqref{gauge constraint} reduces to
\begin{equation}
D_{i}D^{i}\Phi=-\sum_{\text{X}}e^{\text{X}}\delta\mathcal{N}_{\text{X}},
\end{equation}
which has a unique axisymmetric solution that goes to zero at infinity given any prescribed perturbations $\delta\mathcal{N}_{\text{X}}$ \cite{Shi:2021dvd}. Furthermore, Eq. \eqref{momentum constraint para phi} can be cast into 
\begin{align}
 & \varphi_{i}\left(D_{j}D^{(i}H^{j)}-D^{[i}A^{j]}D_{j}\Phi\right)\nonumber \\
= & \frac{1}{2\sqrt{h}}\sum_{\text{X}}\left(e^{\text{X}}\varphi^{a}A_{a}\delta\mathcal{N}_{\text{X}}-\delta\mathcal{J}_{\text{X}}\right)\equiv\mathscr{J},
\end{align}
so we choose $H^{i}$ to satisfy 
\begin{equation}
D_{j}D^{(j}H^{i)}=D^{[i}A^{j]}D_{j}\Phi+\frac{\mathscr{J}\varphi^{i}}{\vert\varphi\vert^{2}}.\label{equation of H}
\end{equation}
Since the right side is a smooth vector field of compact support, there also exists a unique solution to Eq. \eqref{equation of H} that goes to zero at infinity \cite{Chrusciel:2003sr}.

By substituting Eq. \eqref{solution} into the linearized Hamiltonian constraint Eq. \eqref{Hamiltonian constraint}, we have
\begin{align}
 & -D^{i}D_{i}\psi+\mathscr{M}\psi\nonumber \\
= & h^{-\frac{1}{2}}\pi_{ij}D^{i}H^{j}+\frac{1}{2}h^{-\frac{1}{2}}\Pi_{i}D^{i}\Phi\nonumber \\
 & +\sum_{\text{X}}\left(-\frac{1}{u_{\text{X}}^{c}\nu_{c}}\right)\left(-u_{\text{X}}^{a}\mu_{a}^{\text{X}}-e^{\text{X}}\frac{u_{\text{X}}^{b}\varphi_{b}}{\varphi^{d}\varphi_{d}}\varphi^{a}A_{a}\right)\delta\mathcal{N}_{\text{X}}\nonumber \\
 & -\frac{1}{u_{\text{e}}^{c}\nu_{c}}\left(-Ts\delta\mathcal{N}_{\text{e}}+T\delta\mathcal{S}\right)-\sum_{\text{X}}\frac{1}{u_{\text{X}}^{c}\nu_{c}}\frac{u_{\text{X}}^{a}\varphi_{a}}{\varphi^{b}\varphi_{b}}\delta\mathcal{J}_{\text{X}},\label{equation of psi}
\end{align}
where
\begin{align}
\mathscr{M}= & h^{-1}\pi_{ij}\pi^{ij}+\frac{1}{4}h^{-1}\Pi_{i}\Pi^{i}+\frac{1}{2}D_{[i}A_{j]}D^{[i}A^{j]}\nonumber \\
 & +\frac{1}{2\vert\varphi\vert^{2}}\sum_{\text{X,Y}}\left(\varphi_{a}n_{\text{X}}^{a}\right)\mathbb{I}_{\text{XY}}\left(\varphi_{b}n_{\text{Y}}^{b}\right)+\frac{1}{4}\left(-\Lambda_{\text{M}}+3\Psi_{\text{M}}\right),
\end{align}
The positive definiteness of the inertia matrix Eq. \eqref{inertia matrix} implies that $\mathscr{M}$ is non-negative, and since the right side of Eq. \eqref{equation of psi} vanishes suitably rapidly at infinity, then there exists a unique solution, $\psi$, of this equation that vanishes at infinity \cite{Chrusciel:2003sr}. 

Finally, Eq. \eqref{momentum constraint perp phi} boils down into
\begin{align}
 & \sqrt{h}\psi\left[-2D_{j}(h^{-1/2}\pi_{i}^{j})+2D_{[i}A_{j]}\Pi^{j}\right]\left(h^{ik}-\frac{\varphi^{i}\varphi^{k}}{\vert\varphi\vert^{2}}\right)\nonumber \\
= & -\sum_{\text{X}}\mathcal{N}_{\text{X}}\delta(\mu_{i}^{\text{X}}h^{ik})_{\perp},\label{equation of mu}
\end{align}
where we have used Eq. \eqref{equation of H}. Clearly, Eq. \eqref{equation of mu} is an algebraic equation for $\delta(\mu_{i}^{\text{X}}h^{ik})_{\perp}$ and can be readily fulfilled.
\end{proof}

\subsection{Thermodynamic equilibrium}

\label{subsec4.3}

Due to the conservation law Eqs. \eqref{conservation of n=000026s} and \eqref{conservation of j}, the quantities in Eq. \eqref{NSJ} are all conserved. So based on these quantities, we give our first definition of thermodynamic equilibrium: a superconducting-superfluid star in dynamic equilibrium is said to be in \emph{weak thermodynamic equilibrium} if and only if $\delta S=0$ with respect to all perturbations that satisfy the linearized Einstein-Maxwell constraint equations and for which $\delta M=\delta N^{\text{X}}=\delta J^{\text{X}}=0$. The necessary and sufficient condition for a superconducting-superfluid star in dynamic equilibrium to be in weak thermodynamic equilibrium is given in the following theorem.
\begin{thm}
\label{thm:weakly thermodynamic equilibrium}A dynamic equilibrium configuration is in weak thermodynamic equilibrium if and only if $\tilde{\mu}_{\text{X}}$, $\tilde{T}$, and $\Omega_{\text{X}}$ are uniform throughout the star.
\end{thm}

\begin{proof}
If $\tilde{\mu}_{\text{X}}$, $\tilde{T}$, and $\Omega_{\text{X}}$ are uniform throughout the star, then the first law Eq. \eqref{first law} reduces to
\begin{equation}
\delta M=\sum_{\text{X}}\tilde{\mu}_{\text{X}}\delta N^{\text{X}}+\tilde{T}\delta S+\sum_{\text{X}}\Omega_{\text{X}}\delta J^{\text{X}}.
\end{equation}
It is evident that $\delta S=0$ for any perturbation with $\delta M=\delta N^{\text{X}}=\delta J^{\text{X}}=0$, i.e., the star is in weak thermodynamic equilibrium.

According to Lemma \ref{lem:solution to constraint}, consider a star in weak thermodynamic equilibrium and an arbitrary perturbation $\delta_{1}$ satisfying $\delta_{1}N^{\text{X}}=\delta_{1}J^{\text{X}}=0$ and $\delta_{1}S\neq0$, then there must be $\delta_{1}M\neq0$. Suppose that at least one of $\tilde{\mu}_{\text{X}}$, $\tilde{T}$, and $\Omega_{\text{X}}$ were not uniform throughout the star, without loss of generality, say $\tilde{T}$ is not uniform. As our star is a compact object, let $\Gamma$ be the compact support of the star, which is contained in some open subset $\mathcal{U}$ of $\Sigma$ whose closure is also compact. According to the existence of $C_\infty$ partition of unity, there exists a smooth bump function $\kappa:\Sigma\rightarrow\left[0,1\right]$ such that $\kappa\equiv1$ on $\Gamma$ and $\text{supp}\kappa\subset\mathcal{U}$. Let 
\begin{equation}
\overline{T}=\frac{\int_{\mathcal{U}}\kappa\tilde{T}\nu^{a}\boldsymbol{\epsilon}_{abcd}}{\int_{\mathcal{U}}\kappa\nu^{a}\boldsymbol{\epsilon}_{abcd}}.
\end{equation}
Now choose a second perturbation $\delta_{2}$ with perturbative parameter $\varepsilon$ satisfying the linearized constraint equations given by
\begin{align}
\delta_{2}\boldsymbol{N}^{\text{X}} & =\delta_{2}\boldsymbol{J}^{\text{X}}=0,\nonumber \\
\delta_{2}\boldsymbol{S}\vert_{\mathcal{U}} & =\varepsilon\kappa\left(\tilde{T}-\overline{T}\right)\nu^{a}\boldsymbol{\epsilon}_{abcd},\quad\delta_{2}\boldsymbol{S}\vert_{\Sigma\backslash\mathcal{U}}=0,
\end{align}
since clearly $\delta_{2}\boldsymbol{N}^{\text{X}}$, $\delta_{2}\boldsymbol{J}^{\text{X}}$, and $\delta_2\boldsymbol{S}$ are smooth with compact supports, we can always do this choice by using Lemma \ref{lem:solution to constraint} again. With such a perturbation, one has
\begin{equation}
\delta_{2}N^{\text{X}}=\delta_{2}J^{\text{X}}=\delta_{2}S=0,
\end{equation}
and
\begin{align}
\delta_{2}M &=\int_{\Sigma}\left(\sum_{\text{X}}\tilde{\mu}_{\text{X}}\delta_2\boldsymbol{N}^{\text{X}}+\tilde{T}\delta_2\boldsymbol{S}+\sum_{\text{X}}\Omega_{\text{X}}\delta_2\boldsymbol{J}^{\text{X}}\right)\nonumber \\
&=\int_{\Sigma}\left(\tilde{T}-\overline{T}\right)\delta_2\boldsymbol{S}+\overline{T}\delta_2S\\
 & =\varepsilon\int_{\mathcal{U}}\kappa\left(\tilde{T}-\overline{T}\right)^{2}\nu^{a}\boldsymbol{\epsilon}_{abcd}.
\end{align}
Adjust suitable $\varepsilon$ such that $\delta_{2}M=-\delta_{1}M$, then we find a perturbation $\delta=\delta_{1}+\delta_{2}$ satisfying $\delta M=\delta N^{\text{X}}=\delta J^{\text{X}}=0$ but $\delta S\neq0$, which leads to a contradiction. Hence $\tilde{T}$ must be uniform throughout the star, and similar arguments show that $\tilde{\mu}_{\text{X}}$ and $\Omega_{\text{X}}$ need also be uniform throughout the star.
\end{proof}
The conservation of particle number and angular momentum of each constituent in fact also indicate the conservation of the total particle number $N=\sum_{\text{X}}N^{\text{X}}$ and the total angular momentum $J=\sum_{\text{X}}J^{\text{X}}$. So based on the total particle number $N$, the total entropy $S$, and the total angular momentum $J$, we give our second definition of thermodynamic equilibrium: a superconducting-superfluid star in dynamic equilibrium is said to be in \emph{strong thermodynamic equilibrium} if and only if $\delta S=0$ with respect to all perturbations that satisfy the linearized Einstein-Maxwell constraint equations and for which $\delta M=\delta N=\delta J=0$. The necessary and sufficient condition for a superconducting-superfluid star in dynamic equilibrium to be in strong thermodynamic equilibrium is given in the following theorem.
\begin{thm}
\label{thm:strongly thermodynamic equilibrium}A dynamic equilibrium configuration is in strong thermodynamic equilibrium if and only if there is no differential rotation (i.e., there is some $\Omega$ such that $\Omega_{\text{X}}=\Omega$ for $\text{X}=\text{n},\text{p},\text{e}$), the configuration is in chemical equilibrium (i.e., there is some $\tilde{\mu}$ such that $\tilde{\mu}_{\text{X}}=\tilde{\mu}$ for $\text{X}=\text{n},\text{p},\text{e}$), and $\tilde{\mu}$, $\tilde{T}$, $\Omega$ are all uniform throughout the star.
\end{thm}

\begin{proof}
The proof of ``if'' part is straightforward and almost same as the proof of Theorem \ref{thm:weakly thermodynamic equilibrium}.

If the star is in strong thermodynamic equilibrium, consider an arbitrary perturbation $\delta_{1}$ such that $\delta_{1}N=\delta_{1}J=0$ and $\delta_{1}S\neq0$, then there must be $\delta_{1}M\neq0$. Suppose there exists differential rotation, without loss of generality, say $\Omega_{\text{n}}\neq\Omega_{\text{p}}$, choose a second perturbation $\delta_{2}$ with perturbative parameter $\varepsilon$ such that
\begin{align}
\delta_{2}\boldsymbol{N}^{\text{X}} & =\delta_{2}\boldsymbol{S}=\delta_{2}\boldsymbol{J}^{\text{e}}=0,\nonumber \\
\delta_{2}\boldsymbol{J}^{\text{n}}\vert_{\mathcal{U}} & =-\delta_{2}\boldsymbol{J}^{\text{p}}\vert_{\mathcal{U}}=\varepsilon\kappa\left(\Omega_{\text{n}}-\Omega_{\text{p}}\right)\nu^{a}\boldsymbol{\epsilon}_{abcd},\nonumber \\
\delta_{2}\boldsymbol{J}^{\text{n}}\vert_{\Sigma\backslash\mathcal{U}} & =\delta_{2}\boldsymbol{J}^{\text{p}}\vert_{\Sigma\backslash\mathcal{U}}=0,
\end{align}
then one has
\begin{equation}
\delta_2 N=\delta_2 S=\delta_2 J=0,
\end{equation}
and
\begin{align}
\delta_{2}M & =\int_{\Sigma}\left(\sum_{\text{X}}\tilde{\mu}_{\text{X}}\delta_2\boldsymbol{N}^{\text{X}}+\tilde{T}\delta_2\boldsymbol{S}+\sum_{\text{X}}\Omega_{\text{X}}\delta_2\boldsymbol{J}^{\text{X}}\right)\nonumber \\
 & =\varepsilon\int_{\mathcal{U}}\kappa\left(\Omega_{\text{n}}-\Omega_{\text{p}}\right)^{2}\nu^{a}\boldsymbol{\epsilon}_{abcd}.
\end{align}
Adjust suitable $\varepsilon$ such that $\delta_{2}M=-\delta_{1}M$, then we find a perturbation $\delta=\delta_{1}+\delta_{2}$ satisfying $\delta M=\delta N=\delta J=0$, but $\delta S\neq0$. Hence there need be no differential rotation, and similarly, the configuration need be in chemical equilibrium. The proof of that $\tilde{\mu}$, $\tilde{T}$, $\Omega$ should be uniform throughout the star is also almost same as the proof of Theorem \ref{thm:weakly thermodynamic equilibrium}.
\end{proof}
In fact, the chemical equilibrium can not be established unless there is no differential rotation, which can be seen from a simple argument: If the three species rotate at different rates, we must work in one of the rest frames of these species, and the result will depend on which frame we choose. For instance, in the frame rotating with the neutrons, the proton and electron chemical potentials will have an additional kinematic piece. This is also true for the any two chemical potentials in the frame that co-rotates with the third constituents. In other words, it would seem possible to have chemical equilibrium only if the three fluids co-rotate.

As stated in the introduction, we will not consider the transfusive effect in the star, but we can give a simple argument here about the thermodynamic equilibrium in the transfusive case. If the transfusive effect exists, rather than the separate conservation laws Eq.\eqref{conservation of n=000026s} and the equations of motion Eqs. \eqref{eomn}-\eqref{eome} of each constituents, there will only be the conservation of total particle current and entropy current
\begin{equation}
\nabla_{a}n^{a}=\nabla_{a}\left(\sum_{\text{X}}n_{\text{X}}^{a}\right)=0,\quad\nabla_{a}\left(sn_{\text{e}}^{a}\right)=0.
\end{equation}
And the equation of motion is given by \cite{Carter:1987qr}
\begin{equation}
\nabla_{a}T_{\ b}^{a}=\nabla_{a}\left(T_{\text{M}}\right)_{\ b}^{a}+j^{a}F_{ab}=0.\label{conservation of em tensor}
\end{equation}
When the equations of motion of each constituent are not satisfied, Eq. \eqref{conservation of j} implies that the angular momentum of each constituent may be no longer conserved. However, the total angular momentum $J$ is still conserved. Indeed, note that the total electric current vector $j^{a}$ is conserved
\begin{align}
\nabla_{a}j^{a} & =\nabla_{a}\nabla_{b}F^{ab}=\frac{1}{2}\left[\nabla_{a},\nabla_{b}\right]F^{ab}\nonumber \\
 & =\frac{1}{2}\left(R_{\ cab}^{a}F^{cb}+R_{\ cab}^{b}F^{ac}\right)\nonumber \\
 & =\frac{1}{2}\left(R_{cb}F^{cb}-R_{ca}F^{ac}\right)\nonumber \\
 & =0,
\end{align}
accordingly, the total angular momentum current $J^{a}=\left(T_{\text{M}}\right)_{\ b}^{a}\varphi^{b}+j^{a}A_{b}\varphi^{b}$ in Eq. \eqref{total angular momentum} satisfies the conservation law
\begin{align}
 & \nabla_{a}J^{a}\nonumber \\
= & \nabla_{a}\left[\left(T_{\text{M}}\right)_{\ b}^{a}\varphi^{b}+j^{a}A_{b}\varphi^{b}\right]\nonumber \\
= & \varphi^{b}\nabla_{a}\left(T_{\text{M}}\right)_{\ b}^{a}+T_{\text{M}}^{ab}\nabla_{a}\varphi_{b}+j^{a}\left(\varphi^{b}\nabla_{a}A_{b}+A_{b}\nabla_{a}\varphi^{b}\right)\nonumber \\
= & -j^{a}\varphi^{b}F_{ab}+j^{a}\left(2\varphi^{b}\nabla_{[a}A_{b]}+\mathscr{L}_{\varphi}A_{a}\right)\nonumber \\
= & 0,
\end{align}
where we have used $\varphi^{a}$ is Killing vector, the equation of motion Eq. \eqref{conservation of em tensor}, and $\mathscr{L}_{\varphi}A_{a}=0$.

Now, in the transfusive case, without the conserved quantities given in Eq. \eqref{NSJ}, our definition of weak thermodynamic equilibrium obviously can not work, but the definition of strong thermodynamic equilibrium still works due to the conservation of the total particle number $N$, the total entropy $S$, and the total angular momentum $J$. As shown in \cite{Langlois:1997bz}, if there is a differential rotation, the different constituents will drag each other until they have same angular velocity. This statement actually coincides our condition that there will be no differential rotation as stated in Theorem \ref{thm:strongly thermodynamic equilibrium}. On the other hand, although the three fluids with differential rotation will finally be locked together, this progress likely takes many hundreds of dynamical timescales, so when the timescale under consideration is not too long, the transfusive effect is not important, and the non-transfusive model will be a good approximation to describe the neutron star.

\section{Lagrangian formulation of superconducting-superfluid mixtures: symplectic structure, phase space, and trivial displacements}

\label{sec5}

We want to construct the phase space of our Einstein-superconducting-superfluid system. To achieve this, we will give the pre-symplectic form $W$ by using Wald formalism introduced in Sec. \ref{sec3}, and we will also determine the characteristic equation of sound velocity for our neutron star. With these preparations, we then seek the degeneracy of $W$ and construct the phase space by factoring these degeneracy. Moreover, since in the next section, it will be important to determine the symplectic complement corresponding to field variations of the trivial perturbation within the subspace consisting of weak solutions to the linearized constraints, so we will talk about the trivial perturbations generated by the trivial displacements in the last subsection.

\subsection{Lagrangian and symplectic form}

The Lagrangian for the Einstein-superconducting-superfluid system is taken to be
\begin{equation}
\boldsymbol{\mathcal{L}}=\boldsymbol{\epsilon}\left(R-\frac{1}{4}F_{ab}F^{ab}+j^{a}A_{a}+\Lambda_{\text{M}}\right),\label{Lagrangian for star}
\end{equation}
then the variation of Lagrangian yields
\begin{align}
\delta\boldsymbol{\mathcal{L}}= & \boldsymbol{\epsilon}E_{G}^{ab}\delta g_{ab}+\boldsymbol{\epsilon}E_{F}^{a}\delta A_{a}+\boldsymbol{\epsilon}\sum_{\text{X}}f_{a}^{\text{X}}\xi_{\text{X}}^{a}\nonumber \\
 & +d(\iota_{x}\boldsymbol{\epsilon}+\iota_{y}\boldsymbol{\epsilon}+\sum_{\text{X}}\iota_{z_{\text{X}}}\boldsymbol{\epsilon}),\label{Lagrangian variation}
\end{align}
where
\begin{align}
x^{a} & =g^{ab}g^{cd}\left(\nabla_{d}\delta g_{bc}-\nabla_{b}\delta g_{cd}\right),\\
y^{a} & =-F^{ab}\delta A_{b},\\
z_{\text{X}}^{a} & =2\pi_{b}^{\text{X}}n_{\text{X}}^{[a}\xi_{\text{X}}^{b]},
\end{align}
and the equations of motion that beside Eqs. \eqref{eomn}-\eqref{eome} are given by
\begin{align}
E_{G}^{ab} & =-G^{ab}+\frac{1}{2}T^{ab}=0,\\
E_{F}^{a} & =j^{a}-\nabla_{b}F^{ab}=0.
\end{align}
From Eq. \eqref{Lagrangian variation}, we may also read off the symplectic potential current $\boldsymbol{\theta}$
\begin{align}
\boldsymbol{\theta}\left(\phi;\delta\phi\right)= & \boldsymbol{\theta}^{\left(G\right)}\left(g;\delta g\right)+\boldsymbol{\theta}^{\left(F\right)}\left(\phi;\delta\phi\right)+\boldsymbol{\theta}^{\left(\text{M}\right)}\left(\phi;\delta\phi\right)\nonumber \\
= & g^{de}g^{fh}(\nabla_{h}\delta g_{ef}-\nabla_{e}\delta g_{fh})\boldsymbol{\epsilon}_{dabc}\nonumber \\
 & -F^{de}\delta A_{e}\boldsymbol{\epsilon}_{dabc}+\sum_{\text{X}}\xi_{\text{X}}^{d}\left(\boldsymbol{P}_{d}^{\text{X}}\right)_{abc},
\end{align}
with
\begin{equation}
\left(\boldsymbol{P}_{d}^{\text{X}}\right)_{abc}=2\pi_{e}^{\text{X}}n_{\text{X}}^{[f}\delta_{d}^{e]}\boldsymbol{\epsilon}_{fabc},
\end{equation}
And the constraint 3-form $\boldsymbol{C}_{\mathcal{X}}$ is still
\begin{equation}
\left(\boldsymbol{C}_{\mathcal{X}}\right)_{abc}=-\mathcal{X}^{d}\boldsymbol{\epsilon}_{eabc}\left[2\left(E_{G}\right)_{\ d}^{e}+E_{F}^{e}A_{d}\right].
\end{equation}

In order to calculate the pre-symplectic current $\boldsymbol{\omega}$ using Eq. \eqref{pre-symplectic current}, for perturbations $\delta_{1}\phi=\left(\delta_{1}g_{ab},\delta_{1}A_{a},\xi_{\text{X}}^{a}\right)$ and $\delta_{2}\phi=\left(\delta_{2}g_{ab},\delta_{2}A_{a},\zeta_{\text{X}}^{a}\right)$, we choose $\left(\delta_{1}g_{ab},\delta_{1}A_{a}\right)$ and $\left(\delta_{2}g_{ab},\delta_{2}A_{a}\right)$ as variations along a two-parameter family of metrics and gauge fields $\left(g_{ab}\left(\lambda_{1},\lambda_{2}\right),A_{a}\left(\lambda_{1},\lambda_{2}\right)\right)$, which means that
\begin{equation}
\delta_{1}\delta_{2}g_{ab}=\delta_{2}\delta_{1}g_{ab},\quad\delta_{1}\delta_{2}A_{a}=\delta_{2}\delta_{1}A_{a},
\end{equation}
and we choose $\xi_{\text{X}}^{a}$ and $\zeta_{\text{X}}^{a}$ to be fixed, i.e.,
\begin{align}
\delta_{1}\zeta_{\text{X}}^{a}=0,\quad\delta_{2}\xi_{\text{X}}^{a}=0.
\end{align}
Since when the Lie derivative acting on the forms, one has $\left[\mathscr{L}_{\mathcal{X}},\mathscr{L}_{\mathcal{Y}}\right]=\mathscr{L}_{\left[\mathcal{X},\mathcal{Y}\right]}$, so that
\begin{align}
\delta_{1}\delta_{2}\boldsymbol{N}^{\text{X}}-\delta_{2}\delta_{1}\boldsymbol{N}^{\text{X}} & =-\delta_{1}\left(\mathscr{L}_{\zeta_{\text{X}}}\boldsymbol{N}^{\text{X}}\right)+\delta_{2}\left(\mathscr{L}_{\xi_{\text{X}}}\boldsymbol{N}^{\text{X}}\right)\nonumber \\
 & =-\mathscr{L}_{\zeta_{\text{X}}}\delta_{1}\boldsymbol{N}^{\text{X}}+\mathscr{L}_{\xi_{\text{X}}}\delta_{2}\boldsymbol{N}^{\text{X}}\nonumber \\
 & =-\left[\mathscr{L}_{\xi_{\text{X}}},\mathscr{L}_{\zeta_{\text{X}}}\right]\boldsymbol{N}^{\text{X}}\nonumber \\
 & =-\mathscr{L}_{\left[\xi_{\text{X}},\zeta_{\text{X}}\right]}\boldsymbol{N}^{\text{X}},
\end{align}
and similarly
\begin{equation}
\delta_{1}\delta_{2}s-\delta_{2}\delta_{1}s=-\mathscr{L}_{\left[\xi_{\text{e}},\zeta_{\text{e}}\right]}s,
\end{equation}
so that we conclude that the perturbation $\delta_{1}\delta_{2}\phi-\delta_{2}\delta_{1}\phi=\left(\delta g_{ab}=0,\delta  A_{a}=0,\left[\xi_{\text{X}},\zeta_{\text{X}}\right]^{a}\right)$. 

Thus the pre-symplectic form Eq. \eqref{pre-symplectic form} is given by
\begin{align}
 & W_{AB}\delta_{1}\phi^{A}\delta_{2}\phi^{B}\nonumber \\
= & \int_{\Sigma}\left[\delta_{1}\boldsymbol{\theta}\left(\phi;\delta_{2}\phi\right)-\delta_{2}\boldsymbol{\theta}\left(\phi;\delta_{1}\phi\right)-\boldsymbol{\theta}\left(\phi;\delta_{1}\delta_{2}\phi-\delta_{2}\delta_{1}\phi\right)\right]\nonumber \\
= & \int_{\Sigma}\left(\delta_{1}\boldsymbol{\pi}^{ij}\delta_{2}h_{ij}-\delta_{2}\boldsymbol{\pi}^{ij}\delta_{1}h_{ij}\right)\nonumber \\
 & +\int_{\Sigma}\left(\delta_{1}\boldsymbol{\Pi}^{i}\delta_{2}A_{i}-\delta_{2}\boldsymbol{\Pi}^{i}\delta_{1}A_{i}\right)\nonumber \\
 & +\int_{\Sigma}\sum_{\text{X}}\left(\zeta_{\text{X}}^{a}\delta_{1}\boldsymbol{P}_{a}^{\text{X}}-\xi_{\text{X}}^{a}\delta_{2}\boldsymbol{P}_{a}^{\text{X}}-\left[\xi_{\text{X}},\zeta_{\text{X}}\right]^{a}\boldsymbol{P}_{a}^{\text{X}}\right),\label{explicity presymplectic form}
\end{align}
where we have used the well known expression \cite{Burnett:1990tww} for the gravitational and electromagnetic parts of symplectic current, and
\begin{equation}
\boldsymbol{\pi}^{ij}=(K^{ij}-Kh^{ij})\hat{\boldsymbol{\epsilon}},\quad\boldsymbol{\Pi}^{i}=\nu_{a}F^{ai}\hat{\boldsymbol{\epsilon}},
\end{equation}
with $\hat{\boldsymbol{\epsilon}}=\nu\cdot\boldsymbol{\epsilon}$ being the induced volume 3-form on $\Sigma$.

\subsection{Speed of sound and causal behavior}

\label{appendix}

The purpose of this subsection is to determine the propagation speeds and polarization directions of sound wave fronts in our superconducting-superfluid star. A convenient method is the so-called Hadamard technique (introduced long ago by Hadamard \cite{hadamard1903leçons} and then generalized to general-relativistic elasticity theory \cite{Rayner：1963},\cite{bennoun1965etude}) of investigating the characteristic hypersurfaces of possible discontinuity in a partial differential system. By using such a method, Carter and Langlois succeeded in calculating the first and second sound in a relativistic two-constituent superfluid \cite{Carter:1995if}. The method works by considering the first order case in which the algebraically related variables $n_{\text{X}}^{a}$, $s$, and $\mu_{a}^{\text{X}}$ are themselves continuous but have space-time derivatives that are weakly discontinuous across some characteristic hypersurface with tangent direction specified by some normal covector $\lambda_{a}$ say. For any component $\phi$, the discontinuities $\lceil\nabla_{a}\phi\rfloor$ in its gradient components will have to be proportional to the normal $\lambda_{a}$, i.e., we shall have $\lceil\nabla_{a}\phi\rfloor=\hat{\phi}\lambda_{a}$ for some scalar $\hat{\phi}$. Applying this to the relevant variables in the present case, we shall have
\begin{equation}
    \lceil\nabla_{a}\mu_{b}^{\text{X}}\rfloor=\hat{\mu}_{b}^{\text{X}}\lambda_{a},\quad \lceil\nabla_{a}n_{\text{X}}^{b}\rfloor=\hat{n}_{\text{X}}^{b}\lambda_{a},\quad \lceil\nabla_{a}s\rfloor=\hat{s}\lambda_{a},
\end{equation}
for some scalar $\hat{s}$, vectors $\hat{n}_{\text{X}}^{a}$, and covectors $\hat{\mu}_{a}^{\text{X}}$ on the hypersurface. The resulting discontinuities in the set of conservation laws Eq. \eqref{conservation of n=000026s} and equations of motion Eqs. \eqref{eomn}-\eqref{eome} will therefore given by\footnote{Note that the metric $g_{ab}$ and the electromagnetic tensor $F_{ab}$ are assumed to be continuous on the hypersurface}
\begin{align}
\hat{n}_{\text{X}}^{a}\lambda_{a} & =0,\quad u_{\text{e}}^{a}\lambda_{a}\hat{s}=0,\nonumber \\
n_{\Upsilon}^{a}\lambda_{[a}\hat{\mu}_{b]}^{\Upsilon} & =0,\quad n_{\text{e}}^{a}\lambda_{[a}\hat{\mu}_{b]}^{\text{e}}-Tn_{\text{e}}\hat{s}\lambda_{b}=0,\label{discontinuity equations}
\end{align}
where the chemical indices $\Upsilon=\text{n},\text{p}$ will indicate the ``super'' constituents.

In view of the algebraic relationship Eq. \eqref{effective momentum}, the discontinuity covectors $\hat{\mu}_{a}^{\text{X}}$ will not be independent of the corresponding discontinuity vectors $n_{\text{X}}^{a}$. Actually, since
\begin{widetext}
\begin{align}
\delta\mu_{a}^{\text{X}}= & \bigg(-2\frac{\partial\Lambda_{\text{M}}}{\partial n_{\text{X}}^{2}}g_{ab}+4\frac{\partial^{2}\Lambda_{\text{M}}}{\partial n_{\text{X}}^{2}\partial n_{\text{X}}^{2}}n_{\text{X}a}n_{\text{X}b}+4\sum_{\text{Y}\neq\text{X}}\frac{\partial^{2}\Lambda_{\text{M}}}{\partial n_{\text{X}}^{2}\partial x_{\text{XY}}^{2}}n_{\text{Y}(a}n_{\text{X}b)}+\sum_{\text{Z}\neq\text{X}}\sum_{\text{Y}\neq\text{X}}\frac{\partial^{2}\Lambda_{\text{M}}}{\partial x_{\text{XZ}}^{2}\partial x_{\text{XY}}^{2}}n_{\text{Y}a}n_{\text{Z}b}\bigg)\delta n_{\text{X}}^{b}\nonumber \\
 & +\sum_{\text{Y}\neq\text{X}}\bigg(-\frac{\partial\Lambda_{\text{M}}}{\partial x_{\text{XY}}^{2}}g_{ab}+4\frac{\partial^{2}\Lambda_{\text{M}}}{\partial n_{\text{Y}}^{2}\partial n_{\text{X}}^{2}}n_{\text{X}a}n_{\text{Y}b}+2\sum_{\text{Z}\neq\text{Y}}\frac{\partial^{2}\Lambda_{\text{M}}}{\partial x_{\text{YZ}}^{2}\partial n_{\text{X}}^{2}}n_{\text{X}a}n_{\text{Z}b}\nonumber \\
 & \qquad\qquad+2\sum_{\text{Z}\neq\text{X}}\frac{\partial^{2}\Lambda_{\text{M}}}{\partial n_{\text{Y}}^{2}\partial x_{\text{XZ}}^{2}}n_{\text{Z}a}n_{\text{Y}b}+\sum_{\text{Z}\neq\text{Y}}\sum_{\text{W}\neq\text{X}}\frac{\partial^{2}\Lambda_{\text{M}}}{\partial x_{\text{YZ}}^{2}\partial x_{\text{XW}}^{2}}n_{\text{W}a}n_{\text{Z}b}\bigg)\delta n_{\text{Y}}^{b}\nonumber \\
 & +\bigg(-2\frac{\partial\Lambda_{\text{M}}}{\partial n_{\text{X}}^{2}}n_{\text{X}}^{b}\delta_{a}^{c}-\sum_{\text{Y}\neq\text{X}}\frac{\partial\Lambda_{\text{M}}}{\partial x_{\text{XY}}^{2}}n_{\text{Y}}^{b}\delta_{a}^{c}+2\sum_{\text{Y}}\frac{\partial^{2}\Lambda_{\text{M}}}{\partial n_{\text{Y}}^{2}\partial n_{\text{X}}^{2}}n_{\text{X}a}n_{\text{Y}}^{b}n_{\text{Y}}^{c}+2\sum_{\text{Y}\neq\text{Z}}\frac{\partial^{2}\Lambda_{\text{M}}}{\partial x_{\text{YZ}}^{2}\partial n_{\text{X}}^{2}}n_{\text{X}a}n_{\text{Y}}^{b}n_{\text{Z}}^{c}\nonumber \\
 & \qquad+\sum_{\text{Z}}\sum_{\text{Y}\neq\text{X}}\frac{\partial^{2}\Lambda_{\text{M}}}{\partial n_{\text{Z}}^{2}\partial x_{\text{XY}}^{2}}n_{\text{Y}a}n_{\text{Z}}^{b}n_{\text{Z}}^{c}+\sum_{\text{Z}\neq\text{W}}\sum_{\text{Y}\neq\text{X}}\frac{\partial^{2}\Lambda_{\text{M}}}{\partial x_{\text{WZ}}^{2}\partial x_{\text{XY}}^{2}}n_{\text{Y}a}n_{\text{Z}}^{b}n_{\text{W}}^{c}\bigg)\delta g_{bc}\nonumber \\
 & +\bigg(-2\frac{\partial^{2}\Lambda_{\text{M}}}{\partial s\partial n_{\text{X}}^{2}}n_{\text{X}a}-\sum_{\text{Y}\neq\text{X}}\frac{\partial^{2}\Lambda_{\text{M}}}{\partial s\partial x_{\text{XY}}^{2}}n_{\text{Y}a}\bigg)\delta s,\label{delta of effective momentum}
\end{align}
then it implies that the discontinuity amplitudes in Eq. \eqref{discontinuity equations} will correspondingly be related by
\begin{equation}
\hat{\mu}_{a}^{\text{X}}=\sum_{\text{Y}}\mathscr{P}_{ab}^{\text{XY}}\hat{n}_{\text{Y}}^{b}+\mathscr{P}_{a}^{\text{Xs}}\hat{s},
\end{equation}
where the explicit expressions are

\begin{align}
\mathscr{P}_{ab}^{\text{nn}}= & \mathscr{D}g_{ab}-2\frac{\partial\mathcal{D}}{\partial n_{\text{n}}^{2}}n_{\text{n}a}n_{\text{n}b}-4\frac{\partial\mathscr{A}}{\partial n_{\text{n}}^{2}}n_{\text{p}(a}n_{\text{n}b)}-4\frac{\partial\mathscr{B}}{\partial n_{\text{n}}^{2}}n_{\text{e}(a}n_{\text{n}b)}-\frac{\partial\mathscr{A}}{\partial x_{\text{np}}^{2}}n_{\text{p}a}n_{\text{p}b}-\frac{\partial\mathscr{B}}{\partial x_{\text{ne}}^{2}}n_{\text{e}a}n_{\text{e}b}-2\frac{\partial\mathscr{A}}{\partial x_{\text{ne}}^{2}}n_{\text{p}(a}n_{\text{e}b)},\nonumber \\
\mathscr{P}_{ab}^{\text{pp}}= & \mathscr{E}g_{ab}-2\frac{\partial\mathscr{E}}{\partial n_{\text{p}}^{2}}n_{\text{p}a}n_{\text{p}b}-4\frac{\partial\mathscr{A}}{\partial n_{\text{p}}^{2}}n_{\text{n}(a}n_{\text{p}b)}-4\frac{\partial\mathscr{C}}{\partial n_{\text{p}}^{2}}n_{\text{e}(a}n_{\text{p}b)}-\frac{\partial\mathscr{A}}{\partial x_{\text{np}}^{2}}n_{\text{n}a}n_{\text{n}b}-\frac{\partial\mathscr{C}}{\partial x_{\text{pe}}^{2}}n_{\text{e}a}n_{\text{e}b}-2\frac{\partial\mathscr{A}}{\partial x_{\text{pe}}^{2}}n_{\text{n}(a}n_{\text{e}b)},\nonumber \\
\mathscr{P}_{ab}^{\text{ee}}= & \mathscr{F}g_{ab}-2\frac{\partial\mathscr{F}}{\partial n_{\text{e}}^{2}}n_{\text{e}a}n_{\text{e}b}-4\frac{\partial\mathscr{B}}{\partial n_{\text{e}}^{2}}n_{\text{n}(a}n_{\text{e}b)}-4\frac{\partial\mathscr{C}}{\partial n_{\text{e}}^{2}}n_{\text{p}(a}n_{\text{e}b)}-\frac{\partial\mathscr{B}}{\partial x_{\text{ne}}^{2}}n_{\text{n}a}n_{\text{n}b}-\frac{\partial\mathscr{C}}{\partial x_{\text{pe}}^{2}}n_{\text{p}a}n_{\text{p}b}-2\frac{\partial\mathscr{B}}{\partial x_{\text{pe}}^{2}}n_{\text{n}(a}n_{\text{p}b)},\nonumber \\
\mathscr{P}_{ab}^{\text{np}}=\mathscr{P}_{ba}^{\text{pn}}= & \mathscr{A}g_{ab}-2\frac{\partial\mathscr{E}}{\partial n_{\text{n}}^{2}}n_{\text{n}a}n_{\text{p}b}-2\frac{\partial\mathscr{A}}{\partial n_{\text{n}}^{2}}n_{\text{n}a}n_{\text{n}b}-2\frac{\partial\mathscr{C}}{\partial n_{\text{n}}^{2}}n_{\text{n}a}n_{\text{e}b}-2\frac{\partial\mathscr{A}}{\partial n_{\text{p}}^{2}}n_{\text{p}a}n_{\text{p}b}-2\frac{\partial\mathscr{B}}{\partial n_{\text{p}}^{2}}n_{\text{e}a}n_{\text{p}b}\nonumber \\
 & -\frac{\partial\mathscr{A}}{\partial x_{\text{np}}^{2}}n_{\text{p}a}n_{\text{n}b}-\frac{\partial\mathscr{A}}{\partial x_{\text{pe}}^{2}}n_{\text{p}a}n_{\text{e}b}-\frac{\partial\mathscr{A}}{\partial x_{\text{ne}}^{2}}n_{\text{e}a}n_{\text{n}b}-\frac{\partial\mathscr{B}}{\partial x_{\text{pe}}^{2}}n_{\text{e}a}n_{\text{e}b},\nonumber \\
\mathscr{P}_{ab}^{\text{ne}}=\mathscr{P}_{ba}^{\text{en}}= & \mathscr{B}g_{ab}-2\frac{\partial\mathscr{F}}{\partial n_{\text{n}}^{2}}n_{\text{n}a}n_{\text{e}b}-2\frac{\partial\mathscr{B}}{\partial n_{\text{n}}^{2}}n_{\text{n}a}n_{\text{n}b}-2\frac{\partial\mathscr{C}}{\partial n_{\text{n}}^{2}}n_{\text{n}a}n_{\text{p}b}-2\frac{\partial\mathscr{A}}{\partial n_{\text{e}}^{2}}n_{\text{p}a}n_{\text{e}b}-2\frac{\partial\mathscr{B}}{\partial n_{\text{e}}^{2}}n_{\text{e}a}n_{\text{e}b}\nonumber \\
 & -\frac{\partial\mathscr{A}}{\partial x_{\text{ne}}^{2}}n_{\text{p}a}n_{\text{n}b}-\frac{\partial\mathscr{A}}{\partial x_{\text{pe}}^{2}}n_{\text{p}a}n_{\text{p}b}-\frac{\partial\mathscr{B}}{\partial x_{\text{ne}}^{2}}n_{\text{e}a}n_{\text{n}b}-\frac{\partial\mathscr{B}}{\partial x_{\text{pe}}^{2}}n_{\text{e}a}n_{\text{p}b},\nonumber \\
\mathscr{P}_{ab}^{\text{pe}}=\mathscr{P}_{ba}^{\text{ep}}= & \mathscr{C}g_{ab}-2\frac{\partial\mathscr{F}}{\partial n_{\text{p}}^{2}}n_{\text{p}a}n_{\text{e}b}-2\frac{\partial\mathscr{B}}{\partial n_{\text{p}}^{2}}n_{\text{p}a}n_{\text{n}b}-2\frac{\partial\mathscr{C}}{\partial n_{\text{p}}^{2}}n_{\text{p}a}n_{\text{p}b}-2\frac{\partial\mathscr{A}}{\partial n_{\text{e}}^{2}}n_{\text{n}a}n_{\text{e}b}-2\frac{\partial\mathscr{C}}{\partial n_{\text{e}}^{2}}n_{\text{e}a}n_{\text{e}b}\nonumber \\
 & -\frac{\partial\mathscr{A}}{\partial x_{\text{ne}}^{2}}n_{\text{n}a}n_{\text{n}b}-\frac{\partial\mathscr{A}}{\partial x_{\text{pe}}^{2}}n_{\text{n}a}n_{\text{p}b}-\frac{\partial\mathscr{B}}{\partial x_{\text{pe}}^{2}}n_{\text{e}a}n_{\text{n}b}-\frac{\partial\mathscr{C}}{\partial x_{\text{pe}}^{2}}n_{\text{e}a}n_{\text{p}b},\nonumber \\
\mathscr{P}_{a}^{\text{Xs}}= & -2\frac{\partial^{2}\Lambda_{\text{M}}}{\partial s\partial n_{\text{X}}^{2}}n_{\text{X}a}-\sum_{\text{Y}\neq\text{X}}\frac{\partial^{2}\Lambda_{\text{M}}}{\partial s\partial x_{\text{XY}}^{2}}n_{\text{Y}a}.\label{coefficients}
\end{align}
\end{widetext}

With respect to a rest frame determined by unit flow $u_{\text{e}}^{a}$ of electron, the velocity $v$ say of propagation in the direction of the orthogonal unit spacelike vector,  
\begin{equation}
\gamma^{a}=\frac{1}{\sqrt{q_{\text{e}}^{cd}\varphi_{c}\varphi_{d}}}q_{\text{e}}^{ab}\varphi_{b},\label{gamma}
\end{equation}
 will be given for a suitably normalized $\lambda_{a}$ by
\begin{equation}
\lambda_{a}=-vu_{\text{e}a}+\gamma_{a}.\label{normal of discontinuity}
\end{equation}
And as we consider the superconducting-superfluid star background, the circular flow condition Eq. \eqref{circular flow condition} implies that we have the following decomposition for the flows
\begin{equation}
u_{\text{X}}^{a}=-u_{\text{e}}^{a}u_{\text{X}}^{b}u_{\text{e}b}+\gamma^{a}u_{\text{X}}^{b}\gamma_{b}.
\end{equation}
It can now easily be seen from Eq. \eqref{discontinuity equations} that there can be no transverse modes, i.e., for the vector $\theta^{a}$ such that $u_{\text{e}}^{a}\theta_{a}=\gamma^{a}\theta_{a}=0$, one must have
\begin{equation}
\theta^{a}\hat{\mu}_{a}^{\text{X}}=0,\quad\Rightarrow\quad\theta_{a}\hat{n}_{\text{X}}^{a}=0,
\end{equation}
The discontinuity of conservation law of entropy means that $\hat{s}=0$, and the discontinuity of conservation of particle number means that
\begin{equation}
vu_{\text{e}a}\hat{n}_{\text{X}}^{a}=\gamma_{a}\hat{n}_{\text{X}}^{a}.
\end{equation}
Consequently, 
\begin{equation}
g_{ab}\left(-u_{\text{e}}^{a}+v\gamma^{a}\right)\hat{n}_{\text{X}}^{b}=g_{ab}\left(-u_{\text{e}}^{a}+v\gamma^{a}\right)\left(-u_{\text{e}}^{b}+v\gamma^{b}\right)u_{\text{e}c}\hat{n}_{\text{X}}^{c},
\end{equation}
and
\begin{align}
n_{\text{X}a}\hat{n}_{\text{Y}}^{a} & =\left(-u_{\text{e}b}n_{\text{X}a}u_{\text{e}}^{a}+\gamma_{b}n_{\text{X}a}\gamma^{a}\right)\hat{n}_{\text{Y}}^{b}\nonumber \\
 & =n_{\text{X}a}\left(-u_{\text{e}}^{a}+v\gamma^{a}\right)u_{\text{e}b}\hat{n}_{\text{Y}}^{b},
\end{align}
hence the longitudinal modes of the characteristic equation Eq. \eqref{discontinuity equations} will take the form
\begin{align}
 & 2\gamma^{a}u_{\text{X}}^{b}\lambda_{[a}\hat{\mu}_{b]}^{\text{X}}\nonumber \\
= & u_{\text{X}}^{a}\hat{\mu}_{a}^{\text{X}}-\left(-vu_{\text{e}b}u_{\text{X}}^{b}+\gamma_{b}u_{\text{X}}^{b}\right)\gamma^{a}\hat{\mu}_{a}^{\text{X}}\nonumber \\
= & u_{\text{e}d}u_{\text{X}}^{d}\left(-u_{\text{e}}^{a}+v\gamma^{a}\right)\hat{\mu}_{a}^{\text{X}}\nonumber \\
= & u_{\text{e}d}u_{\text{X}}^{d}\sum_{\text{Y}}\mathscr{P}_{ab}^{\text{XY}}\left(u_{\text{e}}^{a}-v\gamma^{a}\right)\left(u_{\text{e}}^{b}-v\gamma^{b}\right)u_{\text{e}c}\hat{n}_{\text{Y}}^{c}.
\end{align}
The resulting eigenvalue equations for the propagation velocity $v$ is 
\begin{equation}
\det\left[\mathscr{P}_{ab}^{\text{XY}}\left(u_{\text{e}}^{a}-v\gamma^{a}\right)\left(u_{\text{e}}^{b}-v\gamma^{b}\right)\right]=0.\label{eigenvalue equation}
\end{equation}
where the determinant is taking with the row index $\text{X}$ and the column index $\text{Y}$. 

The solutions $v$ of Eq. \eqref{eigenvalue equation} should subject to the causal condition that neither root should exceed the speed of light, i.e., $v^{2}\leq1$. This condition will give the constraint to the coefficients of the eigenvalue equation Eq. \eqref{eigenvalue equation}, but since they are not necessary for our calculation in Sec. \ref{subsec:Phase-space}, so we will not write out the explicit form of the constraint here.

\subsection{Phase space}\label{subsec:Phase-space}

As introduced in Sec. \ref{sec3}, the $W_{AB}$ in Eq. \eqref{pre-symplectic form} is the pre-symplectic form. To make it become symplectic form and then construct the phase space, we need factor the space of all fields $\phi=\left(g_{ab},A_{a},\chi_{\text{X}}\right)$ on spacetime by the degeneracy of $W_{AB}$, in other words, the phase space is the space of equivalence classes of field configurations, where two field configuration are equivalent if they lie on an orbit of degeneracy directions of $W$. To proceed, it will be useful to introduce the space of fiducial flowlines $\Sigma_{\text{X}}^{\prime}$ of each constituent, defined as the space of the integral curves of a non-vanishing $u_{\text{X}}^{\prime a}$ with $\iota_{u_{\text{X}}^{\prime}}\boldsymbol{N}^{\text{X}\prime}=0$, and we further introduce the diffeomorphism $\sigma_{\text{X}}$ from $\Sigma_{\text{X}}^{\prime}$ to $\Sigma$ obtained by intersecting the images of the fiducial flowlines under $\chi$ introduced in Sec. \ref{subsec2.2} with $\Sigma$. Next, we will find out the degeneracy of $W$ below.

It is clear from Eq. \eqref{explicity presymplectic form} that $W$ dependents at most on the following quantities on $\Sigma$: $\delta h_{ij}$, $\delta\boldsymbol{\pi}^{ij}$, $\delta\alpha$ (the perturbed lapse), $\delta\beta_{a}$ (the perturbed shift), $\delta A_{a}$, $\delta\boldsymbol{\Pi}^{i}$, $\xi_{\text{X}}^{a}$, and the normal derivative of $\xi_{\text{X}}^{a}$. Note that the pre-symplectic form is linear on the perturbation of fields, hence for the perturbation $\delta_{2}\phi=\left(\delta_{2}g_{ab},\delta_{2}A_{a},\zeta_{\text{X}}^{a}\right)$, let us do a decomposition $\delta_{2}\phi=\delta_{2}^{\prime}\phi+\delta_{2}^{\text{M}}\phi$ with $\delta_{2}^{\prime}\phi=\left(\delta_{2}g_{ab},\delta_{2}A_{a},0\right)$ and $\delta_{2}^{\text{M}}\phi=\left(0,0,\zeta_{\text{X}}^{a}\right)$. With such a decomposition, we have
\begin{align}
  &W\left(\delta_{1}\phi,\delta_{2}\phi\right)\nonumber\\
  = & W\left(\delta_{1}\phi,\delta_{2}^{\prime}\phi\right)+W\left(\delta_{1}\phi,\delta_{2}^{\text{M}}\phi\right)\nonumber \\
= & \int_{\Sigma}\left(\delta_{1}\boldsymbol{\pi}^{ij}\delta_{2}h_{ij}-\delta_{2}\boldsymbol{\pi}^{ij}\delta_{1}h_{ij}\right)\nonumber\\
 &+\int_{\Sigma}\left(\delta_{1}\boldsymbol{\Pi}^{i}\delta_{2}A_{i}-\delta_{2}\boldsymbol{\Pi}^{i}\delta_{1}A_{i}\right)-\int_{\Sigma}\sum_{\text{X}}\xi_{\text{X}}^{a}\delta_{2}^{\prime}\boldsymbol{P}_{a}^{\text{X}}\nonumber\nonumber \\
 & +\int_{\Sigma}\sum_{\text{X}}\left(\zeta_{\text{X}}^{a}\delta_{1}\boldsymbol{P}_{a}^{\text{X}}-\xi_{\text{X}}^{a}\delta_{2}^{\text{M}}\boldsymbol{P}_{a}^{\text{X}}-\left[\xi_{\text{X}},\zeta_{\text{X}}\right]^{a}\boldsymbol{P}_{a}^{\text{X}}\right).
\end{align}
Note that the third and fourth lines only occur the quantities
$\delta_{2}h_{ij}$, $\delta_{2}\boldsymbol{\pi}^{ij}$, $\delta_{2}\alpha$,
$\delta_{2}\beta_{a}$, $\delta_{2}A_{a}$, $\delta_{2}\boldsymbol{\Pi}^{i}$,
while the fifth line only occurs $\zeta_{\text{X}}^{a}$ and its derivative. According to
\begin{equation}
\delta\left(\boldsymbol{\epsilon}n_{\text{X}}^{a}\right)=-\mathscr{L}_{\xi_{\text{X}}}\left(\boldsymbol{\epsilon}n_{\text{X}}^{a}\right),
\end{equation}
one has
\begin{align}
\xi_{\text{X}}^{a}\delta_{2}^{\prime}\boldsymbol{P}_{a}^{\text{X}} 
 & =2\delta_{2}^{\prime}\pi_{a}^{\text{X}}n_{\text{X}}^{[b}\xi_{\text{X}}^{a]}\boldsymbol{\epsilon}_{bpqr}\nonumber=2\left(\delta_{2}^{\prime}\mu_{a}^{\text{X}}+e^{\text{X}}\delta_{2}A_{a}\right)n_{\text{X}}^{[b}\xi_{\text{X}}^{a]}\boldsymbol{\epsilon}_{bpqr}\nonumber \\
 & =2\left[\left(\mathscr{Q}^{\text{X}}\right)_{a}^{\ bc}\delta_{2}g_{bc}+e^{\text{X}}\delta_{2}A_{a}\right]n_{\text{X}}^{[d}\xi_{\text{X}}^{a]}\boldsymbol{\epsilon}_{dpqr},
\end{align}
where 
\begin{widetext}
\begin{align}
(\mathcal{\mathscr{Q}}^{\text{n}})_{a}^{\ bc}= & -\left(\frac{1}{2}\mathscr{D}g^{bc}+\frac{\partial\mathscr{D}}{\partial n_{\text{n}}^{2}}n_{\text{n}}^{2}q_{\text{n}}^{bc}+\frac{\partial\mathscr{E}}{\partial n_{\text{n}}^{2}}n_{\text{p}}^{2}q_{\text{p}}^{bc}+\frac{\partial\mathscr{F}}{\partial n_{\text{n}}^{2}}n_{\text{e}}^{2}q_{\text{e}}^{bc}+2\frac{\partial\mathscr{A}}{\partial n_{\text{n}}^{2}}x_{\text{np}}^{2}q_{\text{np}}^{bc}+2\frac{\partial\mathscr{B}}{\partial n_{\text{n}}^{2}}x_{\text{ne}}^{2}q_{\text{ne}}^{bc}+2\frac{\partial\mathscr{C}}{\partial n_{\text{n}}^{2}}x_{\text{pe}}^{2}q_{\text{pe}}^{bc}\right)n_{\text{n}a}+\mathscr{D}n_{\text{n}}^{(b}\delta_{a}^{c)}\nonumber \\
 & -\left(\frac{1}{2}\mathscr{A}g^{bc}+\frac{\partial\mathscr{A}}{\partial n_{\text{n}}^{2}}n_{\text{n}}^{2}q_{\text{n}}^{bc}+\frac{\partial\mathscr{A}}{\partial n_{\text{p}}^{2}}n_{\text{p}}^{2}q_{\text{p}}^{bc}+\frac{\partial\mathscr{A}}{\partial n_{\text{e}}^{2}}n_{\text{e}}^{2}q_{\text{e}}^{bc}+\frac{\partial\mathscr{A}}{\partial x_{\text{np}}^{2}}x_{\text{np}}^{2}q_{\text{np}}^{bc}+\frac{\partial\mathscr{A}}{\partial x_{\text{ne}}^{2}}x_{\text{ne}}^{2}q_{\text{ne}}^{bc}+\frac{\partial\mathscr{A}}{\partial x_{\text{pe}}^{2}}x_{\text{pe}}^{2}q_{\text{pe}}^{bc}\right)n_{\text{p}a}+\mathscr{A}n_{\text{p}}^{(b}\delta_{a}^{c)}\nonumber \\
 & -\left(\frac{1}{2}\mathscr{B}g^{bc}+\frac{\partial\mathscr{B}}{\partial n_{\text{n}}^{2}}n_{\text{n}}^{2}q_{\text{n}}^{bc}+\frac{\partial\mathscr{B}}{\partial n_{\text{p}}^{2}}n_{\text{p}}^{2}q_{\text{p}}^{bc}+\frac{\partial\mathscr{B}}{\partial n_{\text{e}}^{2}}n_{\text{e}}^{2}q_{\text{e}}^{bc}+\frac{\partial\mathscr{B}}{\partial x_{\text{np}}^{2}}x_{\text{np}}^{2}q_{\text{np}}^{bc}+\frac{\partial\mathscr{B}}{\partial x_{\text{ne}}^{2}}x_{\text{ne}}^{2}q_{\text{ne}}^{bc}+\frac{\partial\mathscr{B}}{\partial x_{\text{pe}}^{2}}x_{\text{pe}}^{2}q_{\text{pe}}^{bc}\right)n_{\text{e}a}+\mathscr{B}n_{\text{e}}^{(b}\delta_{a}^{c)},\nonumber \\
(\mathcal{\mathscr{Q}}^{\text{p}})_{a}^{\ bc}= & -\left(\frac{1}{2}\mathscr{A}g^{bc}+\frac{\partial\mathscr{A}}{\partial n_{\text{n}}^{2}}n_{\text{n}}^{2}q_{\text{n}}^{bc}+\frac{\partial\mathscr{A}}{\partial n_{\text{p}}^{2}}n_{\text{p}}^{2}q_{\text{p}}^{bc}+\frac{\partial\mathscr{A}}{\partial n_{\text{e}}^{2}}n_{\text{e}}^{2}q_{\text{e}}^{bc}+\frac{\partial\mathscr{A}}{\partial x_{\text{np}}^{2}}x_{\text{np}}^{2}q_{\text{np}}^{bc}+\frac{\partial\mathscr{A}}{\partial x_{\text{ne}}^{2}}x_{\text{ne}}^{2}q_{\text{ne}}^{bc}+\frac{\partial\mathscr{A}}{\partial x_{\text{pe}}^{2}}x_{\text{pe}}^{2}q_{\text{pe}}^{bc}\right)n_{\text{n}a}+\mathscr{A}n_{\text{n}}^{(b}\delta_{a}^{c)}\nonumber \\
 & -\left(\frac{1}{2}\mathscr{E}g^{bc}+\frac{\partial\mathscr{D}}{\partial n_{\text{p}}^{2}}n_{\text{n}}^{2}q_{\text{n}}^{bc}+\frac{\partial\mathscr{E}}{\partial n_{\text{p}}^{2}}n_{\text{p}}^{2}q_{\text{p}}^{bc}+\frac{\partial\mathscr{F}}{\partial n_{\text{p}}^{2}}n_{\text{e}}^{2}q_{\text{e}}^{bc}+2\frac{\partial\mathscr{A}}{\partial n_{\text{p}}^{2}}x_{\text{np}}^{2}q_{\text{np}}^{bc}+2\frac{\partial\mathscr{B}}{\partial n_{\text{p}}^{2}}x_{\text{ne}}^{2}q_{\text{ne}}^{bc}+2\frac{\partial\mathscr{C}}{\partial n_{\text{p}}^{2}}x_{\text{pe}}^{2}q_{\text{pe}}^{bc}\right)n_{\text{p}a}+\mathscr{E}n_{\text{p}}^{(b}\delta_{a}^{c)}\nonumber \\
 & -\left(\frac{1}{2}\mathscr{C}g^{bc}+\frac{\partial\mathscr{C}}{\partial n_{\text{n}}^{2}}n_{\text{n}}^{2}q_{\text{n}}^{bc}+\frac{\partial\mathscr{C}}{\partial n_{\text{p}}^{2}}n_{\text{p}}^{2}q_{\text{p}}^{bc}+\frac{\partial\mathscr{C}}{\partial n_{\text{e}}^{2}}n_{\text{e}}^{2}q_{\text{e}}^{bc}+\frac{\partial\mathscr{C}}{\partial x_{\text{np}}^{2}}x_{\text{np}}^{2}q_{\text{np}}^{bc}+\frac{\partial\mathscr{C}}{\partial x_{\text{ne}}^{2}}x_{\text{ne}}^{2}q_{\text{ne}}^{bc}+\frac{\partial\mathscr{C}}{\partial x_{\text{pe}}^{2}}x_{\text{pe}}^{2}q_{\text{pe}}^{bc}\right)n_{\text{e}a}+\mathscr{C}n_{\text{e}}^{(b}\delta_{a}^{c)},\nonumber \\
(\mathcal{\mathscr{Q}}^{\text{e}})_{a}^{\ bc}= & -\left(\frac{1}{2}\mathscr{B}g^{bc}+\frac{\partial\mathscr{B}}{\partial n_{\text{n}}^{2}}n_{\text{n}}^{2}q_{\text{n}}^{bc}+\frac{\partial\mathscr{B}}{\partial n_{\text{p}}^{2}}n_{\text{p}}^{2}q_{\text{p}}^{bc}+\frac{\partial\mathscr{B}}{\partial n_{\text{e}}^{2}}n_{\text{e}}^{2}q_{\text{e}}^{bc}+\frac{\partial\mathscr{B}}{\partial x_{\text{np}}^{2}}x_{\text{np}}^{2}q_{\text{np}}^{bc}+\frac{\partial\mathscr{B}}{\partial x_{\text{ne}}^{2}}x_{\text{ne}}^{2}q_{\text{ne}}^{bc}+\frac{\partial\mathscr{B}}{\partial x_{\text{pe}}^{2}}x_{\text{pe}}^{2}q_{\text{pe}}^{bc}\right)n_{\text{n}a}+\mathscr{B}n_{\text{n}}^{(b}\delta_{a}^{c)}\nonumber \\
 & -\left(\frac{1}{2}\mathscr{C}g^{bc}+\frac{\partial\mathscr{C}}{\partial n_{\text{n}}^{2}}n_{\text{n}}^{2}q_{\text{n}}^{bc}+\frac{\partial\mathscr{C}}{\partial n_{\text{p}}^{2}}n_{\text{p}}^{2}q_{\text{p}}^{bc}+\frac{\partial\mathscr{C}}{\partial n_{\text{e}}^{2}}n_{\text{e}}^{2}q_{\text{e}}^{bc}+\frac{\partial\mathscr{C}}{\partial x_{\text{np}}^{2}}x_{\text{np}}^{2}q_{\text{np}}^{bc}+\frac{\partial\mathscr{C}}{\partial x_{\text{ne}}^{2}}x_{\text{ne}}^{2}q_{\text{ne}}^{bc}+\frac{\partial\mathscr{C}}{\partial x_{\text{pe}}^{2}}x_{\text{pe}}^{2}q_{\text{pe}}^{bc}\right)n_{\text{p}a}+\mathscr{C}n_{\text{p}}^{(b}\delta_{a}^{c)}\nonumber \\
 & -\left(\frac{1}{2}\mathscr{F}g^{bc}+\frac{\partial\mathscr{D}}{\partial n_{\text{e}}^{2}}n_{\text{n}}^{2}q_{\text{n}}^{bc}+\frac{\partial\mathscr{E}}{\partial n_{\text{e}}^{2}}n_{\text{p}}^{2}q_{\text{p}}^{bc}+\frac{\partial\mathscr{F}}{\partial n_{\text{e}}^{2}}n_{\text{e}}^{2}q_{\text{e}}^{bc}+2\frac{\partial\mathscr{A}}{\partial n_{\text{e}}^{2}}x_{\text{np}}^{2}q_{\text{np}}^{bc}+2\frac{\partial\mathscr{B}}{\partial n_{\text{e}}^{2}}x_{\text{ne}}^{2}q_{\text{ne}}^{bc}+2\frac{\partial\mathscr{C}}{\partial n_{\text{e}}^{2}}x_{\text{pe}}^{2}q_{\text{pe}}^{bc}\right)n_{\text{e}a}+\mathscr{F}n_{\text{e}}^{(b}\delta_{a}^{c)},
\end{align}
which can be read out by substituting 
\begin{equation}
\delta_{2}^{\prime}n_{\text{Y}}^{a}=-\frac{1}{2}n_{\text{Y}}^{a}g^{bc}\delta_{2}g_{bc},\quad\delta_{2}^{\prime}s=0,
\end{equation}
into Eq. \eqref{delta of effective momentum} and $q_{\text{XY}}^{ab}=g^{ab}+\frac{1}{x_{\text{XY}}^2}n_{\text{X}}^{(a}n_{\text{Y}}^{b)}$. Also
\begin{align}
 & \sum_{\text{X}}\left(\xi_{\text{X}}^{a}\delta_{2}^{\text{M}}\boldsymbol{P}_{a}^{\text{X}}+\left[\xi_{\text{X}},\zeta_{\text{X}}\right]^{a}\boldsymbol{P}_{a}^{\text{X}}\right)\nonumber \\
= & \sum_{\text{X}}\left[\delta_{2}^{\text{M}}\left(2\pi_{a}^{\text{X}}n_{\text{X}}^{[d}\xi_{\text{X}}^{a]}\boldsymbol{\epsilon}_{dpqr}\right)+\left(\xi_{\text{X}}^{b}\nabla_{b}\zeta_{\text{X}}^{a}-\zeta_{\text{X}}^{b}\nabla_{b}\xi_{\text{X}}^{a}\right)\boldsymbol{P}_{a}^{\text{X}}\right]\nonumber \\
= & \sum_{\text{X}}\left[2\left(\delta_{2}^{\text{M}}\mu_{a}^{\text{X}}\right)n_{\text{X}}^{[d}\xi_{\text{X}}^{a]}\boldsymbol{\epsilon}_{dpqr}+2\pi_{a}^{\text{X}}\left(\delta_{2}^{\text{M}}n_{\text{X}}^{[d}\right)\xi_{\text{X}}^{a]}\boldsymbol{\epsilon}_{dpqr}+\left(\xi_{\text{X}}^{b}\nabla_{b}\zeta_{\text{X}}^{a}-\zeta_{\text{X}}^{b}\nabla_{b}\xi_{\text{X}}^{a}\right)\boldsymbol{P}_{a}^{\text{X}}\right]\nonumber \\
= & \sum_{\text{X}}\left[2\left(\sum_{\text{Y}}\mathscr{P}_{ab}^{\text{XY}}\delta_{2}^{\text{M}}n_{\text{Y}}^{b}+\mathscr{P}_{a}^{\text{Xs}}\delta_{2}^{\text{M}}s\right)n_{\text{X}}^{[d}\xi_{\text{X}}^{a]}\boldsymbol{\epsilon}_{dpqr}\right]\nonumber \\
 & +\sum_{\text{X}}\left[2\pi_{a}^{\text{X}}\left(-\zeta_{\text{X}}^{b}\nabla_{b}n_{\text{X}}^{[d}\xi_{\text{X}}^{a]}+n_{\text{X}}^{b}\nabla_{b}\zeta_{\text{X}}^{[d}\xi_{\text{X}}^{a]}-\nabla_{b}\zeta_{\text{X}}^{b}n_{\text{X}}^{[d}\xi_{\text{X}}^{a]}\right)\boldsymbol{\epsilon}_{dpqr}+\left(\xi_{\text{X}}^{b}\nabla_{b}\zeta_{\text{X}}^{a}-\zeta_{\text{X}}^{b}\nabla_{b}\xi_{\text{X}}^{a}\right)\boldsymbol{P}_{a}^{\text{X}}\right]\nonumber \\
= & \sum_{\text{X}}\left\{ 2\left[\sum_{\text{Y}}\mathscr{P}_{ab}^{\text{XY}}\left(-\zeta_{\text{Y}}^{c}\nabla_{c}n_{\text{Y}}^{b}+n_{\text{Y}}^{c}\nabla_{c}\zeta_{\text{Y}}^{b}-n_{\text{Y}}^{b}\nabla_{c}\zeta_{\text{Y}}^{c}\right)-\mathscr{P}_{a}^{\text{Xs}}\zeta_{\text{e}}^{b}\nabla_{b}s\right]n_{\text{X}}^{[d}\xi_{\text{X}}^{a]}\boldsymbol{\epsilon}_{dpqr}\right\} \nonumber \\
 & +\sum_{\text{X}}\left[2\pi_{a}^{\text{X}}\left(-\zeta_{\text{X}}^{b}\nabla_{b}n_{\text{X}}^{[d}\xi_{\text{X}}^{a]}+n_{\text{X}}^{b}\nabla_{b}\zeta_{\text{X}}^{[d}\xi_{\text{X}}^{a]}-\nabla_{b}\zeta_{\text{X}}^{b}n_{\text{X}}^{[d}\xi_{\text{X}}^{a]}\right)\boldsymbol{\epsilon}_{dpqr}+\left(\xi_{\text{X}}^{b}\nabla_{b}\zeta_{\text{X}}^{a}-\zeta_{\text{X}}^{b}\nabla_{b}\xi_{\text{X}}^{a}\right)\boldsymbol{P}_{a}^{\text{X}}\right]\nonumber \\
= & \sum_{\text{Y}}\left\{ 2\left[\sum_{\text{X}}\mathscr{P}_{ab}^{\text{YX}}\left(-\zeta_{\text{X}}^{c}\nabla_{c}n_{\text{X}}^{b}+n_{\text{X}}^{c}\nabla_{c}\zeta_{\text{X}}^{b}-n_{\text{X}}^{b}\nabla_{c}\zeta_{\text{X}}^{c}\right)-\mathscr{P}_{a}^{\text{Ys}}\zeta_{\text{e}}^{b}\nabla_{b}s\right]n_{\text{Y}}^{[d}\xi_{\text{Y}}^{a]}\boldsymbol{\epsilon}_{dpqr}\right\} \nonumber \\
 & +\sum_{\text{X}}\left[2\pi_{a}^{\text{X}}\left(-\zeta_{\text{X}}^{b}\nabla_{b}n_{\text{X}}^{[d}\xi_{\text{X}}^{a]}+n_{\text{X}}^{b}\nabla_{b}\zeta_{\text{X}}^{[d}\xi_{\text{X}}^{a]}-\nabla_{b}\zeta_{\text{X}}^{b}n_{\text{X}}^{[d}\xi_{\text{X}}^{a]}\right)\boldsymbol{\epsilon}_{dpqr}+\left(\xi_{\text{X}}^{b}\nabla_{b}\zeta_{\text{X}}^{a}-\zeta_{\text{X}}^{b}\nabla_{b}\xi_{\text{X}}^{a}\right)\boldsymbol{P}_{a}^{\text{X}}\right]\nonumber \\
= & 2\sum_{\text{X}}\nabla_{a}\zeta_{\text{X}}^{b}\left[\sum_{\text{Y}}\left(\mathscr{P}_{bc}^{\text{XY}}n_{\text{X}}^{a}-\mathscr{P}_{ec}^{\text{XY}}n_{\text{X}}^{e}\delta_{b}^{a}\right)n_{\text{Y}}^{[d}\xi_{\text{Y}}^{c]}\boldsymbol{\epsilon}_{dpqr}+n_{\text{X}}^{a}\xi_{\text{X}}^{c}\pi_{[c}^{\text{X}}\boldsymbol{\epsilon}_{b]pqr}-\delta_{b}^{a}\xi_{\text{X}}^{[c}n_{\text{X}}^{d]}\pi_{c}^{\text{X}}\boldsymbol{\epsilon}_{dpqr}+\xi_{\text{X}}^{a}n_{\text{X}}^{c}\pi_{[b}^{\text{X}}\boldsymbol{\epsilon}_{c]pqr}\right]\nonumber \\
 & -2\sum_{\text{Y}}\left[\sum_{\text{X}}\mathscr{P}_{bc}^{\text{XY}}\zeta_{\text{X}}^{a}\nabla_{a}n_{\text{X}}^{b}+\mathscr{P}_{c}^{\text{Ys}}\zeta_{\text{e}}^{a}\nabla_{a}s\right]n_{\text{Y}}^{[d}\xi_{\text{Y}}^{c]}\boldsymbol{\epsilon}_{dpqr}-2\sum_{\text{X}}\pi_{b}^{\text{X}}\zeta_{\text{X}}^{a}\nabla_{a}\left(n_{\text{X}}^{[d}\xi_{\text{X}}^{b]}\right)\boldsymbol{\epsilon}_{dpqr}\nonumber \\
= & 2\sum_{\text{X}}\nabla_{a}\zeta_{\text{X}}^{b}\left(\boldsymbol{\mathcal{P}}_{\text{X}}\right)_{\ b}^{a}-2\sum_{\text{X}}\zeta_{\text{X}}^{a}\boldsymbol{\mathcal{Q}}_{a}^{\text{X}},
\end{align}
where we have used the property that $\mathscr{P}_{ab}^{\text{XY}}=\mathscr{P}_{ba}^{\text{YX}}$ (see Eq. \eqref{coefficients}) in the sixth step, and when pull back onto $\Sigma$,
\begin{align}
\left(\boldsymbol{\mathcal{P}}_{\text{X}}\right)_{\ b}^{a} & =\sum_{\text{Y}}\left(\mathscr{P}_{bc}^{\text{XY}}n_{\text{X}}^{a}-\mathscr{P}_{ec}^{\text{XY}}n_{\text{X}}^{e}\delta_{b}^{a}\right)n_{\text{Y}}^{[d}\xi_{\text{Y}}^{c]}\boldsymbol{\epsilon}_{dpqr}+n_{\text{X}}^{a}\xi_{\text{X}}^{c}\pi_{[c}^{\text{X}}\boldsymbol{\epsilon}_{b]pqr}-\delta_{b}^{a}\xi_{\text{X}}^{[c}n_{\text{X}}^{d]}\pi_{c}^{\text{X}}\boldsymbol{\epsilon}_{dpqr}+\xi_{\text{X}}^{a}n_{\text{X}}^{c}\pi_{[b}^{\text{X}}\boldsymbol{\epsilon}_{c]pqr}\nonumber \\
 & =\sum_{\text{Y}}\left(\mathscr{P}_{bc}^{\text{XY}}n_{\text{X}}^{a}-\mathscr{P}_{ec}^{\text{XY}}n_{\text{X}}^{e}\delta_{b}^{a}\right)n_{\text{Y}}^{[d}\xi_{\text{Y}}^{c]}\nu_{d}\hat{\boldsymbol{\epsilon}}+n_{\text{X}}^{a}\xi_{\text{X}}^{c}\pi_{[c}^{\text{X}}\nu_{b]}\hat{\boldsymbol{\epsilon}}-\delta_{b}^{a}\xi_{\text{X}}^{[c}n_{\text{X}}^{d]}\pi_{c}^{\text{X}}\nu_{d}\hat{\boldsymbol{\epsilon}}+\xi_{\text{X}}^{a}n_{\text{X}}^{c}\pi_{[b}^{\text{X}}\nu_{c]}\hat{\boldsymbol{\epsilon}},\nonumber \\
\boldsymbol{\mathcal{Q}}_{a}^{\text{X}} & =\sum_{\text{Y}}\left(\mathscr{P}_{bc}^{\text{XY}}\nabla_{a}n_{\text{X}}^{b}+\mathscr{P}_{c}^{\text{Ys}}\delta_{\text{e}}^{\text{X}}\nabla_{a}s\right)n_{\text{Y}}^{[d}\xi_{\text{Y}}^{c]}\boldsymbol{\epsilon}_{dpqr}+\pi_{b}^{\text{X}}\nabla_{a}\left(n_{\text{X}}^{[d}\xi_{\text{X}}^{b]}\right)\boldsymbol{\epsilon}_{dpqr}\nonumber \\
 & =\sum_{\text{Y}}\left(\mathscr{P}_{bc}^{\text{XY}}\nabla_{a}n_{\text{X}}^{b}+\mathscr{P}_{c}^{\text{Ys}}\delta_{\text{e}}^{\text{X}}\nabla_{a}s\right)n_{\text{Y}}^{[d}\xi_{\text{Y}}^{c]}\nu_{d}\hat{\boldsymbol{\epsilon}}+\pi_{b}^{\text{X}}\nabla_{a}\left(n_{\text{X}}^{[d}\xi_{\text{X}}^{b]}\right)\nu_{d}\hat{\boldsymbol{\epsilon}},\label{coefficient of normal derivative}
\end{align}
again, $\delta_{\text{e}}^{\text{X}}$ is the Kronecker symbol to determine whether the index $\text{X}$ is corresponding to electron. Furthermore, with the 3+1 decomposition in the coordinate given by Eq. \eqref{coordinate}, we can express
\begin{equation}
\delta_{2}g_{ab}=-\frac{2}{\alpha}\nu_{a}\nu_{b}\delta_{2}\alpha-\frac{2}{\alpha}\nu_{(a}\delta_{2}\beta_{b)}+\delta_{2}h_{ab},\quad\delta_{2}A_{a}=-\frac{1}{\alpha}\nu_{a}\delta_{2}\mathcal{A}+\delta_{a}^{i}\delta_{2}A_{i},
\end{equation}
whereby we obtain
\begin{align}
W\left(\delta_{1}\phi,\delta_{2}\phi\right)= & \int_{\Sigma}\left(\delta_{1}\boldsymbol{\pi}^{ij}\delta_{2}h_{ij}-\delta_{1}h_{ij}\delta_{2}\boldsymbol{\pi}^{ij}\right)+\int_{\Sigma}\left(\delta_{1}\boldsymbol{\Pi}^{i}\delta_{2}A_{i}-\delta_{1}A_{i}\delta_{2}\boldsymbol{\Pi}^{i}\right)\nonumber \\
 & -\int_{\Sigma}\sum_{\text{X}}2\left[\left(\mathscr{Q}^{\text{X}}\right)_{a}^{\ bc}\left(-\frac{2}{\alpha}\nu_{a}\nu_{b}\delta_{2}\alpha-\frac{2}{\alpha}\nu_{(a}\delta_{2}\beta_{b)}+\delta_{2}h_{ab}\right)+e^{\text{X}}\left(-\frac{1}{\alpha}\nu_{a}\delta_{2}\mathcal{A}+\delta_{a}^{i}\delta_{2}A_{i}\right)\right]n_{\text{X}}^{[d}\xi_{\text{X}}^{a]}\boldsymbol{\epsilon}_{dpqr}\nonumber \\
 & +\int_{\Sigma}\sum_{\text{X}}\left[\zeta_{\text{X}}^{a}\left(\delta_{1}\boldsymbol{P}_{a}^{\text{X}}+2\boldsymbol{\mathcal{Q}}_{a}^{\text{X}}\right)-2\nabla_{a}\zeta_{\text{X}}^{b}\left(\boldsymbol{\mathcal{P}}_{\text{X}}\right)_{\ b}^{a}\right].\label{integral of symplectic form}
\end{align}
\end{widetext}

As we are seeking the degeneracy directions of $W$ in the full field space, not merely in the solution space, so the field equations and the linearized constraints are not being imposed here, which means that the quantities $\delta_{2}\boldsymbol{\pi}^{ij}$, $\delta_{2}\alpha$, $\delta_{2}\beta_{i}$, $\delta_{2}h_{ij}$, $\delta_{2}\boldsymbol{\Pi}^{i}$, $\delta_{2}\mathcal{A}$, $\delta_{2}\mathcal{A}_{i}$, $\zeta_{\text{X}}^{a}$, and $\nabla_{a}\zeta_{\text{X}}^{b}$ on $\Sigma$ can be varied independently. Thus $\delta\phi\equiv\delta_{1}\phi$ is a degeneracy of $W$ if and only if the coefficients of these quantities in the integral of Eq. \eqref{integral of symplectic form} are each individually zero when pull back onto $\Sigma$, which give rise to the following equations
\begin{align}
\delta_{1}h_{ij} & =0,\label{equation of pi ij}\\
\sum_{\text{X}}\nu_{d}n_{\text{X}}^{[d}\xi_{\text{X}}^{a]}\left(\mathscr{Q}^{\text{X}}\right)_{a}^{\ bc}\nu_{b}\nu_{c} & =0,\label{equation of lapse}\\
\sum_{\text{X}}\nu_{d}n_{\text{X}}^{[d}\xi_{\text{X}}^{a]}\left(\mathscr{Q}^{\text{X}}\right)_{a}^{\ \left(bi\right)}\nu_{b} & =0,\label{equation of shift}\\
\delta_{1}\boldsymbol{\pi}^{ij}-2\sum_{\text{X}}\nu_{d}n_{\text{X}}^{[d}\xi_{\text{X}}^{a]}\left(\mathscr{Q}^{\text{X}}\right)_{a}^{\ \left(ij\right)}\hat{\boldsymbol{\epsilon}} & =0,\label{equation of hij}\\
\delta_{1}A_{i} & =0,\label{equation of pi i}\\
\sum_{\text{X}}e^{\text{X}}n_{\text{X}}^{[d}\xi_{\text{X}}^{a]}\nu_{a}\nu_{d} & =0,\label{equation of A0}\\
\delta_{1}\boldsymbol{\Pi}^{i}-2\sum_{\text{X}}e^{\text{X}}n_{\text{X}}^{[d}\xi_{\text{X}}^{i]}\nu_{d}\hat{\boldsymbol{\epsilon}} & =0,\label{equation of Ai}\\
\delta_{1}\boldsymbol{P}_{a}^{\text{X}}+2\boldsymbol{\mathcal{Q}}_{a}^{\text{X}} & =0,\label{equation of zeta}\\
\left(\boldsymbol{\mathcal{P}}_{\text{X}}\right)_{\ b}^{a} & =0.\label{equation of normal derivative}
\end{align}
Clearly, Eq. \eqref{equation of A0} is satisfied automatically.

Contracting the Eq. \eqref{equation of normal derivative} (which contains three equations) with $\nu_{a}$ and $\varphi^{b}$ gives that
\begin{equation}
n_{\text{X}}^{a}\nu_{a}\sum_{\text{Y}}\mathscr{P}_{bc}^{\text{XY}}\varphi^{b}n_{\text{Y}}^{[d}\xi_{\text{Y}}^{c]}\nu_{d}=0,
\end{equation}
by using the circular flow condition Eq. \eqref{circular flow condition} and the structure of Eq. \eqref{coefficients}, for vector $\theta^{a}$ such that $\theta^{a}\nu_{a}=\theta^{a}\varphi_{a}=0$, one has $\mathscr{P}_{ab}^{\text{XY}}\varphi^{a}\theta^{b}=0$, so that
\begin{align}
0 & =\frac{\varphi^{e}\varphi_{e}}{\left(u_{\text{e}}^{f}\varphi_{f}\right)^{2}}\sum_{\text{Y}}\mathscr{P}_{ac}^{\text{XY}}\varphi^{a}n_{\text{Y}}^{[d}\xi_{\text{Y}}^{c]}\nu_{d}\nonumber \\
 & =\frac{\varphi^{e}\varphi_{e}}{\left(u_{\text{e}}^{f}\varphi_{f}\right)^{2}}\sum_{\text{Y}}\mathscr{P}_{ab}^{\text{XY}}\varphi^{a}\delta_{c}^{b}n_{\text{Y}}^{[d}\xi_{\text{Y}}^{c]}\nu_{d}\nonumber \\
 & =\sum_{\text{Y}}\mathscr{P}_{ab}^{\text{XY}}\frac{\varphi^{a}\varphi^{b}}{\left(u_{\text{e}}^{f}\varphi_{f}\right)^{2}}n_{\text{Y}}^{[d}\xi_{\text{Y}}^{c]}\nu_{d}\varphi_{c},
\end{align}
or we can write them in a matrix form
\begin{equation}
\left(\begin{array}{ccc}
\mathscr{P}_{ab}^{\text{nn}}\frac{\varphi^{a}\varphi^{b}}{\left(u_{\text{e}}^{f}\varphi_{f}\right)^{2}} & \mathscr{P}_{ab}^{\text{np}}\frac{\varphi^{a}\varphi^{b}}{\left(u_{\text{e}}^{f}\varphi_{f}\right)^{2}} & \mathscr{P}_{ab}^{\text{ne}}\frac{\varphi^{a}\varphi^{b}}{\left(u_{\text{e}}^{f}\varphi_{f}\right)^{2}}\\
\mathscr{P}_{ab}^{\text{pn}}\frac{\varphi^{a}\varphi^{b}}{\left(u_{\text{e}}^{f}\varphi_{f}\right)^{2}} & \mathscr{P}_{ab}^{\text{pp}}\frac{\varphi^{a}\varphi^{b}}{\left(u_{\text{e}}^{f}\varphi_{f}\right)^{2}} & \mathscr{P}_{ab}^{\text{pe}}\frac{\varphi^{a}\varphi^{b}}{\left(u_{\text{e}}^{f}\varphi_{f}\right)^{2}}\\
\mathscr{P}_{ab}^{\text{en}}\frac{\varphi^{a}\varphi^{b}}{\left(u_{\text{e}}^{f}\varphi_{f}\right)^{2}} & \mathscr{P}_{ab}^{\text{ep}}\frac{\varphi^{a}\varphi^{b}}{\left(u_{\text{e}}^{f}\varphi_{f}\right)^{2}} & \mathscr{P}_{ab}^{\text{ee}}\frac{\varphi^{a}\varphi^{b}}{\left(u_{\text{e}}^{f}\varphi_{f}\right)^{2}}
\end{array}\right)\left(\begin{array}{c}
n_{\text{n}}^{[c}\xi_{\text{n}}^{d]}\nu_{c}\varphi_{d}\\
n_{\text{p}}^{[c}\xi_{\text{p}}^{d]}\nu_{c}\varphi_{d}\\
n_{\text{e}}^{[c}\xi_{\text{e}}^{d]}\nu_{c}\varphi_{d}
\end{array}\right)=0.\label{matrix equations}
\end{equation}
Note that with the vector $\gamma^{a}$ given in Eq. \eqref{gamma}, we have the following decomposition
\begin{equation}
\varphi^{a}=-u_{\text{e}}^{a}u_{\text{e}}^{b}\varphi_{b}+\gamma^{a}\gamma^{b}\varphi_{b},
\end{equation}
which indicates that
\begin{equation}
\frac{\varphi^{a}\varphi^{b}}{\left(u_{\text{e}}^{f}\varphi_{f}\right)^{2}}=\left(u_{\text{e}}^{a}-\frac{\gamma^{c}\varphi_{c}}{u_{\text{e}}^{f}\varphi_{f}}\gamma^{a}\right)\left(u_{\text{e}}^{b}-\frac{\gamma^{d}\varphi_{d}}{u_{\text{e}}^{f}\varphi_{f}}\gamma^{b}\right).
\end{equation}
Since
\begin{equation}
\left(\frac{\gamma^{a}\varphi_{a}}{u_{\text{e}}^{b}\varphi_{b}}\right)^{2}=\frac{q_{\text{e}}^{ab}\varphi_{a}\varphi_{b}}{\left(u_{\text{e}}^{c}\varphi_{c}\right)^{2}}=\frac{\varphi^{a}\varphi_{a}}{\left(u_{\text{e}}^{c}\varphi_{c}\right)^{2}}+1>1,
\end{equation}
comparing with the eigenvalue equation Eq. \eqref{eigenvalue equation} of the propagation speed, the causal condition implies that $v=\frac{\gamma^{a}\varphi_{a}}{u_{\text{e}}^{b}\varphi_{b}}$ cannot be a root, i.e.,
\begin{equation}
\det\left[\mathscr{P}_{ab}^{\text{XY}}\frac{\varphi^{a}\varphi^{b}}{\left(u_{\text{e}}^{f}\varphi_{f}\right)^{2}}\right]\neq0,
\end{equation}
whence the solutions to the Eq. \eqref{matrix equations} can only be
\begin{equation}
n_{\text{X}}^{[d}\xi_{\text{X}}^{c]}\nu_{d}\varphi_{c}=0.\label{solution to the matrix equation}
\end{equation}
The contraction of Eq. \eqref{equation of normal derivative} with $n_{\text{X}a}$ and $\theta^{b}$ gives that
\begin{equation}
0=\sum_{\text{Y}}\mathscr{P}_{bc}^{\text{XY}}\theta^{b}\xi_{\text{Y}}^{c}n_{\text{Y}}^{d}\nu_{d}=\sum_{\text{Y}}\mathbb{I}_{\text{XY}}\left(\xi_{\text{Y}}^{a}\theta_{a}n_{\text{Y}}^{b}\nu_{b}\right),\label{transverse matrix equation}
\end{equation}
where the inertia matrix Eq. \eqref{inertia matrix} is positive definite, so the solutions to Eq. \eqref{transverse matrix equation} is given by
\begin{equation}
\xi_{\text{X}}^{a}\theta_{a}=0,
\end{equation}
i.e., $\xi_{\text{X}}^{a}$ is the linear combination of $\nu^{a}$ and $\varphi^{a}$. Thus Eq. \eqref{solution to the matrix equation} amounts to say that the only choice is $\xi_{\text{X}}^{a}\propto u_{\text{X}}^{a}$. Substituting return to Eq. \eqref{equation of normal derivative}, one can easily see that $\xi_{\text{X}}^{a}\propto u_{\text{X}}^{a}$ is actually the solution. It immediately that Eqs. \eqref{equation of lapse} and \eqref{equation of shift} are satisfied, in addition, Eqs. \eqref{equation of hij} and \eqref{equation of Ai} yield $\delta_{1}\boldsymbol{\pi}^{ij}=0$ and $\delta_{1}\boldsymbol{\Pi}^{i}=0$ respectively.

Thus we have shown that besides Eq. \eqref{equation of zeta}, the other equations of degeneracy are equivalent to $\delta_{1}h_{ij}=0$, $\delta_{1}\boldsymbol{\pi}^{ij}=0$, $\delta_{1}A_{i}=0$, $\delta_{1}\boldsymbol{\Pi}^{i}=0$, and $\xi_{\text{X}}^{a}\propto u_{\text{X}}^{a}$. With these conditions, let 
\begin{equation}
\xi_{\text{X}}^{a}=U^{\text{X}}u_{\text{X}}^{a}+\tau\eta_{\text{X}}^{a},
\end{equation}
where $U^{\text{X}}$ is an arbitrary smooth function, $\Sigma$ is given by $\tau=0$ with the unit normal covector $\nu_{a}=-\alpha\left(d\tau\right)_{a}$ as in the coordinate Eq. \eqref{coordinate}. Since with the flowline trivial perturbation $\delta\phi=\left(0,0,U^{\text{X}}u_{\text{X}}^{a}\right)$ 
\begin{align}
\delta n_{\text{X}}^{a} & =-U^{\text{X}}u_{\text{X}}^{b}\nabla_{b}n_{\text{X}}^{a}+n_{\text{X}}^{b}\nabla_{b}\left(U^{\text{X}}u_{\text{X}}^{a}\right)-n_{\text{X}}^{a}\nabla_{b}\left(U^{\text{X}}u_{\text{X}}^{b}\right)\nonumber \\
 & =-U^{\text{X}}u_{\text{X}}^{a}\nabla_{b}n_{\text{X}}^{b}\nonumber \\
 & =0,
\end{align}
and
\begin{equation}
\delta s=-U^{\text{e}}u_{\text{e}}^{a}\nabla_{a}s=0,
\end{equation}
where we have used the conservation law Eq. \eqref{conservation of n=000026s}, then $\delta\mu_{a}^{\text{X}}=0$ and so that $\delta\boldsymbol{P}_{a}^{\text{X}}=0$. Denote $\delta_{1}^{\prime}\phi=\left(\delta_{1}g_{ab},\delta_{1}A_{a},\tau\eta_{\text{X}}^{a}\right)$ and consider Eq. \eqref{coefficient of normal derivative}, Eq. \eqref{equation of zeta} reads that
\begin{equation}
\delta_{1}^{\prime}\boldsymbol{P}_{a}^{\text{X}}-\frac{2}{\alpha}\pi_{b}^{\text{X}}n_{\text{X}}^{[d}\eta_{\text{X}}^{b]}\nu_{d}\nu_{a}\hat{\boldsymbol{\epsilon}}=0.\label{equation of zeta 2}
\end{equation}
According to Eq. \eqref{n variation}
\begin{equation}
\delta_{1}^{\prime}n_{\text{X}}^{a}=\frac{2}{\alpha}n_{\text{X}}^{[a}\eta_{\text{X}}^{b]}\nu_{b}-\frac{1}{2}n_{\text{X}}^{a}g^{bc}\delta_{1}g_{bc},
\end{equation}
and $\delta_{1}A_{i}=0$ implies that $\delta_{1}A_{a}=-\nu_{a}\nu^{b}\delta_{1}A_{b}$, then
\begin{align}
\delta_{1}^{\prime}\boldsymbol{P}_{a}^{\text{X}}= & 2\delta_{1}^{\prime}\left(\pi_{b}^{\text{X}}n_{\text{X}}^{[c}\delta_{a}^{b]}\boldsymbol{\epsilon}_{cpqr}\right)\nonumber \\
= & 2\delta_{1}^{\prime}\pi_{b}^{\text{X}}n_{\text{X}}^{[c}\delta_{a}^{b]}\nu_{c}\hat{\boldsymbol{\epsilon}}\nonumber \\
 & +\frac{2}{\alpha}\pi_{b}^{\text{X}}\left(n_{\text{X}}^{[c}\eta_{\text{X}}^{d]}\nu_{d}\delta_{a}^{b}-n_{\text{X}}^{[b}\eta_{\text{X}}^{d]}\nu_{d}\delta_{a}^{c}\right)\nu_{c}\hat{\boldsymbol{\epsilon}}\nonumber \\
= & \left(2\delta_{1}^{\prime}\mu_{b}^{\text{X}}n_{\text{X}}^{[c}\delta_{a}^{b]}-2e^{\text{X}}\delta_{1}A_{b}n_{\text{X}}^{[c}\delta_{a}^{b]}\right)\nu_{c}\hat{\boldsymbol{\epsilon}}\nonumber \\
 & -\frac{2}{\alpha}\pi_{b}^{\text{X}}n_{\text{X}}^{[b}\eta_{\text{X}}^{d]}\nu_{d}\delta_{a}^{c}\nu_{c}\hat{\boldsymbol{\epsilon}}\nonumber \\
= & \left(2\delta_{1}\mu_{d}^{\text{X}}\delta_{b}^{d}n_{\text{X}}^{[c}\delta_{a}^{b]}\nu_{c}-\frac{2}{\alpha}\pi_{b}^{\text{X}}n_{\text{X}}^{[b}\eta_{\text{X}}^{d]}\nu_{d}\nu_{a}\right)\hat{\boldsymbol{\epsilon}}\nonumber \\
= & \left(2\delta_{1}\mu_{d}^{\text{X}}h_{b}^{d}n_{\text{X}}^{[c}\delta_{a}^{b]}\nu_{c}-\frac{2}{\alpha}\pi_{b}^{\text{X}}n_{\text{X}}^{[b}\eta_{\text{X}}^{d]}\nu_{d}\nu_{a}\right)\hat{\boldsymbol{\epsilon}},
\end{align}
hence Eq. \eqref{equation of zeta 2} becomes
\begin{equation}
0=2\delta_{1}\mu_{d}^{\text{X}}h_{b}^{d}n_{\text{X}}^{[c}\delta_{a}^{b]}\nu_{c}=2\delta_{1}\left(\mu_{d}^{\text{X}}h^{di}\right)h_{ib}n_{\text{X}}^{[c}\delta_{a}^{b]}\nu_{c}.
\end{equation}
By contracting with $\nu^{a}$, we have that $\left(\delta_{1}\mu^{\text{X}i}\right)h_{ia}n_{\text{X}}^{a}=0$, whereby it further implies $\delta_{1}\mu^{\text{X}i}=0$.

As a conclusion, we find out that $\delta\phi$ is a degeneracy of $W$ if and only if the quantities $\left(\delta h_{ij},\delta\boldsymbol{\pi}^{ij},\delta A_{i},\delta\boldsymbol{\Pi}^{i},q_{\text{X}}^{ab}\xi_{\text{X}a},\delta_{1}\mu^{\text{X}i}\right)$ vanish on $\Sigma$. These quantities are the first order variations of the quantities $\left(h_{ij},\boldsymbol{\pi}^{ij},A_{i},\boldsymbol{\Pi}^{i},\sigma_{\text{X}},\mu^{\text{X}i}\right)$ on $\Sigma$, where $\sigma_{\text{X}}$ is a diffeomorphism from the space of fiducial flowlines to $\Sigma$ given at the beginning of this subsection, thus the phase space is described by the quantities $\left(h_{ij},\boldsymbol{\pi}^{ij},A_{i},\boldsymbol{\Pi}^{i},\sigma_{\text{X}},\mu^{\text{X}i}\right)$ on $\Sigma$.

However, the variables $\left(\sigma_{\text{X}},\mu^{\text{X}i}\right)$ are not canonically conjugate as the symplectic product Eq. \eqref{explicity presymplectic form} of two pure $\sigma_{\text{X}}$ perturbations (keeping $\mu^{\text{X}i}$ fixed) is not necessarily zero. To construct the canonically conjugate variables for our configuration, we like to choose the coordinates in $\mathcal{M}_{\text{X}}^{\prime}$ such that the fiducial flowlines are given by the integral curves of $\left(\partial_{x_{\text{X}}^{\prime0}}\right)^{a}$, then by $\delta x_{\text{X}}^{\prime\mu}=-\xi_{\text{X}}^{a}\partial_{a}x_{\text{X}}^{\prime\mu}$, we can write
\begin{equation}
\boldsymbol{\theta}^{\left(\text{M}\right)}=-\sum_{\text{X}}\delta x_{\text{X}}^{\prime\mu}\boldsymbol{P}_{\mu}^{\text{X}\prime}=-\sum_{\text{X}}\delta x_{\text{X}}^{\prime i}\boldsymbol{P}_{i}^{\text{X}\prime},
\end{equation}
where $\boldsymbol{P}_{\mu}^{\text{X}\prime}=\left(\partial_{x_{\text{X}}^{\prime\mu}}\right)^{a}\sigma_{\text{X}}^{*}\left(\boldsymbol{P}_{a}^{\text{X}}\right)$ with $\boldsymbol{P}_{0}^{\text{X}\prime}=\sigma_{\text{X}}^{*}\left(u_{\text{X}}^{a}\boldsymbol{P}_{a}^{\text{X}}\right)=0$.
Furthermore, we have
\begin{equation}
\boldsymbol{\omega}^{\left(\text{M}\right)}=-\sum_{\text{X}}\left(\delta_{2}x_{\text{X}}^{\prime i}\delta_{1}\boldsymbol{P}_{i}^{\text{X}\prime}-\delta_{1}x_{\text{X}}^{\prime i}\delta_{2}\boldsymbol{P}_{i}^{\text{X}\prime}\right),
\end{equation}
which tells us that the pairs of variables $\left(x_{\text{X}}^{\prime i},-\boldsymbol{P}_{i}^{\text{X}\prime}\right)$ can be regarded as the canonically conjugate variables for our configuration. Hence the symplectic form for our Einstein-superconducting-superfluid system can be cast into the following canonical way
\begin{equation}
W_{AB}\delta_{1}\phi^{A}\delta_{2}\phi^{B}=\int_{\Sigma}\left(\delta_{2}q^{\alpha}\delta_{1}p_{\alpha}-\delta_{1}q^{\alpha}\delta_{2}p_{\alpha}\right),\label{canonical symplectic form}
\end{equation}
with $q^{\alpha}=\left(h_{ij},A_{i},x_{\text{X}}^{\prime i}\right)$ and $p_{\alpha}=\left(\boldsymbol{\pi}^{ij},\boldsymbol{\Pi}^{i},-\boldsymbol{P}_{i}^{\text{X}\prime}\right)$, but the explicit form of above symplectic form is not needed here.

Using such canonically conjugate coordinates, we can define the Hilbert space structure $\mathcal{H}$ on perturbations by introducing the $L^{2}$ inner product 
\begin{equation}
\braket{\delta_{1}\phi,\delta_{2}\phi}=\int_{\Sigma}\sum_{\alpha}\left(\delta_{1}q^{\alpha}\cdot\delta_{2}q^{\alpha}+\delta_{1}p_{\alpha}\cdot\delta_{2}p_{\alpha}\right),\label{L2 inner product}
\end{equation}
where we use ``$\cdot$'' denotes the contraction of all tensor indices after using the background metric $h_{ab}$ on $\Sigma$ to raise and lower indices. Thus the elements of $\mathcal{H}$ are the square integrable tensor fields $\left(q^{\alpha},p_{\alpha}\right)$ on $\Sigma$. Note that the perturbations for which $\delta M\neq0$ fall off too slowly to be square integrable, but $\mathcal{H}$ contains all perturbations of interest for which $\delta M=0$.

By comparing Eqs. \eqref{canonical symplectic form} and \eqref{L2 inner product}, it can be seen that $W$ is a bounded quadratic form on $\mathcal{H}$ and hence corresponds to a bounded linear map $\hat{W}:\mathcal{H}\rightarrow\mathcal{H}$ given by 
\begin{equation}
\hat{W}\left(\delta q^{\alpha},\delta p_{\alpha}\right)=\left(-\delta p_{\alpha},\delta q^{\alpha}\right),
\end{equation}
where it is understood that any tensor indices on $\left(\delta q^{\alpha},\delta p_{\alpha}\right)$ are converted to the corresponding dual indices on the right side via raising and lowering with $h^{ab}$ and $h_{ab}$ and we have assumed $h=1$. Accordingly, the symplectic product can be written as following
\begin{equation}
W_{AB}\delta_{1}\phi^{A}\delta_{2}\phi^{B}=\braket{\delta_{1}\phi,\hat{W}\delta_{2}\phi},
\end{equation}
and immediately, $\hat{W}^{2}=-\text{id}$ and $\hat{W}^{\dagger}=-\hat{W}$, so, in particular, $\hat{W}$ is an orthogonal map. 

Let $\mathcal{S}$ be any subspace of $\mathcal{H}$, we define the $symplectic$ $complement$, $\mathcal{S}^{\perp_{s}}$, of $\mathcal{S}$ by
\begin{equation}
\mathcal{S}^{\perp_{s}}=\left\{ v\in\mathcal{H}\vert\braket{v,\hat{W}u}=0\text{, }\forall u\in\mathcal{S}\right\} .
\end{equation}
Clearly, we have $\mathcal{S}^{\perp_{s}}=\left(\hat{W}\left[\mathcal{S}\right]\right)^{\perp}$, where $\hat{W}\left[\mathcal{S}\right]$ denotes the image of $\mathcal{S}$ under $\hat{W}$ and ``$\perp$'' denotes the ordinary orthogonal complement in $\mathcal{H}$. Since $\hat{W}$ is orthogonal, then
\begin{align}
 & v\in\left(\hat{W}\left[\mathcal{S}\right]\right)^{\perp},\nonumber \\
\Longleftrightarrow & \braket{\hat{W}v,u}=\braket{v,\hat{W}u}=0\text{ for }u\in\mathcal{S},\nonumber \\
\Longleftrightarrow & \hat{W}v\in\mathcal{S}^{\perp},\nonumber \\
\Longleftrightarrow & v=-\hat{W}^{2}v\in\hat{W}\left[\mathcal{S}^{\perp}\right],
\end{align}
so that we have $\left(\hat{W}\left[\mathcal{S}\right]\right)^{\perp}=\hat{W}\left[\mathcal{S}^{\perp}\right]$. And we also have
\begin{align}
\left(\mathcal{S}^{\perp_{s}}\right)^{\perp_{s}} & =\left(\hat{W}\left[\mathcal{S}^{\perp}\right]\right)^{\perp_{s}}=\hat{W}\left[\left(\hat{W}\left[\mathcal{S}^{\perp}\right]\right)^{\perp}\right]\nonumber \\
 & =\hat{W}^{2}\left[\left(\mathcal{S}^{\perp}\right)^{\perp}\right]=\left(\mathcal{S}^{\perp}\right)^{\perp}=\overline{\mathcal{S}},\label{double symplectic complement}
\end{align}
where the bar denotes the closure in $\mathcal{H}$. Thus the double symplectic complement of any subspace is its closure. Since any subspace is dense in its closure, we shall not bother ourselves by saying that the double symplectic complement of any subspace is itself.

Now let $\phi$ satisfy the equations of motion and let $\mathcal{X}^{a}$ be smooth and of compact support. By the fundamental identity Eq. \eqref{fundamental identity}, we have for all $\delta\phi\in\mathcal{H}$, 
\begin{equation}
\braket{\delta\phi,\hat{W}\mathscr{L}_{\mathcal{X}}\phi}=\int_{\Sigma}X^{a}\delta\boldsymbol{C}_{a}.
\end{equation}
By definition, the right side vanishes if and only if $\delta\phi$ is a weak solution of the constraint equations, $\delta\boldsymbol{C}_{a}=0$. Thus if we take $\mathcal{G}$ to be subspace of $\mathcal{H}$ spanned by perturbations of the form $\mathscr{L}_{\mathcal{X}}\phi$, we see that $\mathcal{G}^{\perp_{s}}$ is precisely the subspace, $\mathcal{C}$, of weak solutions to the linearized constraints. Furthermore, by the general argument of the previous paragraph, we have $\mathcal{C}^{\perp_{s}}=\overline{\mathcal{G}}$. Another way of saying this is that if we restrict the action of the original quadratic form to $\mathcal{C}\times\mathcal{C}$, it becomes degenerate precisely on the gauge transformations $\mathscr{L}_{\mathcal{X}}\phi$.

\subsection{Trivial displacements}

\label{subsec:Trivial-displacements}

As stated in Sec. \ref{subsec2.2}, the \emph{trivial displacement} $\eta_{\text{X}}^{a}$ is the displacement satisfying that 
\begin{align}
0 & =\delta s=-\mathscr{L}_{\eta_{\text{e}}}s,\nonumber \\
0 & =\delta\boldsymbol{N}^{\text{X}}=-\mathscr{L}_{\eta_{\text{X}}}\boldsymbol{N}^{\text{X}}.\label{trivial perturbation condition}
\end{align}
and the corresponding trivial perturbation is given by $\delta_{t}\phi=\left(0,0,\eta_{\text{X}}^{a}\right)$. We first give the general form of a trivial displacement. Since $\iota_{u_{\text{X}}}\boldsymbol{N}^{\text{X}}=0$, any vector field $\eta_{\text{X}}^{a}$ inside the star can be uniquely decomposed as
\begin{equation}
\eta_{\text{X}}^{a}=U^{\text{X}}u_{\text{X}}^{a}+\frac{1}{n_{\text{X}}^{2}}N^{\text{X}abc}H_{bc}^{\text{X}},
\end{equation}
where $U^{\text{X}}$ is an arbitrary function, $\boldsymbol{H}^{\text{X}}$ is an arbitrary 2-form satisfying $\iota_{u_{\text{X}}}\boldsymbol{H}^{\text{X}}=0$, and the necessary and sufficient condition for $\boldsymbol{H}^{\text{X}}$ is given by
\begin{equation}
0=\mathscr{L}_{\eta_{\text{X}}}\boldsymbol{N}^{\text{X}}=d\left(\iota_{\eta_{\text{X}}}\boldsymbol{N}^{\text{X}}\right)=2d\boldsymbol{H}^{\text{X}},
\end{equation}
where we have used that $N^{\text{X}abc}N_{ade}^{\text{X}}=2n_{\text{X}}^{2}q_{\text{X}d}^{[b}q_{\text{X}e}^{c]}$. It follows immediately that 
\begin{equation}
\mathscr{L}_{u_{\text{X}}}\boldsymbol{H}^{\text{X}}=d\left(\iota_{u_{\text{X}}}\boldsymbol{H}^{\text{X}}\right)+\iota_{u_{\text{X}}}d\boldsymbol{H}^{\text{X}}=0,
\end{equation}
so $\boldsymbol{H}^{\text{X}}$ may be viewed as a 2-form on the manifold of orbits of $u_{\text{X}}^{a}$. Assuming that there are no holes in our star, then $d\boldsymbol{H}^{\text{X}}=0$ implies that
\begin{equation}
\boldsymbol{H}^{\text{X}}=d\boldsymbol{Z}^{\text{X}},
\end{equation}
where $\boldsymbol{Z}^{\text{X}}$ is an arbitrary 1-form on the manifold of $u_{\text{X}}^{a}$-orbits, or, equivalently, $\boldsymbol{Z}^{\text{X}}$ is a 1-form on spacetime satisfying 
\begin{equation}
\iota_{u_{\text{X}}}\boldsymbol{Z}^{\text{X}}=0,\quad\mathscr{L}_{u_{\text{X}}}\boldsymbol{Z}^{\text{X}}=0.\label{condition of Z}
\end{equation}
Thus the necessary and sufficient condition for $\eta_{\text{X}}^{a}$ to satisfy $\mathscr{L}_{\eta_{\text{X}}}\boldsymbol{N}^{\text{X}}=0$ is that it be of the form
\begin{equation}
\eta_{\text{X}}^{a}=U^{\text{X}}u_{\text{X}}^{a}+\frac{1}{n_{\text{X}}^{2}}N^{\text{X}abc}\nabla_{b}Z_{c}^{\text{X}}.\label{general trivial displacement}
\end{equation}
Since $u_{\text{e}}^{a}\nabla_{a}s=0$, the necessary and sufficient condition for $\eta_{\text{e}}^{a}$ to also satisfy $\eta_{\text{e}}^{a}\nabla_{a}s=0$ is
\begin{equation}
\nabla_{[a}s\nabla_{b}Z_{c]}^{\text{e}}=0.\label{condition of ZS}
\end{equation}
As a conclusion, $\eta_{\text{X}}^{a}$ are trivial displacements if and only if they are of the form Eq. \eqref{general trivial displacement} and satisfy Eqs. \eqref{condition of Z} and \eqref{condition of ZS}.

Now let us compute the symplectic product of a trivial perturbation $\delta_{t}\phi=\left(0,0,\eta_{\text{X}}^{a}\right)$ with an arbitrary perturbation. First of all, consider the flowline trivial perturbation $\delta_{ft}\phi=\left(0,0,U^{\text{X}}u_{\text{X}}^{a}\right)$. As mentioned in previous subsection, one has $\delta_{ft}\boldsymbol{P}_{a}^{\text{X}}=0$, and for an arbitrary perturbation $\delta\phi=\left(\delta g_{ab},\delta A_{a},\xi_{\text{X}}^{a}\right)$, since $\delta\left(u_{\text{X}}^{a}\boldsymbol{P}_{a}^{\text{X}}\right)=0$ and $\left(\delta+\mathscr{L}_{\xi_{\text{X}}}\right)u_{\text{X}}^{a}\propto u_{\text{X}}^{a}$ by Eq. \eqref{u variation}, then
\begin{align}
W\left(\delta\phi,\delta_{ft}\phi\right)= & \int_{\Sigma}\sum_{\text{X}}\left(U^{\text{X}}u_{\text{X}}^{a}\delta\boldsymbol{P}_{a}^{\text{X}}-\left[\xi_{\text{X}},U^{\text{X}}u_{\text{X}}\right]^{a}\boldsymbol{P}_{a}^{\text{X}}\right)\nonumber \\
= & -\int_{\Sigma}\sum_{\text{X}}U^{\text{X}}\boldsymbol{P}_{a}^{\text{X}}\left(\delta+\mathscr{L}_{\xi_{\text{X}}}\right)u_{\text{X}}^{a}\nonumber \\
\propto & \int_{\Sigma}\sum_{\text{X}}U^{\text{X}}\boldsymbol{P}_{a}^{\text{X}}u_{\text{X}}^{a}\nonumber \\
= & 0,
\end{align}
i.e., all of the flowline trivial perturbations are degeneracy of $W$. Immediately, these flowline trivial perturbations are factored out by our construction of phase space. However, the trivial perturbations $\delta_{t}\phi=\left(0,0,\tilde{\eta}_{\text{X}}^{a}\right)$ with displacements of the form
\begin{equation}
\tilde{\eta}_{\text{X}}^{a}=\frac{1}{n_{\text{X}}^{2}}N^{\text{X}abc}\nabla_{b}Z_{c}^{\text{X}},\label{general nontrivial trivial displacement}
\end{equation}
may not be the degeneracy of $W$. Indeed, since 
\begin{equation}
\delta_{t}n_{\text{X}}^{a}=-\frac{1}{3!}\delta_{t}\left(\boldsymbol{N}_{bcd}^{\text{X}}\boldsymbol{\epsilon}^{abcd}\right)=0,
\end{equation}
then it is clearly that $\delta_{t}\boldsymbol{P}_{a}^{\text{X}}=0$, so that
\begin{align}
 & W\left(\delta\phi,\delta_{t}\phi\right)\nonumber \\
= & \int_{\Sigma}\sum_{\text{X}}\left(\tilde{\eta}_{\text{X}}^{a}\delta\boldsymbol{P}_{a}^{\text{X}}-\left[\xi_{\text{X}},\tilde{\eta}_{\text{X}}\right]^{a}\boldsymbol{P}_{a}^{\text{X}}\right)\nonumber \\
= & \int_{\Sigma}\sum_{\text{X}}\left(\frac{1}{n_{\text{X}}^{2}}N^{\text{X}abc}\nabla_{b}Z_{c}^{\text{X}}\delta\boldsymbol{P}_{a}^{\text{X}}-\mathscr{L}_{\xi_{\text{X}}}\tilde{\eta}_{\text{X}}^{a}\boldsymbol{P}_{a}^{\text{X}}\right)\nonumber \\
= & \int_{\Sigma}\sum_{\text{X}}\nabla_{b}Z_{c}^{\text{X}}\left[\delta\left(\frac{1}{n_{\text{X}}^{2}}N^{\text{X}abc}\boldsymbol{P}_{a}^{\text{X}}\right)-\boldsymbol{P}_{a}^{\text{X}}\delta\left(\frac{1}{n_{\text{X}}^{2}}N^{\text{X}abc}\right)\right]\nonumber \\
 & -\int_{\Sigma}\sum_{\text{X}}\mathscr{L}_{\xi_{\text{X}}}\tilde{\eta}_{\text{X}}^{a}\boldsymbol{P}_{a}^{\text{X}}.
\end{align}
Note that $u_{\text{X}}^{a}\boldsymbol{P}_{[a}^{\text{X}}\nabla_{b}Z_{c]}^{\text{X}}=0$, $\boldsymbol{P}_{[a}^{\text{X}}\nabla_{b}Z_{c}^{\text{X}}u_{\text{X}d]}\propto\boldsymbol{\epsilon}_{abcd}$ as a top form, and $N^{\text{X}abc}N_{abc}^{\text{X}}=6n_{\text{X}}^{2}$, whence 
\begin{equation}
\boldsymbol{P}_{[a}^{\text{X}}\nabla_{b}Z_{c]}^{\text{X}}=\frac{1}{6n_{\text{X}}^{2}}\boldsymbol{P}_{[d}^{\text{X}}\nabla_{e}Z_{f]}^{\text{X}}N^{\text{X}def}N_{abc}^{\text{X}}.
\end{equation}
Accordingly, we have
\begin{widetext}
\begin{align}
W\left(\delta\phi,\delta_{t}\phi\right) & =\int_{\Sigma}\sum_{\text{X}}\left[\nabla_{b}Z_{c}^{\text{X}}\delta\left(\frac{1}{n_{\text{X}}^{2}}N^{\text{X}abc}\boldsymbol{P}_{a}^{\text{X}}\right)-\frac{1}{6n_{\text{X}}^{2}}\boldsymbol{P}_{[d}^{\text{X}}\nabla_{e}Z_{f]}^{\text{X}}N^{\text{X}def}N_{abc}^{\text{X}}\delta\left(\frac{1}{n_{\text{X}}^{2}}N^{\text{X}abc}\right)-\mathscr{L}_{\xi_{\text{X}}}\tilde{\eta}_{\text{X}}^{a}\boldsymbol{P}_{a}^{\text{X}}\right]\nonumber \\
 & =\int_{\Sigma}\sum_{\text{X}}\left[\nabla_{b}Z_{c}^{\text{X}}\delta\left(\frac{1}{n_{\text{X}}^{2}}N^{\text{X}abc}\boldsymbol{P}_{a}^{\text{X}}\right)+\frac{1}{6n_{\text{X}}^{2}}\boldsymbol{P}_{[d}^{\text{X}}\nabla_{e}Z_{f]}^{\text{X}}N^{\text{X}def}\frac{1}{n_{\text{X}}^{2}}N^{\text{X}abc}\delta N_{abc}^{\text{X}}-\mathscr{L}_{\xi_{\text{X}}}\tilde{\eta}_{\text{X}}^{a}\boldsymbol{P}_{a}^{\text{X}}\right]\nonumber \\
 & =\int_{\Sigma}\sum_{\text{X}}\left[\nabla_{b}Z_{c}^{\text{X}}\delta\left(\frac{1}{n_{\text{X}}^{2}}N^{\text{X}abc}\boldsymbol{P}_{a}^{\text{X}}\right)-\frac{1}{6n_{\text{X}}^{2}}\boldsymbol{P}_{[d}^{\text{X}}\nabla_{e}Z_{f]}^{\text{X}}N^{\text{X}def}\frac{1}{n_{\text{X}}^{2}}N^{\text{X}abc}\mathscr{L}_{\xi_{\text{X}}}N_{abc}^{\text{X}}-\mathscr{L}_{\xi_{\text{X}}}\tilde{\eta}_{\text{X}}^{a}\boldsymbol{P}_{a}^{\text{X}}\right]\nonumber \\
 & =\int_{\Sigma}\sum_{\text{X}}\left[\nabla_{b}Z_{c}^{\text{X}}\delta\left(\frac{1}{n_{\text{X}}^{2}}N^{\text{X}abc}\boldsymbol{P}_{a}^{\text{X}}\right)+\frac{1}{6n_{\text{X}}^{2}}\boldsymbol{P}_{[d}^{\text{X}}\nabla_{e}Z_{f]}^{\text{X}}N^{\text{X}def}N_{abc}^{\text{X}}\mathscr{L}_{\xi_{\text{X}}}\left(\frac{1}{n_{\text{X}}^{2}}N^{\text{X}abc}\right)-\mathscr{L}_{\xi_{\text{X}}}\tilde{\eta}_{\text{X}}^{a}\boldsymbol{P}_{a}^{\text{X}}\right]\nonumber \\
 & =\int_{\Sigma}\sum_{\text{X}}\left[\nabla_{b}Z_{c}^{\text{X}}\delta\left(\frac{1}{n_{\text{X}}^{2}}N^{\text{X}abc}\boldsymbol{P}_{a}^{\text{X}}\right)+\boldsymbol{P}_{a}^{\text{X}}\nabla_{b}Z_{c}^{\text{X}}\mathscr{L}_{\xi_{\text{X}}}\left(\frac{1}{n_{\text{X}}^{2}}N^{\text{X}abc}\right)-\mathscr{L}_{\xi_{\text{X}}}\left(\frac{1}{n_{\text{X}}^{2}}N^{\text{X}abc}\nabla_{b}Z_{c}^{\text{X}}\right)\boldsymbol{P}_{a}^{\text{X}}\right]\nonumber \\
 & =\int_{\Sigma}\sum_{\text{X}}\left[\nabla_{b}Z_{c}^{\text{X}}\delta\left(\frac{1}{n_{\text{X}}^{2}}N^{\text{X}abc}\boldsymbol{P}_{a}^{\text{X}}\right)-\frac{1}{n_{\text{X}}^{2}}N^{\text{X}abc}\boldsymbol{P}_{a}^{\text{X}}\mathscr{L}_{\xi_{\text{X}}}\left(\nabla_{b}Z_{c}^{\text{X}}\right)\right]\nonumber \\
 & =\int_{\Sigma}\sum_{\text{X}}\left[\nabla_{b}Z_{c}^{\text{X}}\delta\left(\frac{1}{n_{\text{X}}^{2}}\pi_{a}^{\text{X}}n_{\text{X}}^{[d}N^{\text{X}a]bc}\boldsymbol{\epsilon}_{dpqr}\right)-\frac{1}{n_{\text{X}}^{2}}\pi_{a}^{\text{X}}n_{\text{X}}^{[d}N^{\text{X}a]bc}\boldsymbol{\epsilon}_{dpqr}\mathscr{L}_{\xi_{\text{X}}}\left(\nabla_{b}Z_{c}^{\text{X}}\right)\right]\nonumber \\
 & =6\int_{\Sigma}\sum_{\text{X}}\left[\nabla_{[b}Z_{c]}^{\text{X}}\delta\left(\pi_{[r}^{\text{X}}q_{\text{X}p}^{b}q_{\text{X}q]}^{c}\right)-\pi_{[r}^{\text{X}}q_{\text{X}p}^{b}q_{\text{X}q]}^{c}\mathscr{L}_{\xi_{\text{X}}}\left(\nabla_{[b}Z_{c]}^{\text{X}}\right)\right].\label{symplectic product of trivial}
\end{align}
Also, according to Eq. \eqref{condition of Z}, we have $\mathscr{L}_{\xi_{\text{X}}}\left(u_{\text{X}}^{b}\nabla_{[b}Z_{c]}^{\text{X}}\right)=0$, which implies that
\begin{align}
 & \nabla_{[b}Z_{c]}^{\text{X}}\pi_{r}^{\text{X}}\delta\left(q_{\text{X}p}^{b}q_{\text{X}q}^{c}\right)-\pi_{r}^{\text{X}}q_{\text{X}p}^{b}q_{\text{X}q}^{c}\mathscr{L}_{\xi_{\text{X}}}\left(\nabla_{[b}Z_{c]}^{\text{X}}\right)\nonumber \\
= & \nabla_{[b}Z_{c]}^{\text{X}}\pi_{r}^{\text{X}}\delta\left(\delta_{p}^{b}u_{\text{X}}^{c}u_{\text{X}q}+u_{\text{X}}^{b}u_{\text{X}p}\delta_{q}^{c}\right)-\pi_{r}^{\text{X}}\left(\delta_{p}^{b}u_{\text{X}}^{c}u_{\text{X}q}+u_{\text{X}}^{b}u_{\text{X}p}\delta_{q}^{c}\right)\mathscr{L}_{\xi_{\text{X}}}\left(\nabla_{[b}Z_{c]}^{\text{X}}\right)-\pi_{r}^{\text{X}}\mathscr{L}_{\xi_{\text{X}}}\left(\nabla_{[p}Z_{q]}^{\text{X}}\right)\nonumber \\
= & \nabla_{[b}Z_{c]}^{\text{X}}\pi_{r}^{\text{X}}\left(\delta_{p}^{b}\delta u_{\text{X}}^{c}u_{\text{X}q}+\delta u_{\text{X}}^{b}u_{\text{X}p}\delta_{q}^{c}\right)+\nabla_{[b}Z_{c]}^{\text{X}}\pi_{r}^{\text{X}}\left(\delta_{p}^{b}\mathscr{L}_{\xi_{\text{X}}}u_{\text{X}}^{c}u_{\text{X}q}+\mathscr{L}_{\xi_{\text{X}}}u_{\text{X}}^{b}u_{\text{X}p}\delta_{q}^{c}\right)-\pi_{r}^{\text{X}}\mathscr{L}_{\xi_{\text{X}}}\left(\nabla_{[p}Z_{q]}^{\text{X}}\right)\nonumber \\
= & \nabla_{[b}Z_{c]}^{\text{X}}\pi_{r}^{\text{X}}\left[\delta_{p}^{b}u_{\text{X}q}\left(\delta+\mathscr{L}_{\xi_{\text{X}}}\right)u_{\text{X}}^{c}+\delta_{q}^{c}u_{\text{X}p}\left(\delta+\mathscr{L}_{\xi_{\text{X}}}\right)u_{\text{X}}^{b}\right]-\pi_{r}^{\text{X}}\mathscr{L}_{\xi_{\text{X}}}\left(\nabla_{[p}Z_{q]}^{\text{X}}\right)\nonumber \\
= & -\pi_{r}^{\text{X}}\mathscr{L}_{\xi_{\text{X}}}\left(\nabla_{[p}Z_{q]}^{\text{X}}\right).\label{Z calculation}
\end{align}
Substitute Eq. \eqref{Z calculation} into Eq. \eqref{symplectic product of trivial}, we get that
\begin{align}
W\left(\delta\phi,\delta_{t}\phi\right) & =6\int_{\Sigma}\sum_{\text{X}}\left[\nabla_{[b}Z_{c]}^{\text{X}}q_{\text{X[}p}^{b}q_{\text{X}q}^{c}\delta\pi_{r]}^{\text{X}}-\pi_{[r}^{\text{X}}\mathscr{L}_{\xi_{\text{X}}}\left(\nabla_{p}Z_{q]}^{\text{X}}\right)\right]\nonumber \\
 & =6\int_{\Sigma}\sum_{\text{X}}\left[\nabla_{[p}Z_{q}^{\text{X}}\delta\pi_{r]}^{\text{X}}+\nabla_{[p}Z_{q}^{\text{X}}\mathscr{L}_{\xi_{\text{X}}}\pi_{r]}^{\text{X}}-\mathscr{L}_{\xi_{\text{X}}}\left(\pi_{[r}^{\text{X}}\nabla_{p}Z_{q]}^{\text{X}}\right)\right]\nonumber \\
 & =\int_{\Sigma}\sum_{\text{X}}\left[d\boldsymbol{Z}^{\text{X}}\wedge\left(\delta+\mathscr{L}_{\xi_{\text{X}}}\right)\boldsymbol{\pi}^{\text{X}}-\left(d\iota_{\xi_{\text{X}}}+\iota_{\xi_{\text{X}}}d\right)\left(d\boldsymbol{Z}^{\text{X}}\wedge\boldsymbol{\pi}^{\text{X}}\right)\right]\nonumber \\
 & =\int_{\Sigma}\sum_{\text{X}}\left[\boldsymbol{Z}^{\text{X}}\wedge d\left(\delta+\mathscr{L}_{\xi_{\text{X}}}\right)\boldsymbol{\pi}^{\text{X}}-\iota_{\xi_{\text{X}}}d\left(d\boldsymbol{Z}^{\text{X}}\wedge\boldsymbol{\pi}^{\text{X}}\right)\right]\nonumber \\
 & =\int_{\Sigma}\sum_{\text{X}}\left[\boldsymbol{Z}^{\text{X}}\wedge\left(\delta+\mathscr{L}_{\xi_{\text{X}}}\right)\boldsymbol{w}^{\text{X}}-\iota_{\xi_{\text{X}}}\left(d\boldsymbol{Z}^{\text{X}}\wedge\boldsymbol{w}^{\text{X}}\right)\right],
\end{align}
where we have used the fact that the Lie derivative commutes with the exterior derivative, and $\boldsymbol{w}^{\text{X}}=d\boldsymbol{\pi}^{\text{X}}$ is the vorticity tensor.
\end{widetext}
Since the equations of motion Eqs. \eqref{eomn}-\eqref{eome} indicates that
\begin{equation}
\iota_{u_{\Upsilon}}\left(d\boldsymbol{Z}^{\Upsilon}\wedge\boldsymbol{w}^{\Upsilon}\right)=\iota_{u_{\Upsilon}}\boldsymbol{w}^{\Upsilon}\wedge d\boldsymbol{Z}^{\Upsilon}=0,
\end{equation}
and
\begin{equation}
\iota_{u_{\text{e}}}\left(d\boldsymbol{Z}^{\text{e}}\wedge\boldsymbol{w}^{\text{e}}\right)=\iota_{u_{\text{e}}}\boldsymbol{w}^{\text{e}}\wedge d\boldsymbol{Z}^{\text{e}}=Tds\wedge d\boldsymbol{Z}^{\text{e}}=0,
\end{equation}
where $\Upsilon=\text{n},\text{p}$ are the indices of ``super'' constituents, and we have used Eq. \eqref{condition of ZS} in the last step. Then as a top form, $d\boldsymbol{Z}^{\text{X}}\wedge\boldsymbol{w}^{\text{X}}=0$, which implies that the symplectic product of a trivial perturbation $\delta_{t}\phi$ with an arbitrary perturbation $\delta\phi$ is given by
\begin{equation}
W\left(\delta\phi,\delta_{t}\phi\right)=\int_{\Sigma}\sum_{\text{X}}\boldsymbol{Z}^{\text{X}}\wedge\left(\delta+\mathscr{L}_{\xi_{\text{X}}}\right)\boldsymbol{w}^{\text{X}}.\label{symplectic product with trivial perturbation:}
\end{equation}
Thus we can see that a sufficient condition for symplectic orthogonality to all trivial perturbations is vanishing ``Lagrangian-like'' change $\Delta_{\text{X}}\boldsymbol{w}^{\text{X}}=\left(\delta+\mathscr{L}_{\xi_{\text{X}}}\right)\boldsymbol{w}^{\text{X}}=0$.

Next, let us consider the case of axisymmetric trivial perturbations. It is evident from Eq. \eqref{trivial perturbation condition} that
\begin{equation}
\eta_{\text{X}}^{a}=U^{\text{X}}\varphi^{a},\label{axisymmetric tirvial displacement}
\end{equation}
is a trivial displacement for any axisymmetric function $U^{\text{X}}$ satisfying $\mathscr{L}_{u_{\text{X}}}U^{\text{X}}=0$. Although a general axisymmetric trivial displacement $\tilde{\eta}_{\text{X}}^{a}$ is still of the form Eq. \eqref{general nontrivial trivial displacement} with $\mathscr{L}_{\varphi}\boldsymbol{Z}^{\text{X}}=0$, the time derivative $\mathscr{L}_{t}\tilde{\eta}_{\text{X}}^{a}$ is always a trivial of the form Eq. \eqref{axisymmetric tirvial displacement}. To see this, taking the Lie derivative with respect to the timelike Killing vector $t^{a}$ and using the circular flow condition Eq. \eqref{circular flow condition}
\begin{align}
\mathscr{L}_{t}\tilde{\eta}_{\text{X}}^{a} & =\frac{1}{n_{\text{X}}^{2}}N^{\text{X}abc}\mathscr{L}_{t}\left(\nabla_{[b}Z_{c]}^{\text{X}}\right)\nonumber \\
 & =\frac{1}{n_{\text{X}}^{2}}N^{\text{X}abc}\mathscr{L}_{\vert v_{\text{X}}\vert u_{\text{X}}-\Omega_{\text{X}}\varphi}\left(d\boldsymbol{Z}^{\text{X}}\right)_{bc}\nonumber \\
 & =\frac{1}{n_{\text{X}}^{2}}N^{\text{X}abc}\left[d\left(\vert v_{\text{X}}\vert\iota_{u_{\text{X}}}d\boldsymbol{Z}^{\text{X}}-\Omega_{\text{X}}\iota_{\varphi}d\boldsymbol{Z}^{\text{X}}\right)\right]_{bc}\nonumber \\
 & =-\frac{2}{n_{\text{X}}^{2}}N^{\text{X}abc}\nabla_{b}\Omega_{\text{X}}\varphi^{d}\nabla_{[d}Z_{c]}^{\text{X}},
\end{align}
Since the contractions of both $\nabla_{b}\Omega_{\text{X}}$ and $\varphi^{d}\nabla_{[d}Z_{b]}^{\text{X}}$ with $\varphi^{b}$ vanish, it follows that $\mathscr{L}_{t}\tilde{\eta}_{\text{X}}^{a}$ must be proportional to $\varphi^{a}$, which establishes our claim.

In parallel to general case, the symplectic product of an axisymmetric trivial perturbation $\delta_{at}\phi=\left(0,0,\eta_{\text{X}}^{a}\right)$ of the form Eq. \eqref{axisymmetric tirvial displacement} with an arbitrary axisymmetric perturbation $\delta\phi$ is 
\begin{align}
 & W\left(\delta\phi,\delta_{at}\phi\right)\nonumber \\
= & \int_{\Sigma}\sum_{\text{X}}\left(U^{\text{X}}\varphi^{a}\delta\boldsymbol{P}_{a}^{\text{X}}-\left[\xi_{\text{X}},U^{\text{X}}\varphi\right]^{a}\boldsymbol{P}_{a}^{\text{X}}\right)\nonumber \\
= & \int_{\Sigma}\sum_{\text{X}}\left[U^{\text{X}}\delta\left(\varphi^{a}\boldsymbol{P}_{a}^{\text{X}}\right)-\mathscr{L}_{\xi_{\text{X}}}U^{\text{X}}\varphi^{a}\boldsymbol{P}_{a}^{\text{X}}\right]\nonumber \\
= & \int_{\Sigma}\sum_{\text{X}}\left(U^{\text{X}}\delta\boldsymbol{J}^{\text{X}}-\mathscr{L}_{\xi_{\text{X}}}U^{\text{X}}\boldsymbol{J}^{\text{X}}\right)\nonumber \\
= & \int_{\Sigma}\sum_{\text{X}}\left[U^{\text{X}}\left(\delta+\mathscr{L}_{\xi_{\text{X}}}\right)\boldsymbol{J}^{\text{X}}-\mathscr{L}_{\xi_{\text{X}}}\left(U^{\text{X}}\boldsymbol{J}^{\text{X}}\right)\right]\nonumber \\
= & \int_{\Sigma}\sum_{\text{X}}\left[U^{\text{X}}\Delta_{\text{X}}\boldsymbol{J}^{\text{X}}-\iota_{\xi_{\text{X}}}\left(dU^{\text{X}}\wedge\boldsymbol{J}^{\text{X}}\right)\right],
\end{align}
where we have used the axisymmetric condition $\mathscr{L}_{\varphi}\xi_{\text{X}}^{a}=0$ for axisymmetric perturbation $\delta\phi$ in the second step, the fact that the pullback of $\varphi^{a}\boldsymbol{\epsilon}_{abcd}$ onto the axisymmetric Cauchy surface $\Sigma$ vanishes, and $\boldsymbol{J}^{\text{X}}$ is the closed dual form of the conserved current $J_{\text{X}}^{a}$ given in Eq. \eqref{angular momentum current}. Since $dU^{\text{X}}\wedge\boldsymbol{J}^{\text{X}}$ is a top form and
\begin{equation}
\iota_{u_{\text{X}}}\left(dU^{\text{X}}\wedge\boldsymbol{J}^{\text{X}}\right)=\left(\iota_{u_{\text{X}}}dU^{\text{X}}\right)\boldsymbol{J}^{\text{X}}=\left(\mathscr{L}_{u_{\text{X}}}U^{\text{X}}\right)\boldsymbol{J}^{\text{X}}=0,
\end{equation}
so that $dU\wedge\boldsymbol{J}^{\text{X}}=0$. Thus we see that the symplectic product of axisymmetric trivial perturbation is 
\begin{equation}
W\left(\delta\phi,\delta_{at}\phi\right)=\int_{\Sigma}\sum_{\text{X}}U^{\text{X}}\Delta_{\text{X}}\boldsymbol{J}^{\text{X}},
\end{equation}
whence in the axisymmetric case, the necessary and sufficient condition for symplectic orthogonality to the trivial perturbations of the form $\eta^{a}=U^{\text{X}}\varphi^{a}$ is
\begin{equation}
\Delta_{\text{X}}\boldsymbol{J}^{\text{X}}=\left(\delta+\mathscr{L}_{\xi_{\text{X}}}\right)\boldsymbol{J}^{\text{X}}=0.\label{restriction in axisymmetric case}
\end{equation}

\section{Canonical energy and dynamic stability}

\label{sec6}

The dynamic stability we are concerned with is mode stability. That is to say, our superconducting-superfluid star in dynamic equilibrium is mode stable if there does not exists any non-pure-gauge linearized solution with the time dependence of the form $e^{kt}$ with $\text{Re}\left(k\right)>0$. Otherwise, it is said to be mode unstable. Rather than a complete analysis of linearized perturbation equations, one favorable way of proving mode stability is to construct a positive definite conserved norm on the space $\mathcal{C}$ of linearized on-shell perturbations, because it precludes those perturbations with exponential growth. A candidate is the canonical energy $\mathcal{E}$.

The \emph{canonical energy} $\mathcal{E}$ associated with the background timelike Killing vector $t^{a}$ is a bilinear form on the space $\mathcal{C}$ of linear on-shell perturbations defined as
\begin{equation}
\mathcal{E}\left(\delta_{1}\phi,\delta_{2}\phi\right)=W\left(\delta_{1}\phi,\mathscr{L}_{t}\delta_{2}\phi\right).
\end{equation}
It is easy to show that not only is the canonical energy symmetric and conserved. Indeed, the symmetry comes from that
\begin{align}
     &\mathcal{E}\left(\delta_1\phi,\delta_2\phi\right)-\mathcal{E}\left(\delta_2\phi,\delta_1\phi\right)\nonumber\\
    =&W\left(\delta_1\phi,\mathscr{L}_t\delta_2\phi\right)+W\left(\mathscr{L}_t\delta_1\phi,\delta_2\phi\right)\nonumber\\
    =&\int_\Sigma\left[\boldsymbol{\omega}\left(\delta_1\phi,\mathscr{L}_t\delta_2\phi\right)+\boldsymbol{\omega}\left(\mathscr{L}_t\delta_1\phi,\delta_2\phi\right)\right]\nonumber\\
    =&\int_\Sigma\mathscr{L}_t\boldsymbol{\omega}\left(\delta_1\phi,\delta_2\phi\right)\nonumber\\
    =&\int_\Sigma d\left[\iota_t\boldsymbol{\omega}\left(\delta_1\phi,\delta_2\phi\right)\right]\nonumber\\
    =&0.
\end{align}
where we have used $W$ is antisymmetric in the first step and $\boldsymbol{\omega}$ is closed for on-shell linearized solution in the fourth step. The conservation of $\mathcal{E}$ also comes from the closeness of $\omega$. Furthermore, the canonical energy is gauge invariant in the sense that 
\begin{equation}
\mathcal{E}\left(\delta_{1}\phi,\delta_{2}\phi\right)=\mathcal{E}\left(\delta_{1}\phi,\delta_{2}\phi+\mathscr{L}_{\mathcal{X}}\phi\right),
\end{equation}
with $\mathcal{X}^{a}$ smooth and of compact support. Moreover, since our star has a spatial compact support, then $\mathcal{E}\left(\delta\phi,\delta\phi\right)$ has a non-negative net flux at null infinity if the perturbation is asymptotically physically stationary (which will be defined below) at late times. This has been shown in \cite{Shi:2022jya}. 

A smooth linearized solution $\delta\phi=\left(\delta g_{ab},\delta A_{a},\xi_{\text{X}}^{a}\right)$ is said to be \emph{physically stationary} if the physical fields $\delta g_{ab}$, $\delta A_{a}$, $\delta\boldsymbol{N}^{\text{X}}$, and $\delta s$ can be made stationary by a gauge transformation, i.e., if there exists a smooth vector field $\mathcal{X}^{a}$, which is an asymptotic symmetry near infinity, such that
\begin{align}
0 & =\mathscr{L}_{t}\left(\delta g_{ab}+\mathscr{L}_{\mathcal{X}}g_{ab}\right),\\
0 & =\mathscr{L}_{t}\left(\delta A_{a}+\mathscr{L}_{\mathcal{X}}A_{a}\right),\\
0 & =\mathscr{L}_{t}\left(\delta\boldsymbol{N}^{\text{X}}+\mathscr{L}_{\mathcal{X}}\boldsymbol{N}^{\text{X}}\right)=-\mathscr{L}_{\left[t,\xi_{\text{X}}-\mathcal{X}\right]}\boldsymbol{N}^{\text{X}},\label{physically stationary of n}\\
0 & =\mathscr{L}_{t}\left(\delta s+\mathscr{L}_{\mathcal{X}}s\right)=-\mathscr{L}_{\left[t,\xi_{\text{X}}-\mathcal{X}\right]}s.\label{physically stationary of s}
\end{align}
Note that the last two equations Eqs. \eqref{physically stationary of n} and \eqref{physically stationary of s} means that the perturbation $\left(0,0,\left[t,\xi_{\text{X}}-\mathcal{X}\right]^{a}\right)$ is trivial perturbation with trivial displacements $\left[t,\xi_{\text{X}}-\mathcal{X}\right]^{a}$. So equivalently, a linearized solution $\delta\phi$ is physically stationary if and only if there exists a smooth vector field $X^{a}$, which is an asymptotic symmetry near infinity, such that
\begin{align}
\mathscr{L}_{t}\delta\phi= & \left(-\mathscr{L}_{\left[t,\mathcal{X}\right]}g_{ab},-\mathscr{L}_{\left[t,\mathcal{X}\right]}A_{a},\zeta_{\text{X}}^{a}=\left[t,\mathcal{X}\right]^{a}\right)\nonumber \\
 & +\left(0,0,\text{trivial displacements}\right).\label{physically stationary solution}
\end{align}
We shall use the notion $\delta_{\text{ps}}\phi$ to denote the physically stationary solutions.

As stated at the beginning of this section, if $\mathcal{E}$ provides a positive definite conserved norm, i.e., $\mathcal{E}\left(\delta\phi,\delta\phi\right)>0$ for all linearized solutions $\delta\phi$, then it implies the mode stability. However, since the physically stationary solutions are obviously physically stable, so we would also have mode stability if $\mathcal{E}\left(\delta\phi,\delta\phi\right)\geq0$ for all linearized solutions provided that $\mathcal{E}$ is degenerate only on physically stationary solutions. ($\mathcal{E}$ is said to be degenerate on $\delta\phi$ if $\mathcal{E}\left(\delta\phi,\delta^{\prime}\phi\right)=0$
for all $\delta^{\prime}\phi$ in the domain of $\mathcal{E}$.) Indeed, suppose that $\mathcal{E}\left(\delta\phi,\delta\phi\right)=0$ for some linearized solution $\delta\phi$ which is not physically stationary, then since it is not the degeneracy of $\mathcal{E}$, there must be some $\delta^{\prime}\phi$ in the domain of $\mathcal{E}$ such that $\mathcal{E}\left(\delta\phi,\delta^{\prime}\phi\right)\neq0$. Hence we have
\begin{equation}
\mathcal{E}\left(\delta^{\prime}\phi+\varepsilon\delta\phi,\delta^{\prime}\phi+\varepsilon\delta\phi\right)=\mathcal{E}\left(\delta^{\prime}\phi,\delta^{\prime}\phi\right)+2\varepsilon\mathcal{E}\left(\delta^{\prime}\phi,\delta\phi\right),
\end{equation}
and the left side can be negative by a suitable choice of $\varepsilon$, contradicts $\mathcal{E}\left(\delta\phi,\delta\phi\right)\geq0$ for all $\delta\phi$ in the domain of $\mathcal{E}$. Thus we see that if the degeneracy of $\mathcal{E}$ are only physically stationary solutions, then for all linearized solutions $\delta\phi$ which are not physically stationary, the non-negative definiteness of $\mathcal{E}$ implies that $\mathcal{E}\left(\delta\phi,\delta\phi\right)>0$, i.e., $\mathcal{E}$ provides a positive definite conserved norm on these perturbations, guaranteeing the mode stability. In other words, the non-negative definiteness of $\mathcal{E}$ is a sufficient condition for mode stability. In order to use the positivity of $\mathcal{E}$ serves also as a necessary condition of stability, i.e., the linearized solution $\delta\phi$ is instable in the alternative case where $\mathcal{E}\left(\delta\phi,\delta\phi\right)<0$, we further need that $\mathcal{E}$ is degenerate on not only, but all physically stationary perturbations. In such a case that $\mathcal{E}\left(\delta\phi,\delta\phi\right)<0$ for some linearized solution $\delta\phi$, suppose $\delta\phi$ asymptotically approached a physically stationary solution $\delta_{\text{ps}}\phi$ at late time, then the degeneracy of $\mathcal{E}$ on physically stationary solutions implies $\mathcal{E}\rightarrow0$, but this leads a contradiction as the positive net flux property of $\mathcal{E}$ indicates that $\mathcal{E}$ will become more negative at late times, so that $\delta\phi$ can not be stable. Thus, we need $\mathcal{E}$ to be degenerate precisely on the physically stationary solutions to use the positivity of $\mathcal{E}$ as a criterion for both stability and instability, where the non-negativity of $\mathcal{E}$ indicates mode stability, while the failure of non-negativity indicates the existence of linearized solutions that cannot asymptote to a physically stationary final state.

Unfortunately, we will find that $\mathcal{E}$ is not degenerate on all physically stationary solutions. In fact, since  $\mathcal{E}\left(\delta^{\prime}\phi,\delta\phi\right)=W\left(\delta^{\prime}\phi,\mathscr{L}_{t}\delta\phi\right)$, it follows that $\delta\phi$ is a degeneracy of $\mathcal{E}$ if and only if $\mathscr{L}_{t}\delta\phi$ is a degeneracy of $W$. As discussed at the end of Sec. \ref{subsec:Phase-space}, when restricted to $\mathcal{C}$, $W$ is degenerate precisely on the gauge transformations $\mathscr{L}_{\mathcal{X}}\phi$ with $\mathcal{X}^{a}$ smooth and going to zero at infinity , so that $\delta\phi$ in the domain of $\mathcal{E}$ is a degeneracy if and only if
\begin{equation}
\mathscr{L}_{t}\delta\phi=\left(\mathscr{L}_{\mathcal{X}}g_{ab},\mathscr{L}_{\mathcal{X}}A_{a},\xi_{\text{X}}^{a}=-\mathcal{X}^{a}\right).\label{degeneracy of E}
\end{equation}
By comparing Eqs. \eqref{physically stationary solution} and \eqref{degeneracy of E}, one can see that the degeneracy of $\mathcal{E}$ is a proper subset of all physically stationary solutions, and so that $\mathcal{E}$ fails to be degenerate on all physically stationary solutions. 

A way to avoid this obstacle is to restrict $\mathcal{E}$ to a smaller subspace of $\mathcal{C}$ such that it is degenerate on all physically stationary solutions. According to Eq. \eqref{physically stationary solution}
\begin{align}
\mathcal{E}\left(\delta\phi,\delta_{\text{ps}}\phi\right)=& W\left(\delta\phi,\mathscr{L}_t\delta_{\text{ps}}\phi\right)\nonumber\\
= & -W\left[\delta\phi,\mathscr{L}_{\left[t,\mathcal{X}\right]}\phi\right]\nonumber\\
&+W\left[\delta\phi,\left(0,0,\text{trivial displacements}\right)\right].
\end{align}
For a general asymptotic symmetry $\mathcal{X}^{a}$, the commutator $\left[t,\mathcal{X}\right]^{a}$ is, at most, an asymptotic translation (as occurs when $\mathcal{X}^{a}$ is an asymptotic boost). Therefore, in order the first term of right side vanishes, we need to restrict $\delta\phi$ such that $\delta H_{\left[t,\mathcal{X}\right]}=0$, where $\delta H_{\left[t,\mathcal{X}\right]}$ is the ADM linear momentum (see Eq. \eqref{delta Hamiltonian}). This in fact does not impose a physical restriction on the perturbations as we can achieve this by addition of the action of an infinitesimal Lorentz boost on the background solution. On the other hands, to make the second term vanishes, we
need to restrict $\delta\phi$ such that 
\begin{equation}
W\left(\delta\phi,\text{trivial perturbation}\right)=0,\label{restriction of C}
\end{equation}
for all trivial perturbations. As a consequence, let us define a subspace $\mathcal{V}\subset\mathcal{C}$ of the linearized solutions composed of perturbations that have $\delta H_{\left[t,\mathcal{X}\right]}=0$ and are symplectically orthogonal to all trivial perturbations, in other words, $\mathcal{V}$ is the symplectic complement of the subspace $\mathcal{W}$ spanned by the perturbations of the form $\mathscr{L}_{t}\delta_{\text{ps}}\phi$. Note that the double symplectic complement of $\mathcal{W}$ is itself according to Eq. \eqref{double symplectic complement}, thus the symplectic complement of $\mathcal{V}$ in itself is $\mathcal{V}\cap\mathcal{W}$, which implies that the degeneracy of the canonical energy $\mathcal{E}$ is given precisely by the physically stationary solutions when restricted onto $\mathcal{V}$. 

Putting together all discussions above, we have the following criterion for the dynamic stability of our superconducting-superfluid star:
\begin{thm}
\label{thm:dynamic stability criterion}If $\mathcal{E}$ is non-negative on the subspace $\mathcal{V}$ of perturbations, then one has stability in the sense that there do not exist any exponentially growing modes lying in this subspace with respect to the perturbations within $\mathcal{V}$. Conversely, if $\mathcal{E}\left(\delta\phi,\delta\phi\right)<0$ for some $\delta\phi\in\mathcal{V}$, then one has instability in the sense that such a $\delta\phi$ cannot approach a physically stationary solution at asymptotically late times.
\end{thm}

Let us consider the restriction condition given in Eq. \eqref{restriction of C}. For non-axisymmetric perturbations, this restricted condition is in fact same as the condition to make Eq. \eqref{symplectic product with trivial perturbation:} vanish. If the superfluid neutrons and superconducting protons in our star are irrotational, i.e., their vorticity vanish $\boldsymbol{w}^{\Upsilon}=0$, and so that their total momentum covector should be of the form
\begin{equation}
\pi_{a}^{\Upsilon}=\frac{\hbar}{2}\nabla_{a}\vartheta^{\Upsilon},
\end{equation}
where the locally defined potentials $\vartheta^{\text{n}},\vartheta^{\text{p}}$ are interpretable as phase angles associated with underlying boson condensates, the factor 2 in the denominators are included to allow for the fact that the relevant bosons are presumed to consist not of single protons or neutrons but of Cooper type pairs. In such a case, as $\delta\boldsymbol{w}^{\Upsilon}=\frac{\hbar}{2}d^{2}\delta\vartheta^{\Upsilon}=0$, Eq. \eqref{symplectic product with trivial perturbation:} reduces to
\begin{equation}
W\left(\delta\phi,\delta_{t}\phi\right)=\int_{\Sigma}\boldsymbol{Z}^{\text{e}}\wedge\Delta_{\text{e}}\boldsymbol{w}^{\text{e}}.
\end{equation}
As shown in \cite{Shi:2022jya}, if we further focus onto the background in which $\nabla_{a}s\neq0$, then the condition Eq. \eqref{restriction of C} does not lead to a real physical restriction to $\mathcal{V}$ for non-axisymmetric perturbations, since it can be achieved in the suitable background solution by adding a trivial perturbation. However, if there are vortex in our superconducting-superfluid star and so that the irrotation is violated, then whether the condition Eq. \eqref{restriction of C} will impose a physical restriction to $\mathcal{V}$ is still unknown to our knowledge.

Different from the non-axisymmetric perturbations, in the axisymmetric case, the restriction to $\mathcal{V}$ does impose a physical restriction on perturbations. In particular, as shown in Eq. \eqref{restriction in axisymmetric case}, symplectic orthogonality to axisymmetric trivial perturbations of the form $U^{\text{X}}\varphi^{a}$ requires that $\left(\delta+\mathscr{L}_{\xi_{\text{X}}}\right)\boldsymbol{J}^{\text{X}}=0$, which is a significant physical restriction. So when we consider such physical restrictions on perturbations, Theorem \ref{thm:dynamic stability criterion} becomes limited since it gives stability criterion only for perturbations in the restricted subspace $\mathcal{V}$. But fortunately, in the axisymmetric case, the mode stability for perturbations in $\mathcal{V}$
in fact will imply the mode stability for all perturbations, including those that cannot be described within the Lagrangian displacement framework. This result is a direct consequence of the following lemma:
\begin{lem}
Let $\delta\phi=\left(\delta g_{ab},\delta A_{a},\delta\boldsymbol{N}^{\text{X}},\delta s\right)$ be an axisymmetric solution to the linearized Einstein-superconducting-superfluid equations (not necessarily arising in the Lagrangian displacement framework). Then there exist vector fields $\xi_{\text{X}}^{a}$ such that
\begin{align}
\mathscr{L}_{t}\delta\boldsymbol{N}^{\text{X}} & =-\mathscr{L}_{\xi_{\text{X}}}\boldsymbol{N}^{\text{X}},\nonumber \\
\mathscr{L}_{t}\delta s & =-\mathscr{L}_{\xi_{\text{e}}}s,\nonumber \\
\mathscr{L}_{t}\delta\boldsymbol{J}^{\text{X}} & =-\mathscr{L}_{\xi_{\text{X}}}\boldsymbol{J}^{\text{X}}.
\end{align}
Thus, $\mathscr{L}_{t}\delta\phi$ can be represented in the Lagrangian displacement framework and has $\Delta_{\text{X}}\boldsymbol{J}^{\text{X}}=0$. Furthermore, $\mathscr{L}_{t}^{2}\delta\phi\in\mathcal{V}$.
\end{lem}

\begin{proof}
Let
\begin{equation}
\xi_{\text{X}}^{a}=\vert v_{\text{X}}\vert\delta u_{\text{X}}+\beta_{\text{X}}\varphi^{a},
\end{equation}
where $v_{\text{X}}^{a}=t^{a}+\Omega_{\text{X}}\varphi^{a}$ and $\beta_{\text{X}}$ is any axisymmetric scalar such that
\begin{equation}
u_{\text{X}}^{a}\nabla_{a}\beta_{\text{X}}=\delta u_{\text{X}}^{a}\nabla_{a}\Omega_{\text{X}}.
\end{equation}
The perturbation of the conservation law of entropy yields
\begin{align}
0 & =\delta\left(u_{\text{e}}^{a}\nabla_{a}s\right)\nonumber \\
 & =\delta u_{\text{e}}^{a}\nabla_{a}s+u_{\text{e}}^{a}\nabla_{a}\delta s\nonumber \\
 & =\frac{1}{\vert v_{\text{e}}\vert}\left[\left(\xi_{\text{e}}^{a}-\beta_{\text{e}}\varphi^{a}\right)\nabla_{a}s+\left(t^{a}+\Omega_{\text{e}}\varphi^{a}\right)\nabla_{a}\delta s\right]\nonumber \\
 & =\frac{1}{\vert v_{\text{e}}\vert}\left(\xi_{\text{e}}^{a}\nabla_{a}s+t^{a}\nabla_{a}\delta s\right),
\end{align}
where we suppose that our star is already in dynamic equilibrium, we have used the circular flow condition Eq. \eqref{circular flow condition}, and for axisymmetric perturbation $\mathscr{L}_{\varphi}\delta s=0$. Hence we have
\begin{equation}
\mathscr{L}_{t}\delta s=-\mathscr{L}_{\xi_{\text{e}}}s.
\end{equation}
The perturbation of the conservation law of particle number yields
\begin{equation}
\delta\left(d\boldsymbol{N}^{\text{X}}\right)=d\left(\delta\boldsymbol{N}^{\text{X}}\right)=0,
\end{equation}
so that
\begin{align}
\mathscr{L}_{t}\delta\boldsymbol{N}^{\text{X}} & =d\left(\iota_{t}\delta\boldsymbol{N}^{\text{X}}\right)\nonumber \\
 & =d\left[\left(\vert v_{\text{X}}\vert u_{\text{X}}^{a}-\Omega_{\text{X}}\varphi^{a}\right)\delta N_{abc}^{\text{X}}\right]\nonumber \\
 & =-d\left[\vert v_{\text{X}}\vert\delta u_{\text{X}}^{a}N_{abc}^{\text{X}}+\Omega_{\text{X}}\varphi^{a}\delta N_{abc}^{\text{X}}\right]\nonumber \\
 & =-d\left[\left(\xi_{\text{X}}^{a}-\beta_{\text{X}}\varphi^{a}\right)N_{abc}^{\text{X}}+\Omega_{\text{X}}\varphi^{a}\delta N_{abc}^{\text{X}}\right]\nonumber \\
 & =-d\left[\iota_{\xi_{\text{X}}}\boldsymbol{N}^{\text{X}}-\iota_{\varphi}\left(\beta_{\text{X}}\boldsymbol{N}^{\text{X}}-\Omega_{\text{X}}\delta\boldsymbol{N}^{\text{X}}\right)\right]\nonumber \\
 & =-\mathscr{L}_{\xi_{\text{X}}}\boldsymbol{N}^{\text{X}}+d\left[\iota_{\varphi}\left(\beta_{\text{X}}\boldsymbol{N}^{\text{X}}-\Omega_{\text{X}}\delta\boldsymbol{N}^{\text{X}}\right)\right]\nonumber \\
 & =-\mathscr{L}_{\xi_{\text{X}}}\boldsymbol{N}^{\text{X}}-\iota_{\varphi}d\left(\beta_{\text{X}}\boldsymbol{N}^{\text{X}}-\Omega_{\text{X}}\delta\boldsymbol{N}^{\text{X}}\right),
\end{align}
where we have used that $\delta\left(u_{\text{X}}^{a}N_{abc}^{\text{X}}\right)=0$ in the third step, and $\mathscr{L}_{\varphi}\left(\beta_{\text{X}}\boldsymbol{N}^{\text{X}}-\Omega_{\text{X}}\delta\boldsymbol{N}^{\text{X}}\right)=0$ in the last step. Since $d\left(\beta_{\text{X}}\boldsymbol{N}^{\text{X}}-\Omega_{\text{X}}\delta\boldsymbol{N}^{\text{X}}\right)$ is a top form, then
\begin{align}
 & \iota_{\varphi}d\left(\beta_{\text{X}}\boldsymbol{N}^{\text{X}}-\Omega_{\text{X}}\delta\boldsymbol{N}^{\text{X}}\right)=0,\nonumber \\
\Longleftrightarrow & d\left(\beta_{\text{X}}\boldsymbol{N}^{\text{X}}-\Omega_{\text{X}}\delta\boldsymbol{N}^{\text{X}}\right)=0,\nonumber \\
\Longleftrightarrow & d\beta_{\text{X}}\wedge\boldsymbol{N}^{\text{X}}=d\Omega_{\text{X}}\wedge\delta\boldsymbol{N}^{\text{X}},\nonumber \\
\Longleftrightarrow & \boldsymbol{\epsilon}^{abcd}\nabla_{a}\beta_{\text{X}}N_{bcd}^{\text{X}}=\boldsymbol{\epsilon}^{abcd}\nabla_{a}\Omega_{\text{X}}\delta N_{bcd}^{\text{X}},\nonumber \\
\Longleftrightarrow & 
\boldsymbol{\epsilon}^{abcd}\nabla_{a}\beta_{\text{X}}\boldsymbol{\epsilon}_{ebcd}n_{\text{X}}^e=\boldsymbol{\epsilon}^{abcd}\nabla_{a}\Omega_{\text{X}}\delta\left( \boldsymbol{\epsilon}_{ebcd}n_{\text{X}}^e\right),\nonumber\\
\Longleftrightarrow & n_{\text{X}}^{a}\nabla_{a}\beta_{\text{X}}=\nabla_{a}\Omega_{\text{X}}\left(n_{\text{X}}\delta u_{\text{X}}^{a}+u_{\text{X}}^{a}\delta n_{\text{X}}+\frac{1}{2}n_{\text{X}}^{a}g^{bc}\delta g_{bc}\right),\nonumber \\
\Longleftrightarrow & u_{\text{X}}^{a}\nabla_{a}\beta_{\text{X}}=\left(\delta u_{\text{X}}^{a}\right)\nabla_{a}\Omega_{\text{X}},\label{proof of n part of lemma}
\end{align}
where we have used that $u_{\text{X}}^{a}\nabla_{a}\Omega_{\text{X}}=0$ due to dynamic equilibrium and circular flow condition. But since we defined $\beta_{\text{X}}$ such the last equality holds, so we have shown that
\begin{equation}
\mathscr{L}_{t}\delta\boldsymbol{N}^{\text{X}}=-\mathscr{L}_{\xi_{\text{X}}}\boldsymbol{N}^{\text{X}}.
\end{equation}
For the conservation law of angular momentum, since the perturbation is the solution to linearized equations of motion, then the axisymmetric perturbation of Eq. \eqref{conservation of j} gives that
\begin{align}
\delta\left(\nabla_{a}J^{a}\right) & =\delta n_{\text{X}}^{a}\mathscr{L}_{\varphi}\pi_{a}^{\text{X}}+n_{\text{X}}^{a}\mathscr{L}_{\varphi}\delta\pi_{a}^{\text{X}}+\varphi^{b}\delta\left(n_{\text{X}}^{a}w_{ab}^{\text{X}}\right)\nonumber \\
 & =\delta_{\text{X}}^{\text{e}}\varphi^{b}\delta\left(Tn_{\text{e}}\nabla_{b}s\right)\nonumber \\
 & =0.
\end{align}
So we still have the perturbed conservation of angular momentum, i.e., $\delta\left(d\boldsymbol{J}^{\text{X}}\right)=d\left(\delta\boldsymbol{J}^{\text{X}}\right)=0$. Replacing $\boldsymbol{N}^{\text{X}}$ by $\boldsymbol{J}^{\text{X}}$, then an identical calculation, besides the sixth equality in Eq. \eqref{proof of n part of lemma} replaced by
\begin{align}
J_{\text{X}}^{a}\nabla_{a}\beta_{\text{X}}= & \varphi^{b}\pi_{b}^{\text{X}}\nabla_{a}\Omega_{\text{X}}\left(n_{\text{X}}\delta u_{\text{X}}^{a}+u_{\text{X}}^{a}\delta n_{\text{X}}+\frac{1}{2}n_{\text{X}}^{a}g^{bc}\delta g_{bc}\right)\nonumber \\
 & +u_{\text{X}}^{a}\nabla_{a}\Omega_{\text{X}}\delta\left(\varphi^{b}\pi_{b}^{\text{X}}\right),
\end{align}
 shows that
\begin{equation}
\mathscr{L}_{t}\delta\boldsymbol{J}^{\text{X}}=-\mathscr{L}_{\xi_{\text{X}}}\boldsymbol{J}^{\text{X}}.
\end{equation}
Thus we have shown that $\mathscr{L}_{t}\delta\phi$ can be represented in Lagrangian displacement framework as 
\begin{equation}
\mathscr{L}_{t}\delta\phi=\delta^{\prime}\phi=\left(\mathscr{L}_{t}\delta g_{ab},\mathscr{L}_{t}\delta A_{a},\xi_{\text{X}}^{a}\right),
\end{equation}
and has 
\begin{equation}
\Delta_{\text{X}}\boldsymbol{J}^{\text{X}}=\left(\delta^{\prime}+\mathscr{L}_{\xi_{\text{X}}}\right)\boldsymbol{J}^{\text{X}}=0.
\end{equation}

Clearly, if $\delta\phi$ is an axisymmetric solution to the linearized Einstein-superconducting-superfluid equations, then so is $\mathscr{L}_{t}\delta\phi$. And let $\eta_{\text{X}}^{a}$ be any axisymmetric trivial displacement, as discussed in Sec. \ref{subsec:Trivial-displacements}, $\mathscr{L}_{t}\eta_{\text{X}}^{a}$ is of the form $U^{\text{X}}\varphi^{a}$. Then we have
\begin{align}
W\left[\left(0,0,\eta_{\text{X}}^{a}\right),\mathcal{\mathscr{L}}_{t}^{2}\delta\phi\right] & =\mathcal{E}\left[\left(0,0,\eta_{\text{X}}^{a}\right),\mathscr{L}_{t}\delta\phi\right]\nonumber \\
 & =\mathcal{E}\left[\mathscr{L}_{t}\delta\phi,\left(0,0,\eta_{\text{X}}^{a}\right)\right]\nonumber \\
 & =W\left[\mathscr{L}_{t}\delta\phi,\left(0,0,\mathscr{L}_{t}\eta_{\text{X}}^{a}\right)\right]\nonumber \\
 & =0,
\end{align}
where the second equality follows from the symmetry of $\mathcal{E}$, and the last equality follows from $\mathscr{L}_{t}\delta\phi$ satisfies $\Delta_{\text{X}}\boldsymbol{J}^{\text{X}}=0$. So $\mathscr{L}_{t}^{2}\delta\phi$ is symplectically orthogonal to all axisymmetric trivial perturbations. Furthermore, for a smooth vector field $\mathcal{X}^{a}$ which is an asymptotic symmetry near infinity, since
\begin{align}
W\left(\mathscr{L}_{\left[t,\mathcal{X}\right]}\phi,\mathscr{L}_{t}^{2}\delta\phi\right)= & \mathcal{E}\left(\mathscr{L}_{\left[t,\mathcal{X}\right]}\phi,\mathscr{L}_{t}\delta\phi\right)\nonumber \\
= & \mathcal{E}\left(\mathscr{L}_{t}\delta\phi,\mathscr{L}_{\left[t,\mathcal{X}\right]}\phi\right)\nonumber \\
= & W\left(\mathscr{L}_{t}\delta\phi,\mathscr{L}_{\left[t,\left[t,\mathcal{X}\right]\right]}\phi\right)\nonumber \\
= & 0,
\end{align}
where we have used the fact that $\mathcal{Y}^{a}=\left[t,\left[t,\mathcal{X}\right]\right]^{a}$ vanishes at infinity for any asymptotic symmetry generator $\mathcal{X}^{a}$. Thus, we have shown that $\mathscr{L}_{t}^{2}\delta\phi\in\mathcal{V}$.
\end{proof}
Now, if the axisymmetric perturbation $\delta\phi$, which may not be described in the Lagrangian displacement framework, has a exponentially growth in time, then so does $\mathscr{L}_{t}^{2}\delta\phi$. Therefore, the absence of exponentially growing solutions of $\mathscr{L}_{t}^{2}\delta\phi$ implies the absence of any exponentially growing solutions of $\delta\phi$ at all. Thus when applying in the axisymmetric case, the stability criterion in Theorem \ref{thm:dynamic stability criterion} implies the following result:
\begin{thm}
\label{thm:dynamic stability in axisymmetric case}If $\mathcal{E}$ is non-negative on the subspace of axisymmetric perturbations in $\mathcal{V}$, then there are no smooth, axisymmetric solutions to the Einstein-superconducting-superfluid equations with suitable fall-off condition at infinity that have exponential growth in time, i.e., mode stability holds for all axisymmetric perturbations. Conversely, if $\mathcal{E}\left(\delta\phi,\delta\phi\right)<0$ for some axisymmetric $\delta\phi\in\mathcal{V}$, then one has instability in the sense as in Theorem \ref{thm:dynamic stability criterion}.
\end{thm}

\section{Thermodynamic stability}

\label{sec7}

Now let us consider the thermodynamic stability of stars in thermodynamic equilibrium. Our superconducting-superfluid star in weak thermodynamic equilibrium is said to be \emph{weakly thermodynamically stable}\footnote{Our definition for thermodynamic stability only makes sense when our star is already in thermodynamic equilibrium, otherwise the first order change of total entropy will in general do not vanish even $M$, $N^{\text{X}}$, and $J^{\text{X}}$ are fixed under first order perturbations.} if $\delta^{2}S<0$ for all linearized solutions with 
\begin{equation}
\delta M=\delta N^{\text{X}}=\delta J^{\text{X}}=\delta^{2}M=\delta^{2}N^{\text{X}}=\delta^{2}J^{\text{X}}=0.\label{condition for thermodynamic stabiltiy}
\end{equation}
Define
\begin{equation}
\mathcal{E}^{\prime}\equiv\delta^{2}M-\sum_{\text{X}}\tilde{\mu}_{\text{X}}\delta^{2}N^{\text{X}}-\tilde{T}\delta^{2}S-\sum_{\text{X}}\Omega_{\text{X}}\delta^{2}J^{\text{X}},
\end{equation}
then it is evident that our star has weak thermodynamic stability if and only if $\mathcal{E}^{\prime}$ is positive for all linearized solutions with Eq. \eqref{condition for thermodynamic stabiltiy} (we assume the redshifted temperature $\tilde{T}\geq0$). But by taking the perturbation of the first law Eq. \eqref{first law}, we find that
\begin{equation}
\mathcal{E}^{\prime}=\int_{\Sigma}\left(\sum_{\text{X}}\delta\tilde{\mu}_{\text{X}}\delta\boldsymbol{N}^{\text{X}}+\delta\tilde{T}\delta\boldsymbol{S}+\sum_{\text{X}}\delta\Omega_{\text{X}}\delta\boldsymbol{J}^{\text{X}}\right),
\end{equation}
which tells us that $\mathcal{E}^{\prime}$ is only depend on the first order perturbation of the star and independent of the choice of second order perturbation. Thus the weak thermodynamic stability is equivalent to the positivity of $\mathcal{E}^{\prime}$ for all perturbations for linear on-shell perturbations only with $\delta M=\delta N^{\text{X}}=\delta J^{\text{X}}=0$. 

As a consequence, when we restrict the perturbations in the Lagrangian displacement framework (where $\delta N^{\text{X}}=\delta S=0$ holds automatically) such that $\delta J^{\text{X}}=0$, since in such a case we must also have $\delta M=0$ because of the weak thermodynamic equilibrium, the positivity of $\mathcal{E}^{\prime}$ becomes just a necessary condition for weak thermodynamic stability, and $\mathcal{E}^{\prime}$ takes the form
\begin{equation}
\mathcal{E}^{\prime}=\delta^{2}M-\sum_{\text{X}}\Omega_{\text{X}}\delta^{2}J^{\text{X}}.\label{thermodynamic stability criterion}
\end{equation}
As shown in \cite{Green:2013ica}, there can be a non-axisymmetric perturbation which is made to have $\sum_{\text{X}}\Omega_{\text{X}}\delta^{2}J^{\text{X}}>\delta^{2}M$, whence $\mathcal{E}^{\prime}<0$, in other words, all rotating stars are thermodynamically unstable with respect to the non-axisymmetric perturbations. So we will solely consider the axisymmetric perturbations within the Lagrangian displacement framework such that $\delta J^{\text{X}}=0$ below, and we finally show that in the axisymmetric case, the canonical energy $\mathcal{E}\left(\delta\phi,\delta\phi\right)$ coincides with $\mathcal{E}^{\prime}$.

To achieve this, let us first decompose the Lagrangian Eq. \eqref{Lagrangian for star} into the Einstein-Maxwell part and the matter part as follows
\begin{align}
\boldsymbol{\mathcal{L}} & =\boldsymbol{\mathcal{L}}_{EM}+\boldsymbol{\mathcal{L}}_{\text{M}},\nonumber \\
\boldsymbol{\mathcal{L}}_{EM} & =\boldsymbol{\epsilon}\left(R-\frac{1}{4}F_{ab}F^{ab}\right),\nonumber \\
\boldsymbol{\mathcal{L}}_{\text{M}} & =\boldsymbol{\epsilon}\left(j^{a}A_{a}+\Lambda_{\text{M}}\right).
\end{align}
With above decomposition, the symplectic form will also split into two parts
\begin{align}
 & \boldsymbol{\omega}\left(\phi;\delta_{1}\phi,\delta_{2}\phi\right)\nonumber \\
= & \boldsymbol{\omega}^{\left(EM\right)}\left(\phi;\delta_{1}\phi,\delta_{2}\phi\right)+\boldsymbol{\omega}^{\left(\text{M}\right)}\left(\phi;\delta_{1}\phi,\delta_{2}\phi\right)\nonumber \\
= & \boldsymbol{\omega}^{\left(EM\right)}\left[\phi;\delta_{1}\phi,\left(\delta_{2}g_{ab},\delta_{2}A_{a}\right)\right]+\boldsymbol{\omega}^{\left(\text{M}\right)}\left(\phi;\delta_{1}\phi,\delta_{2}\phi\right).
\end{align}
Hence, for the axisymmetric on-shell perturbation $\delta\phi$, the canonical energy $\mathcal{E}$ reads
\begin{widetext}
\begin{align}
 & \mathcal{E}\left(\delta\phi,\delta\phi\right)\nonumber \\
= & W\left(\delta\phi,\mathscr{L}_{t}\delta\phi\right)\nonumber \\
= & \int_{\Sigma}\left\{ \boldsymbol{\omega}^{\left(EM\right)}\left[\delta\phi,\left(\mathscr{L}_{t}\delta g_{ab},\mathscr{L}_{t}\delta A_{a}\right)\right]+\boldsymbol{\omega}^{\left(\text{M}\right)}\left(\delta\phi,\mathscr{L}_{t}\delta\phi\right)\right\} \nonumber \\
= & \int_{\Sigma}\left\{ \delta\boldsymbol{\omega}^{\left(EM\right)}\left[\delta\phi,\left(\mathscr{L}_{t}g_{ab},\mathscr{L}_{t}A_{a}\right)\right]+\delta\boldsymbol{\omega}^{\left(\text{M}\right)}\left(\delta\phi,\mathscr{L}_{t}\phi\right)\right\}+\int_{\Sigma}\left[\boldsymbol{\omega}^{\left(\text{M}\right)}\left(\delta\phi,\mathscr{L}_{t}\delta\phi\right)-\delta\boldsymbol{\omega}^{\left(\text{M}\right)}\left(\delta\phi,\mathscr{L}_{t}\phi\right)\right]\nonumber \\
= & \int_{\Sigma}\delta\boldsymbol{\omega}\left(\delta\phi,\mathscr{L}_{t}\phi\right)+\int_{\Sigma}\left[\boldsymbol{\omega}^{\left(\text{M}\right)}\left(\delta\phi,\mathscr{L}_{t}\delta\phi\right)-\delta\boldsymbol{\omega}^{\left(\text{M}\right)}\left(\delta\phi,\mathscr{L}_{t}\phi\right)\right]\nonumber \\
= & \int_{\Sigma}\delta\left\{ \iota_{t}\boldsymbol{E}\cdot\delta\phi+\delta\boldsymbol{C}_{t}+d\left[\delta\boldsymbol{Q}_{t}-\iota_{t}\boldsymbol{\theta}\left(\phi;\delta\phi\right)\right]\right\}+\int_{\Sigma}\left[\boldsymbol{\omega}^{\left(\text{M}\right)}\left(\delta\phi,\mathscr{L}_{t}\delta\phi\right)-\delta\boldsymbol{\omega}^{\left(\text{M}\right)}\left(\delta\phi,\mathscr{L}_{t}\phi\right)\right]\nonumber \\
= & \delta^{2}M+\int_{\Sigma}\left[\boldsymbol{\omega}^{\left(\text{M}\right)}\left(\delta\phi,\mathscr{L}_{t}\delta\phi\right)-\delta\boldsymbol{\omega}^{\left(\text{M}\right)}\left(\delta\phi,\mathscr{L}_{t}\phi\right)\right],
\end{align}
where we have used the fundamental identity Eq. \eqref{fundamental identity} in the fifth step, the fact that $\delta\phi$ is on-shell perturbation and Eq. \eqref{ADM mass} in the last step. A direct calculation shows that
\begin{align}
 & \boldsymbol{\omega}^{\left(\text{M}\right)}\left(\delta\phi,\mathscr{L}_{t}\delta\phi\right)= \sum_{\text{X}}\left(\mathscr{L}_{t}\xi_{\text{X}}^{a}\delta\boldsymbol{P}_{a}^{\text{X}}-\xi_{\text{X}}^{a}\mathscr{L}_{t}\delta\boldsymbol{P}_{a}^{\text{X}}-\left[\xi_{\text{X}},\mathscr{L}_{t}\xi_{\text{X}}\right]^{a}\boldsymbol{P}_{a}^{\text{X}}\right),
\end{align}
and
\begin{equation}
\boldsymbol{\omega}^{\left(\text{M}\right)}\left(\delta\phi,\mathscr{L}_{t}\phi\right)=-t^{a}\sum_{\text{X}}\delta\boldsymbol{P}_{a}^{\text{X}}-\mathscr{L}_{t}\left(\sum_{\text{X}}\xi_{\text{X}}^{a}\boldsymbol{P}_{a}^{\text{X}}\right),
\end{equation}
so that
\begin{align}
 & \boldsymbol{\omega}^{\left(\text{M}\right)}\left(\delta\phi,\mathscr{L}_{t}\delta\phi\right)-\delta\boldsymbol{\omega}^{\left(\text{M}\right)}\left(\delta\phi,\mathscr{L}_{t}\phi\right)= \sum_{\text{X}}\left(t^{a}\delta^{2}\boldsymbol{P}_{a}^{\text{X}}+2\mathscr{L}_{t}\xi_{\text{X}}^{a}\delta\boldsymbol{P}_{a}^{\text{X}}-\left[\xi_{\text{X}},\mathscr{L}_{t}\xi_{\text{X}}\right]^{a}\boldsymbol{P}_{a}^{\text{X}}\right).
\end{align}
Since $\delta\phi$ is axisymmetric, and Theorem \ref{thm:weakly thermodynamic equilibrium} implies that $\Omega_{\text{X}}$ are uniform throughout the star, so that
\begin{equation}
\mathscr{L}_{\xi_{\text{X}}}\left(\Omega_{\text{X}}\varphi^{a}\right)=\left(\iota_{\xi_{\text{X}}}d\Omega_{\text{X}}\right)\varphi^{a}-\Omega_{\text{X}}\mathscr{L}_{\varphi}\xi_{\text{X}}^{a}=0.
\end{equation}
Accordingly, by using the circular flow condition, we have
\begin{equation}
\mathscr{L}_{\xi_{\text{X}}}t^{a}=\mathscr{L}_{\xi_{\text{X}}}\left(\vert v_{\text{X}}\vert u_{\text{X}}^{a}-\Omega_{\text{X}}\varphi^{a}\right)=\mathscr{L}_{\xi_{\text{X}}}\left(\vert v_{\text{X}}\vert u_{\text{X}}^{a}\right),
\end{equation}
whence
\begin{align}
 & \boldsymbol{\omega}^{\left(\text{M}\right)}\left(\delta\phi,\mathscr{L}_{t}\delta\phi\right)-\delta\boldsymbol{\omega}^{\left(\text{M}\right)}\left(\delta\phi,\mathscr{L}_{t}\phi\right)\nonumber \\
= & \sum_{\text{X}}\left[\vert v_{\text{X}}\vert u_{\text{X}}^{a}\delta^{2}\boldsymbol{P}_{a}^{\text{X}}-\Omega_{\text{X}}\varphi^{a}\delta^{2}\boldsymbol{P}_{a}^{\text{X}}-2\mathscr{L}_{\xi_{\text{X}}}\left(\vert v_{\text{X}}\vert u_{\text{X}}^{a}\right)\delta\boldsymbol{P}_{a}^{\text{X}}+\mathscr{L}_{\xi_{\text{X}}}\mathscr{L}_{\xi_{\text{X}}}\left(\vert v_{\text{X}}\vert u_{\text{X}}^{a}\right)\boldsymbol{P}_{a}^{\text{X}}\right].\label{omegaM-delta omega}
\end{align}
Note that according to $u_{\text{X}}^{a}\boldsymbol{P}_{a}^{\text{X}}=0$ , one has $\delta\left(u_{\text{X}}^{a}\boldsymbol{P}_{a}^{\text{X}}\right)=0$ and $\delta^{2}\left(u_{\text{X}}^{a}\boldsymbol{P}_{a}^{\text{X}}\right)=0$, i.e., respectively,
\begin{equation}
u_{\text{X}}^{a}\delta\boldsymbol{P}_{a}^{\text{X}}=-\delta u_{\text{X}}^{a}\boldsymbol{P}_{a}^{\text{X}}=\mathscr{L}_{\xi_{\text{X}}}u_{\text{X}}^{a}\boldsymbol{P}_{a}^{\text{X}},
\end{equation}
where we have used Eq. \eqref{u variation}, and
\begin{align}
 & u_{\text{X}}^{a}\delta^{2}\boldsymbol{P}_{a}^{\text{X}}\nonumber \\
= & -\delta^{2}u_{\text{X}}^{a}\boldsymbol{P}_{a}^{\text{X}}-2\delta u_{\text{X}}^{a}\delta\boldsymbol{P}_{a}^{\text{X}}\nonumber \\
= & -\delta\left(\frac{1}{2}u_{\text{X}}^{a}u_{\text{X}}^{b}u_{\text{X}}^{c}\delta g_{bc}-q_{\text{X}b}^{a}\mathscr{L}_{\xi_{\text{X}}}u_{\text{X}}^{b}\right)\boldsymbol{P}_{a}^{\text{X}}-\left(u_{\text{X}}^{b}u_{\text{X}}^{c}\delta g_{bc}-2u_{\text{X}b}\mathscr{L}_{\xi_{\text{X}}}u_{\text{X}}^{b}\right)u_{\text{X}}^{a}\delta\boldsymbol{P}_{a}^{\text{X}}+2\mathscr{L}_{\xi_{\text{X}}}u_{\text{X}}^{a}\delta\boldsymbol{P}_{a}^{\text{X}}\nonumber \\
= & -\left(\frac{1}{2}u_{\text{X}}^{b}u_{\text{X}}^{c}\delta g_{bc}-u_{\text{X}b}\mathscr{L}_{\xi_{\text{X}}}u_{\text{X}}^{b}\right)\delta u_{\text{X}}^{a}\boldsymbol{P}_{a}^{\text{X}}+\mathscr{L}_{\xi_{\text{X}}}\delta u_{\text{X}}^{a}\boldsymbol{P}_{a}^{\text{X}} +\left(u_{\text{X}}^{b}u_{\text{X}}^{c}\delta g_{bc}-2u_{\text{X}b}\mathscr{L}_{\xi_{\text{X}}}u_{\text{X}}^{b}\right)\delta u_{\text{X}}^{a}\boldsymbol{P}_{a}^{\text{X}}+2\mathscr{L}_{\xi_{\text{X}}}u_{\text{X}}^{a}\delta\boldsymbol{P}_{a}^{\text{X}}\nonumber \\
= & \mathscr{L}_{\xi_{\text{X}}}\left(\frac{1}{2}u_{\text{X}}^{a}u_{\text{X}}^{b}u_{\text{X}}^{c}\delta g_{bc}-q_{\text{X}b}^{a}\mathscr{L}_{\xi_{\text{X}}}u_{\text{X}}^{b}\right)\boldsymbol{P}_{a}^{\text{X}}+\left(\frac{1}{2}u_{\text{X}}^{b}u_{\text{X}}^{c}\delta g_{bc}-u_{\text{X}b}\mathscr{L}_{\xi_{\text{X}}}u_{\text{X}}^{b}\right)\delta u_{\text{X}}^{a}\boldsymbol{P}_{a}^{\text{X}}+2\mathscr{L}_{\xi_{\text{X}}}u_{\text{X}}^{a}\delta\boldsymbol{P}_{a}^{\text{X}}\nonumber \\
= & \left(\frac{1}{2}u_{\text{X}}^{b}u_{\text{X}}^{c}\delta g_{bc}-u_{\text{X}b}\mathscr{L}_{\xi_{\text{X}}}u_{\text{X}}^{b}\right)\mathscr{L}_{\xi_{\text{X}}}u_{\text{X}}^{a}\boldsymbol{P}_{a}^{\text{X}}-\mathscr{L}_{\xi_{\text{X}}}\mathscr{L}_{\xi_{\text{X}}}u_{\text{X}}^{a}\boldsymbol{P}_{a}^{\text{X}}-\left(\frac{1}{2}u_{\text{X}}^{b}u_{\text{X}}^{c}\delta g_{bc}-u_{\text{X}b}\mathscr{L}_{\xi_{\text{X}}}u_{\text{X}}^{b}\right)\mathscr{L}_{\xi_{\text{X}}}u_{\text{X}}^{a}\boldsymbol{P}_{a}^{\text{X}}+2\mathscr{L}_{\xi_{\text{X}}}u_{\text{X}}^{a}\delta\boldsymbol{P}_{a}^{\text{X}}\nonumber \\
= & 2\mathscr{L}_{\xi_{\text{X}}}u_{\text{X}}^{a}\delta\boldsymbol{P}_{a}^{\text{X}}-\mathscr{L}_{\xi_{\text{X}}}\mathscr{L}_{\xi_{\text{X}}}u_{\text{X}}^{a}\boldsymbol{P}_{a}^{\text{X}}.
\end{align}
\end{widetext}
Substituting return to Eq. \eqref{omegaM-delta omega} gives that
\begin{align}
 & \vert v_{\text{X}}\vert u_{\text{X}}^{a}\delta^{2}\boldsymbol{P}_{a}^{\text{X}}-2\mathscr{L}_{\xi_{\text{X}}}\left(\vert v_{\text{X}}\vert u_{\text{X}}^{a}\right)\delta\boldsymbol{P}_{a}^{\text{X}}\nonumber\\
  &+\mathscr{L}_{\xi_{\text{X}}}\mathscr{L}_{\xi_{\text{X}}}\left(\vert v_{\text{X}}\vert u_{\text{X}}^{a}\right)\boldsymbol{P}_{a}^{\text{X}}\nonumber \\
= & \vert v_{\text{X}}\vert\left(2\mathscr{L}_{\xi_{\text{X}}}u_{\text{X}}^{a}\delta\boldsymbol{P}_{a}^{\text{X}}-\mathscr{L}_{\xi_{\text{X}}}\mathscr{L}_{\xi_{\text{X}}}u_{\text{X}}^{a}\boldsymbol{P}_{a}^{\text{X}}\right)\nonumber \\
 & -2\mathscr{L}_{\xi_{\text{X}}}\vert v_{\text{X}}\vert u_{\text{X}}^{a}\delta\boldsymbol{P}_{a}^{\text{X}}-2\vert v_{\text{X}}\vert\mathscr{L}_{\xi_{\text{X}}}u_{\text{X}}^{a}\delta\boldsymbol{P}_{a}^{\text{X}}\nonumber \\
 & +2\mathscr{L}_{\xi_{\text{X}}}\vert v_{\text{X}}\vert\mathscr{L}_{\xi_{\text{X}}}u_{\text{X}}^{a}\boldsymbol{P}_{a}^{\text{X}}+\vert v_{\text{X}}\vert\mathscr{L}_{\xi_{\text{X}}}\mathscr{L}_{\xi_{\text{X}}}u_{\text{X}}^{a}\boldsymbol{P}_{a}^{\text{X}}\nonumber \\
= & 0.
\end{align}
Thus we finally get
\begin{align}
 & \mathcal{E}\left(\delta\phi,\delta\phi\right)\nonumber \\
= & \delta^{2}M+\int_{\Sigma}\left[\boldsymbol{\omega}^{\left(\text{M}\right)}\left(\delta\phi,\mathscr{L}_{t}\delta\phi\right)-\delta\boldsymbol{\omega}^{\left(\text{M}\right)}\left(\delta\phi,\mathscr{L}_{t}\phi\right)\right]\nonumber \\
= & \delta^{2}M-\int_{\Sigma}\sum_{\text{X}}\Omega_{\text{X}}\varphi^{a}\delta^{2}\boldsymbol{P}_{a}^{\text{X}}\nonumber \\
= & \delta^{2}M-\sum_{\text{X}}\Omega_{\text{X}}\int_{\Sigma}\delta^{2}\left(\varphi^{a}\boldsymbol{P}_{a}^{\text{X}}\right)\nonumber \\
= & \delta^{2}M-\sum_{\text{X}}\Omega_{\text{X}}\int_{\Sigma}\delta^{2}\boldsymbol{J}^{\text{X}}\nonumber \\
= & \delta^{2}M-\sum_{\text{X}}\Omega_{\text{X}}\delta^{2}J^{\text{X}},\label{canonical energy for axisymmetric perturbation}
\end{align}
where the third equality follows that $\Omega_{\text{X}}$ is uniform throughout the star, and the fourth equality follows that $\Sigma$ is chosen to be axisymmetric and so that the pullback of $\varphi^{a}\boldsymbol{\epsilon}_{abcd}$ to $\Sigma$ vanishes. Compare Eqs. \eqref{thermodynamic stability criterion} and \eqref{canonical energy for axisymmetric perturbation}, we have shown that in the case of axisymmetric perturbations within the Lagrangian displacement framework such that $\delta J^{\text{X}}=0$, 
\begin{equation}
\mathcal{E}^{\prime}=\mathcal{E}\left(\delta\phi,\delta\phi\right),
\end{equation}
and the criterion of weak thermodynamic stability is given by following theorem:
\begin{thm}
\label{thm:weak thermodynamic stability}For a superconducting-superfluid star in weak thermodynamic equilibrium, a necessary condition for weak thermodynamic stability with respect to axisymmetric perturbations is positivity of $\mathcal{E}$ on all axisymmetric linearized solutions within the Lagrangian framework such that $\delta J^{\text{X}}=0$.
\end{thm}

As an application of our results, consider a star at $T=0$ for which the entropy per electron, $s$, takes it minimum value $s=0$ throughout the star, then any perturbation for which $\delta S=0$ must have $\delta s=0$ everywhere. Similar to the single perfect fluid case as shown in \cite{Friedman:1978wla}, in this isentropic case, with $\delta S=\delta N^{\text{X}}=0$, every perturbation can be described in the Lagrangian framework. So in this case, the word ``necessary'' in Theorem \ref{thm:weak thermodynamic stability} can be replaced by ``necessary and sufficient''. 

On the other hand, consider the radial perturbations of a static, spherically symmetric isentropic star. For this kind of perturbations, one clearly has $\delta J^{\text{X}}=0$. Let 
\begin{equation}
\eta_{\text{X}}^{a}=U^{\text{X}}\left(\frac{\partial}{\partial r}\right)^{a},
\end{equation}
be a spherically symmetric trivial displacement with $r$ is the radial coordinate in a spherical coordinate, then
\begin{align}
0 & =\mathscr{L}_{\eta_{\text{X}}}\boldsymbol{N}^{\text{X}}=d\left(U^{\text{X}}\iota_{\frac{\partial}{\partial r}}\boldsymbol{N}^{\text{X}}\right)\nonumber \\
 & =d\left(\sqrt{h}n_{\text{X}}U^{\text{X}}d\theta\wedge d\varphi\right)\nonumber \\
 & =\frac{\partial\left(\sqrt{h}n_{\text{X}}U^{\text{X}}\right)}{\partial r}dr\wedge d\theta\wedge d\varphi.
\end{align}
Subjecting to the boundary condition $U^{\text{X}}=0$ at $r=0$ in the spherical coordinate leads to $U^{\text{X}}=0$ throughout the star, so there is no spherically symmetric trivial displacement. It is immediately that in the case of radial perturbations, there is no restriction on $\mathcal{V}$, since $\mathcal{V}$ consists of all linearized solutions which is symplectically orthogonal to trivial displacements as mentioned in Sec. \ref{sec6}, i.e., $\mathcal{V}=\mathcal{C}$ contains all linearized solutions. Compare the Theorem \ref{thm:dynamic stability in axisymmetric case} with the Theorem \ref{thm:weak thermodynamic stability} (with the word ``necessary'' is replaced by ``necessary and sufficient''), we see that \emph{in the isentropic case, for radial perturbations of static, spherically symmetric stars, the weak thermodynamic stability is equivalent to the dynamic stability.}

We would like to end this section by introducing the definition of strong thermodynamic stability: a superconducting-superfluid star in strong thermodynamic equilibrium is said to be \emph{strongly thermodynamically stable} if $\delta^{2}S<0$ for all linearized solutions with 
\begin{equation}
\delta M=\delta N=\delta J=\delta^{2}M=\delta^{2}N=\delta^{2}J=0.\label{condition for strong thermodynamic stabiltiy}
\end{equation}
Since the star is in strong thermodynamic equilibrium means that there is no differential rotation and it is in chemical equilibrium in the sense of Theorem \ref{thm:strongly thermodynamic equilibrium}, so define
\begin{equation}
\mathcal{E}^{\prime\prime}\equiv\delta^{2}M-\tilde{\mu}\delta^{2}N-\tilde{T}\delta^{2}S-\Omega\delta^{2}J.
\end{equation}
A calculation parallel to the above calculation shows that for the axisymmetric perturbations within the Lagrangian displacement framework such that $\delta J=0$, we have
\begin{equation}
\mathcal{E}\left(\delta\phi,\delta\phi\right)=\delta^{2}M-\sum_{\text{X}}\Omega_{\text{X}}\delta^{2}J^{\text{X}}=\delta^{2}M-\Omega\delta^{2}J=\mathcal{E}^{\prime\prime},
\end{equation}
and the criterion of strong thermodynamic stability is given by following
theorem:
\begin{thm}
\label{thm:strong thermodynamic stability}For a superconducting-superfluid star in strong thermodynamic equilibrium, a necessary condition for strong thermodynamic stability with respect to axisymmetric perturbations is positivity of $\mathcal{E}$ on all axisymmetric linearized solutions within the Lagrangian framework such that $\delta J=0$.
\end{thm}

Similar as the discussion of strong thermodynamic equilibrium given at the end of Sec. \ref{subsec4.3}, the concept of strong thermodynamic stability should apply not only to the non-transfusive model which we mainly concerned in this paper, but also to the case including allowance of transfusion.

\section{Conclusion and discussion}

\label{sec8}

We have established the criterion for both dynamic and thermodynamic stability, as summarized into Theorems \ref{thm:dynamic stability criterion}, \ref{thm:dynamic stability in axisymmetric case}, \ref{thm:weak thermodynamic stability}, and \ref{thm:strong thermodynamic stability}. To this end, we first derived the necessary and sufficient condition for the thermodynamic equilibrium, identified the degeneracy of pre-symplectic form, and constructed the phase space. This framework allowed us to derive the canonical energy, which we then established as a stability criterion by considering physically stationary solutions. As a by-product, we also derive the eigenvalue equation Eq. \eqref{eigenvalue equation} for the speed of sound in our superconducting-superfluid star model.

The analysis and results in this work have broader applicability than the specific neutron star context discussed. On the one hand, although we restrict ourselves to a non-transfusive multi-constituent fluid model, for which both our weak and strong definitions for thermodynamic equilibrium and stability are applicable, the definition and criterion for strong thermodynamic equilibrium and stability remain valid even in the model including the allowance for transfusion. On the other hand, since we mainly concern the neutron stars, the three constituents are taken as the neutrons, protons, and electrons. As our calculations rely primarily on summation over the abstract chemical indices $\text{X}$ rather than the particular fluid indices $\text{n}$, $\text{p}$, or $\text{e}$, our results can be directly generalized to the star consisting of arbitrary number of fluids. Specifically, if one considers a multi-constituent fluid star, with the particle currents are given by $n_{\text{X}}^{a}$, the electric current is given by $j^{a}=\sum_{\text{X}}e^{\text{X}}n_{\text{X}}^{a}$, where now $\text{X}=\text{X}_{1},\text{X}_{2},\cdots,\text{X}_{k}$ with $k$ is an arbitrary positive integer due to how many kind of fluids in the system, and one constituent is assumed to be normal(i.e., carries the entropy per particle $s$), then all the analysis in our paper will hold by replacing $\text{X}=\text{n},\text{p},\text{e}$ with $\text{X}=\text{X}_{1},\text{X}_{2},\cdots,\text{X}_{k}$, and the main theorem for dynamic and thermodynamic stability will be given by the same statements. Moreover, since our results have no restriction on whether the superfluid neutrons and superconducting protons are irrotational, so they are apply to the case even there are vortex in the neutron stars.

There are several further directions to be investigated. First, although the vortex is allowed in our superconducting-superfluid stars, it would be beneficial to investigate the dynamic and thermodynamic stability directly from the model for dealing with the macroscopic effects of vorticity quantization \cite{Carter:1995mj,Carter:1998rn}. Second, in our analysis, we suppose that all neutrons and protons are in "super" states, however, it will be more realistic to consider the case that the neutrons and/or protons carry entropy due to there are mixtures of normal and super states in these constituents. In such a case, the constituents that should be treated independently in our system will be the superfluid neutrons, the normal fluid neutrons, the superconducting protons, the (charged) normal fluid protons, and the normal fluid electrons. Finding the stability criterion in this more complex case will be helpful. Furthermore, one could also explore whether the perturbative approach employed in \cite{Caballero:2024qtv} can be adapted to this superconducting-superfluid stellar model to yield an alternative stability criterion. Finally, it will be interesting that whether there is equivalence between the dynamic and thermodynamic stability for stars in AdS spacetimes, where the criterion for black branes in AdS spacetimes has been established in \cite{Hollands:2012sf}, or in other generalized gravitational theories. 

\begin{acknowledgments}
This work is partially supported by the National Key Research and Development Program of China with Grant No. 2021YFC2203001 as well as the National Natural Science Foundation of China with Grant Nos. 12035016, 12275350, 12375048, 12375058, 12361141825, 12447182, and 12575047.
\end{acknowledgments}

\bibliography{ref}

\end{document}